%% file: MAIN_CHC_C.tex
\documentclass{amsart}
\usepackage{amsmath,amssymb}
\usepackage{pstricks, pst-plot, pst-node, pst-grad, pst-math}
\usepackage{graphicx,xcolor}
\usepackage{arydshln}

\newtheorem{theorem}{Theorem}[section]
\newtheorem{lemma}[theorem]{Lemma}
\newtheorem{proposition}[theorem]{Proposition}

\theoremstyle{definition}

\theoremstyle{remark}
\newtheorem{remark}[theorem]{Remark}

\newcounter{num}
\renewcommand{\thenum}{\roman{num}}

\newtheorem{corollary}[theorem]{Corollary}

\numberwithin{equation}{section}

\newcommand\ka{\kappa}
\newcommand\la{\lambda}
\newcommand\al{\alpha}
\newcommand\de{\delta}
\newcommand\ga{\gamma}

\newcommand\CL{{\mathcal L}}
\newcommand\CH{{\mathcal H}}

\newcommand\CA{{\mathcal A}}

\newcommand\CK{{\mathcal K}}

\newcommand\CP{{\mathcal P}}

\newcommand\CS{{\mathcal S}}
\newcommand\CT{{\mathcal T}}
\renewcommand{\matrix}[2]{\left( \!\! \begin{array}{#1} #2 \end{array} \!\! \right)}
\newcommand{\ov}[1]{\overline{#1}}
\renewcommand\d{{\rm d}}

\newcommand{\dom}{\mathop{\rm dom}}
\newcommand{\Ran}{\mathop{\rm ran}}
\newcommand{\rank}{\mathop{\rm rank}}
\newcommand{\Ker}{\mathop{\rm ker}}

\newcommand{\N}{{\mathbb N}}

\newcommand{\R}{{\mathbb R}}
\newcommand{\C}{{\mathbb C}}
\newcommand{\Z}{{\mathbb Z}}
\newcommand{\T}{{\mathbb{T}}}

\newcommand\eps{\varepsilon}

\newcommand\wt{\widetilde}
\newcommand\wh{\widehat}

\newcommand\ontops[2]{\genfrac{}{}{0pt}{3}{#1}{#2}}

\newcounter{marke}
\newcommand{\bl}{\begin{list}{\roman{marke})}{\usecounter{marke}
\topsep 0 cm \itemsep 0cm}}
\newcommand{\el}{\end{list}}

\newcommand\be{\begin{equation}}
\newcommand\ee{\end{equation}}

\providecommand*{\mrm}[1]{\mathrm{#1}}
\newcommand{\iu}{\mrm{i}}
\newcommand{\eu}{\mrm{e}}
\newcommand{\Hone}{H^{1}(\mathbb{T}^{2})}
\newcommand{\Htwo}{H^{2}(\mathbb{T}^{2})}

\newcommand{\Ltwo}{L^{2}(\mathbb{T}^{2})}

\newcommand{\Hi}{\CH}
\newcommand{\Hii}{\wh\CH}
\renewcommand{\H}{\wt\CH}

\newcommand\ct[1]{\textcolor{red}{#1}}

\newcommand\cg[1]{\textcolor{green}{#1}}

\newcommand{\Si}{\CS}
\newcommand{\Sii}{\wh\CS}

\newcommand{\kam}{\kappa_1}
\newcommand{\kap}{\kappa_2}

\newcommand{\lanm}{\la_{1,j}}
\newcommand{\lanp}{\la_{2,j}}

\newcommand\PA{P}
\newcommand\PB{Q}

\newcommand{\lanmU}{\la_{1,j}^{U}}
\newcommand{\lanpU}{\la_{2,j}^{U}}
\newcommand{\lanmL}{\la_{1,j}^{L}}
\newcommand{\lanpL}{\la_{2,j}^{L}}

\newcommand{\TlanmU}{\wt{\la}_{1,j}^{U}}

\newcommand{\aaa}{a}
\newcommand{\bb}{b}
\newcommand{\cc}{c}

\begin{document}

\title[Non-linear eigenvalue problems and applications to photonic crystals \ ]{Non-linear eigenvalue problems and applications to photonic crystals}

\author{Christian Engstr\"om} 
\address{Department of Mathematics and Mathematical Statistics, Ume\aa \ University, SE-901 87 Ume\aa, Sweden}
\email{christian.engstrom@math.umu.se}

\author{Heinz Langer}
\address{Institute for Analysis and Scientific Computing, Vienna University of Technology, 
Wiedner Hauptstr.~8--10, 1040 Vienna, Austria}
\email{hlanger@email.tuwien.ac.at}

\author{Christiane Tretter}
\address{Mathematisches Institut, Universit\"at Bern, Sidlerstr.\ 5, 3012 Bern, Switzerland 
\ {\rm\&} \ Matematiska institutionen, Stockholms universitet, SE-106 91 Stockholm, Sweden}
\email{tretter@math.unibe.ch}

\date{\today}

\begin{abstract}
We establish new analytic and numerical results on a general class of rational operator Nevanlinna functions that arise e.g.\ in modelling
photonic crystals. 
The capability of these dielectric nano-structured materials to control the flow of light depends on specific features of their eigenvalues.
Our results provide a complete spectral analysis including variational principles and two-sided estimates for all eigenvalues along with numerical implementations.
They even apply to multi-pole Lorentz models of permittivity functions and to the eigenvalues 
between the poles where classical min-max variational principles fail completely. 
In particular, we show that our abstract two-sided eigenvalue estimates are optimal and we derive explicit bounds on the band gap above a Lorentz pole.  
A high order finite element method is used to compute the two-sided estimates of a selection of eigenvalues for several concrete Lorentz models, e.g.\ polaritonic materials and multi-pole models.
\vspace{-6mm}
\end{abstract}

\maketitle

\input{intro_C.tex}

\input{Part_General_Lemmas_CTHL-hc-2C.tex}
\input{Part_minmax_CT.tex}
\input{Part_H_add_C.tex}
\input{Part_Photonic_Crystals_v2-ccc.tex}

\input{Part_Galerkin-cc.tex}

\vspace{2.5mm}

{\small
{\bf Acknowledgements.} \ 
The authors gratefully acknowledge the support of the Swedish Research Council under Grant No.\ $621$-$2012$-$3863$. 
H.\ Langer and C.\ Tretter thank the Institutionen f\"or matematik och matematisk statistik 
at Ume\aa \ universitet very much for the kind hospitality.
C.\ Tretter also gratefully acknowledges a guest professorship of the \emph{Knut och Alice Wallen\-bergs Stiftelse} at Stockholms universitet where this work was completed.
}


\bibliographystyle{alpha}
\bibliography{Bib-CH-CE2}


\end{document}

%% file: intro_C.tex
\section{Introduction}

Many physical systems are passive in the sense that they do not produce energy. 
For example, in linear electromagnetic field theory energy can be transferred from the electromagnetic field into the material, but not from the material into the electromagnetic field.
Moreover, materials are in general dispersive which, when frequency is the spectral parameter, results in a non-linear operator function. Therefore a large number of systems are accurately described by Nevanlinna functions (also called Herglotz functions) whose values are differential operators.
In numerical analysis matrix-valued Nevanlinna functions such as Schur complements are known as versatile tools \cite{MR1866366}.
In recent years operator-valued Nevanlinna functions and their analytical properties have been studied intensively \cite[Introduction]{book}. However, we still lack a more detailed understanding of a class of rational Nevanlinna operator functions that are sufficiently general to cover several important  physical applications, such as acoustic and electromagnetic problems with frequency dependent materials e.g.\ in photonic crystals. Mathematical research in this direction has started, but is still in its infancy \cite{MR2718134}, \cite{MR2876569},~\cite{Schmalkoke}.

Here we establish, under new and very general conditions, a complete picture of the spectral properties for such rational operator functions. 
The novelty of our approach is that it applies, in the case of several poles, to the eigenvalues 
between the poles where classical min-max variational principles fail completely. As a result, in applications to photonics, we cover   
piecewise constant multi-pole Lorentz models \cite{KSSKB15}, \vspace{-2.75mm} \cite{PhysRevB.83.205131}
\begin{equation}
\label{veryverylast}
 \epsilon(\cdot,\omega) = \sum_{m=1}^M \epsilon_m(\omega) \chi_{\Omega_m}(\cdot), \quad \epsilon_m(\omega) = \epsilon_{m,\infty}  + \epsilon_{m,\infty} \sum_{\ell=1}^{L_m} \frac{\omega_{p,m,\ell}^2}{\omega_{0,m,\ell}^2 - \omega^2},
\vspace{-1.5mm}
\end{equation}
periodic on some bounded domain $\Omega=\Omega_1 \dot\cup \dots \dot\cup\,\Omega_M$, as well as permittivity functions $ \epsilon(\cdot,\omega)$ where the linear part of the corresponding operator function has eigenvalues $\lambda:=\omega^2$ below the Lorentz poles $\omega_{p,m,\ell}^2$.
The abstract operator functions we consider have the Nevanlinna property in the sense that they are analytic on the complex plane, have self-adjoint values on the real axis, and a finite number of poles which are real and of first order with non-positive residues, and their derivatives are non-negative between the poles.
This property enables us to introduce generalized Rayleigh functionals,
establish variational principles, and derive two-sided estimates for \emph{all} eigenvalues of this important class of rational operator~functions. 

We demonstrate the efficacy of the new theory for unbounded operator functions modelling photonic crystals. These dielectric nano-structured materials which are used to control and manipulate the flow of light \cite{JJWM2008} 
are commonly 
modelled by periodic Lorentz permittivity functions \eqref{veryverylast} with several rational terms. 
Explicit computations show that the abstract two-sided eigenvalue estimates are optimal and we derive explicit bounds on the band gap above a Lorentz pole. The operator function is discretised with a high order finite element method and several concrete examples e.g.\ for polaritonic materials illustrate the general theory. In particular, we compute the two-sided estimates of a selection of eigenvalues and we illustrate the 
accumulation of eigenvalues at the poles and the corresponding singular sequence.  In most examples a continuous finite element method is used to compute the eigenvalues, but in cases were a block diagonal mass matrix is an advantage a discontinuous Galerkin method is employed.

The paper is organized as follows. In Section \ref{sec:1} we set up the required operator theoretic framework. In Section \ref{sec:2} we consider the one pole case and establish min-max variational characterizations and two-sided estimates for all eigenvalues. In Sec\-tion \ref{sec:4} we generalize the min-max principles to the multi-pole case and identify cases where a band gap occurs. In Section \ref{Application} we apply our abstract results to pho\-tonic crystals with multi-pole Lorentz models \eqref{veryverylast}. Section \ref{sec:Galerkin} contains the~nume- rical finite element analysis for several material models, illustrating different features of the abstract results such as the occurrence of an index shift or band~gaps. \hspace{2mm}

Throughout this paper we use the following notations and conventions. All Hil\-bert spaces are separable. For a closed linear operator $T$ in a Hil\-bert space $\CH$ we denote by $\ker T$, $\Ran T$, $\rho(T)$, $\sigma(T)$, and $\sigma_{\rm p}(T)$ its kernel, range, resolvent set, spectrum, and point spectrum, respectively; the essential spectrum of $T$ is defined as  $\sigma_{\rm ess}(T)\!:=\!\{\la\in\C: T \!-\! \la \mbox{ is not Fredholm}\}$. If $T$ is self-adjoint, then $\la \!\in\! \sigma_{\rm ess}(T)$ iff $\la \!\in\! \sigma(T)$ and $\la$ is \emph{not} an isolated eigenvalue of finite multiplicity. Further, 
for a Borel set $I \!\subset\! \R$, we  denote by $\CL_I(T)$ the spectral subspace of $T$ cor\-responding to the set $I$ and, if $T$ is bounded from below and $\mu \!<\! \min \sigma_{\rm ess}(T)$, by $N(T,\mu) := \dim \CL_{(-\infty,\mu]}(T)$  the number of eigenvalues of $T$  that are $\le \mu$ counted with 
multiplicities.

%% file: Part_General_Lemmas_CTHL-hc-2C.tex
\section{\bf Operator theoretic framework}
\label{sec:1}

Recent applications in nanophotonics require to study the spectral properties of operator functions that depend rationally on the spectral parameter.
From the analytical point of view, there are two possible approaches and often a combination of both is most advantageous.

The first one is to analyze the operator function directly using properties of its operator coefficients. Here we consider analytic operator functions $\Si$ 
whose values are linear operators in a Hilbert space $\Hi$ and which are given by
\begin{equation}
\label{schur1}
  \Si(\la) = A-\la - B(C-\la)^{-1}B^*, \quad \dom \Si(\la) = \dom A \subset \Hi, \quad \la \in \C \setminus \sigma(C),
\end{equation}
where $A$ and $C$ are linear operators in Hilbert spaces $\Hi$ and $\Hii$, respectively, $B$ acts from $\Hii$ to $\Hi$, and $B$, $C$ are bounded.

The second approach is to consider a linearization of the operator function $\Si$, i.e.\ a linear operator $\CA$ in a larger Hilbert space that reflects all the spectral properties of~$\Si$. A particular linearization of $\Si$ in \eqref{schur1} is the block operator matrix $\CA$ in the product Hilbert space $\H=\Hi \oplus \Hii$ given by
\begin{equation}\label{eq:A-gen}
  \CA = \matrix{cc}{A & B \\ B^* & C }, \quad 
  \dom \CA= \dom A \oplus \Hii.
\end{equation}
In fact, $\Si$ is the first Schur complement of $\CA$; the relations between spectral properties of $\Si$ and $\CA$ summarized in the next proposition are well-known and not difficult to check (see e.g.\ \cite[Section~2.3]{book}). Recall \vspace{-0.15mm} that the spectrum of $\Si$ is defined as
\[
 \sigma(\Si):= \{ \la \in \C \setminus \sigma(C): 0 \in \sigma(\Si(\la)) \},
\vspace{-0.15mm} 
\]
and analogously for the point spectrum $ \sigma_{\rm p}(\Si)$ and essential spectrum $\sigma_{\rm ess}(\Si)$ of $\Si$.
 
\begin{proposition}
\label{prop:2.1}
Let $\H$ be a Hilbert space, $\H=\Hi\oplus\Hii$ with Hilbert spaces $\Hi$,~$\Hii$.
Let $A$ be a closed linear operator in $\Hi$ and let $B: \Hii \to \Hi$, $C: \Hii \to \Hii$ be bounded linear operators. Then,
for the block operator matrix $\CA$ in \eqref{eq:A-gen} and the operator function $\Si$ in \eqref{schur1},
\begin{itemize}
\item[{\rm i)}] $\sigma(\CA) \setminus \sigma(C) = \sigma(\Si)$, $\sigma_{\rm p}(\CA) \setminus \sigma(C) = \sigma_{\rm p}(\Si)$, 
and $\sigma_{\rm ess}(\CA) \setminus \sigma(C) = \sigma_{\rm ess}(\Si)$;
\item[{\rm ii)}] if $ \la_0 \in \sigma_{\rm p}(\CA) \setminus \sigma(C)$ with eigenvector $(u_0\ \wh u_0)^{\rm t}$, then 
$u_0$ is an eigenvector of $\,\Si$ at \vspace{-2mm} $\la_0$;  
\item[{\rm iii)}] if $\,\la_0 \!\in\! \sigma_{\rm p}(\Si)$ with eigenvector $u_0$, then $\left(\!\!\!\begin{array}{c}u_0\\ -(C\!-\!\la_0)^{-1} B^* u_0\end{array}\!\!\!\right)$ 
\vspace{-2mm} is an eigenvector of $\CA$ at $\la_0$.
\end{itemize}
\end{proposition}

\begin{proof}
\!All claims follow from the 
Schur-Frobenius factorization (see \vspace{-1mm} e.g.~\cite[(2.2.12)]{book}), 
\[
\CA-\la = \matrix{cc}{I & B(C-\la)^{-1} \\ 0 & I} \matrix{cc}{ \Si(\la) & 0 \\ 0 & C-\la} \matrix{cc}{I & 0 \\ (C-\la)^{-1}B^* & I}, \quad \la \in \C \setminus \sigma(C), 
\]
since the outer two factors are bounded and boundedly invertible and so is $C-\la$ if $\la \notin \sigma(C)$ 
(comp.\ \cite[Theorem~2.3.3 ii)]{book});
note that closures in \cite[(2.2.12)]{book} may be omitted here since $B$ and $C$ are bounded.
\end{proof}  
 
The claims of the following proposition were proved in various degrees of generality (see \cite{MR1285306}, \cite{MR1354980}, \cite{MR1389460}, and also \cite{book}); for the convenience of the reader we give a simple proof. 
%

\begin{proposition}
\label{prop:1.1}
Let $\H$ be a Hilbert space, $\H=\Hi\oplus\Hii$ with Hilbert spaces $\Hi$,~$\Hii$. 
Let $A$ be a closed linear operator in $\Hi$ with compact resolvent, and let 
$B: \Hii \to \Hi$, $C: \Hii \to \Hii$ be bounded linear operators.   
Then
\begin{enumerate}
\item[{\rm i)}] the essential spectrum of the block operator matrix $\CA$ in \eqref{eq:A-gen} is given by
\begin{align}
\label{eq:sess1}
  \sigma_{\rm ess}(\CA)=\sigma_{\rm ess}(C).
\end{align}
\end{enumerate}
If $A$ and $C$ are self-adjoint, then so is $\CA$; if, in addition, $A$ is bounded from below, then so is $\CA$ and
\begin{enumerate}
\item[{\rm ii)}] if $\,\max \sigma(C)< \min\sigma(A)$, then $\big(\max \sigma(C),\min\sigma(A)\big) \subset \rho(\CA)$;
\item[{\rm iii)}] if \,{\rm dim} $\Hi = \infty$, then $\CA$ has a sequence of eigenvalues of finite multiplicities accumulating only at $+\infty$;
\item[{\rm iv)}] 
if $\,c\ge \sup \sigma_{\rm ess}(C)$ and
$\big( c, c+\delta \big)\!\subset\!\rho(C)$ for some $\delta\!>\!0$, then   
$\big(c,c+\varepsilon \big)\subset\rho(\CA)$ for some~$\varepsilon >0$.
\end{enumerate}
\end{proposition}

\begin{proof}
The block operator matrix $\CA$ is a bounded symmetric perturbation of the block diagonal operator matrix diag\,$(A, C)$ in $\Hi\oplus\Hii$. Hence $\CA$ is closed since so are $A$ and $C$, $\CA$~is self-adjoint if so are $A$ and $C$, and $\CA$ is semi-bounded if so is $A$. 

i) To prove \eqref{eq:sess1}, 
we consider the second Schur complement $\Sii$ of $\CA$ given by
\[
    \Sii(\la) = C -\la - B^*(A - \la)^{-1} B, \quad \dom \Sii(\la) = \Hii, \quad \la \in \C \setminus \sigma(A).
\] 
Since the two outer factors in the corresponding Frobenius-Schur factorization (see e.g.\ \cite[(2.2.11)]{book})
\[
\CA - \la = \matrix{cc}{I & 0 \\ B^*(A-\la)^{-1} & I } \matrix{cc}{A-\la & 0 \\ 0 & \Sii(\la)} \matrix{cc}{I & (A-\la)^{-1}B \\ 0 & I}, \quad \la \in\C\setminus\sigma(A),
\]
are bounded and boundedly invertible, it follows that $\sigma_{\rm ess}(\CA) \setminus \sigma(A) = \sigma_{\rm ess}(\Sii)$ (see \cite[Theorem 2.4.7]{book}). 
Since $B$ is bounded and $A$ has compact resolvent, $B^*(A - \la)^{-1} B$ is compact in $\Hii$ and hence $\sigma_{\rm ess}(\Sii) = \sigma_{\rm ess}(C) \setminus \sigma(A)$.
Altogether we obtain that, for $\la \notin \sigma(A)$, 
\begin{equation}
\label{eq:essS2}
  \la \in \sigma_{\rm ess}(\CA) \iff \la \in \sigma_{\rm ess}(\Sii) \iff \la \in \sigma_{\rm ess}(C).
\end{equation} 

It remains to be proved that $\la \in \sigma_{\rm ess}(\CA) \iff \la \in \sigma_{\rm ess}(C)$ also for $\la\in\sigma(A)$. 
In this case, since $A$ has compact resolvent, $\la$ is an isolated eigenvalue of $A$ and the algebraic eigenspace $\CL_\la(A)$ of $A$ at $\la$ is finite dimensional. If $P_\la$ is the orthogonal 
projection in $\Hi$ onto $\CL_\la(A)$ and we choose $\mu\in\R \setminus \sigma(A)$, then 
$A_0:=A- (\lambda -\mu) P_\la$  is a finite dimensional perturbation of $A$ with $\la \notin\sigma(A_0)$. 
Then the equivalence \eqref{eq:essS2} applies to the block operator matrix
\[
  \CA_0 := \matrix{cc}{A_0 & B \\ B^* & C },  \quad \dom \CA= \dom A \oplus \Hii,
\]
and yields that $\la \in \sigma_{\rm ess}(\CA_0) \iff \la \in \sigma_{\rm ess}(C)$. Since
$\CA_0$ is a finite dimensional perturbation of $\CA$, we have $\sigma_{\rm ess}(\CA_0)= \sigma_{\rm ess}(\CA)$, which completes the proof of \eqref{eq:sess1}.


Now suppose that $A$ and $C$ are self-adjoint, and hence so is $\CA$.

ii) The claim is an immediate consequence of \cite[Thm.\ 2.1]{MR1354980} (see also \cite[Theorem 5.2, Corollary~5.3]{CT09}).

iii) Since $A$ is semi-bounded with compact resolvent in $\Hi$ and \,{\rm dim} $\Hi = \infty$, the spectrum of the block diagonal operator matrix diag\,$(A, C)$ in $\Hi\oplus\Hii$ 
consists of the sequence $(\nu_j(A))_{j=1}^\infty$ of eigenvalues of $A$ of finite multiplicities accumulating only at~$+\infty$ and of the spectrum of~$C$. 
As $\CA$ is a bounded perturbation of diag\,$(A, C)$, it will also have a sequence of eigenvalues of finite multiplicities accumulating only at~$+\infty$ (see e.g.\ \cite[Section V.4.3]{MR1335452}). 

iv) 
By the assumptions on $A$, $C$, and $c$, we know that
$A$ has finitely many eigenvalues $\!\le\!c$ and $C$ has finitely many eigenvalues~$\!>\!c$ (counted with their finite multiplicities),~say, 
\[
\nu_1(A) \le \nu_2(A) \le \dots \le \nu_N(A) \le c < \nu_1(C) \le \nu_2(C) \le \dots \le \nu_M(C).
\]
If $P_j$ and $Q_k$ are the orthogonal projections in $\Hi$ onto $\ker(A-\nu_j(A))$ and in $\Hii$ onto $\ker(C-\nu_k(C))$, respectively, and $\nu_{N+1}(A)$ is the smallest eigenvalue of $A$ that is $>c$, then 
\[
A_1:=A-\sum_{j=1}^N \big(\nu_j(A)-\nu_{N+1}(A)\big) P_j, \quad
C_1:=C-\sum_{k=1}^M \big(\nu_k(C)-c \big) Q_k
\]
are finite dimensional perturbations of $A$ and $C$ with the property that 
\[
\max \sigma(C_1) = c < \nu_{N+1}(A) = \min \sigma(A_1).
\]
Now claim ii) applied to the block operator matrix
\[
\CA_1 := \matrix{cc}{A_1 & B \\ B^* & C_1 }, \quad \dom \CA = \dom A \oplus \Hii,
\]
yields that $\big(c,\nu_{N+1}(A)\big) \subset \rho(\CA_1)$. 
Since the self-adjoint operator $\CA$ is a finite dimensional perturbation of the self-adjoint operator $\CA_1$, 
it follows that $\CA$ has only a finite number of eigenvalues in the interval $\big(c,\nu_{N+1}(A)\big)$ counted with multiplicities
(see \cite[Theorem~9.3.3, p.~215]{MR1192782}). Hence $(c,c+\varepsilon)\subset\rho(\CA)$ for some $\varepsilon>0$.
\end{proof}

\begin{remark}
\label{rem:1.1}
In fact, the total multiplicity of the eigenvalues of $\CA$ in the interval $\big(c,\nu_{N+1}(A)\big)$ is at most
$N+M$, i.e.\ at most the sum of the multiplicities of all eigenvalues of $A$ less than or equal to~$c$ and of those of $C$ greater than $c$;
this follows from \cite[Theorem~9.3.3, p.~215]{MR1192782}.
\end{remark}

Since the diagonal part diag\,$(A, C)$ of $\CA$ in $\Hi\oplus\Hii$ is bounded from below and the off-diagonal part is bounded with norm $\|B\|$,
the block operator matrix $\CA$ is bounded from below with
\[ 
\min \!\big\{\! \min \sigma(A), \min \sigma(C) \!\big\}\! -\! \|B\| 
\le \min\sigma(\CA) \le 
\min \!\big\{\! \min \sigma(A), \min \sigma(C) \!\big\}\! +\! \|B\| 
\]
(see \cite[Theorem~V.4.11]{MR1335452} and e.g.\ \cite[Lemma~5.2]{CT13}). 
The off-diagonal nature of the perturbation allows one to strengthen this result (see  \cite{MR2247883} for the case of bounded $\CA$); 
the following proposition is immediate from \cite[Theorem~5.6~(i)]{CT09}.


\begin{proposition}
\label{stillill}
The block operator matrix $\CA$ in \eqref{eq:A-gen} is bounded from below with
\[
  \hspace{10mm}
  \min\bigl\{\min \sigma(A), \min \sigma(C) \bigr\} - \de \le \min \sigma(\CA) \le \min\bigl\{\min \sigma(A), \min \sigma(C)\bigr\} 
\]
with
\[
  \de := \|B\| \tan \left( \frac 12 \arctan \frac{2\,\|B\|}{|\min \sigma(A) - \min \sigma(C)|} \right);
\]
here $\de  \le \|B\|$, and $\de < \|B\|$ if and only if $\,\min \sigma(C) \ne \min \sigma(A)$.
\end{proposition}


In the next two sections, we will study the spectral properties of the block operator matrix $\CA$ in \eqref{eq:A-gen} and its Schur complement $\Si$ given by \eqref{schur1} in greater detail. Here we distinguish the cases that $\CS$ has one pole and more than one pole.

%% file: Part_minmax_CT.tex
\section{\bf One pole case: variational principles and eigenvalue estimates}
\label{sec:2}

In this section we consider the case that $C=c I_{\Hii}$, i.e.\ the operator function $\Si$ in the Hilbert space $\Hi$
in \eqref{schur1} has a single pole at the point $c\in\R$, 
\begin{equation}
\label{opfS}
 \Si(\la) = A-\la - \frac{BB^*}{c-\la}, \quad \dom \Si(\la) = \dom A \subset \Hi, \quad \la \in \C\setminus\{c\},
\end{equation}
where $A$ is a self-adjoint operator with compact resolvent and bounded from below and where $B$ is a non-zero bounded linear operator.

\subsection{Spectral properties of $\CS$.}
\label{subsec:2.1}
To investigate the spectrum of $\Si$, we use its close relation to the spectrum 
of the block operator matrix $\CA$ in \eqref{eq:A-gen} given by
\begin{equation}
\label{eq:A}
  \CA = \matrix{cc}{A & B \\ B^* & c I_{\Hii}}, \quad 
  \dom \CA= \dom A \oplus \Hii \subset \Hi \oplus \Hii.
\end{equation}
%

Since $\CA$ is self-adjoint and bounded from below, the spectra $\sigma(\CA)$ and $\sigma(\CS)=\sigma(\CA) \setminus \{c\}$ are real and bounded from below.
By Propositions \ref{prop:2.1} and \ref{prop:1.1}, we have
\begin{align*}
 &\sigma_{\rm ess}(\CA) = \sigma_{\rm ess}(c I_{\Hii}) = \begin{cases} \{c\} & \mbox{if } \dim \Hii = \infty, \\ \ \,\emptyset & \mbox{if } \dim \Hii < \infty, \end{cases} \\
 &\sigma_{\rm ess}(\CS) =  \sigma_{\rm ess}(\CA) \setminus \sigma(c I_{\Hii}) = \emptyset,
\end{align*}
and $(c,c+\eps)\cap \sigma(\CA) = (c,c+\eps)\cap \sigma(\CS) = \emptyset$ for some $\eps>0$. 

Hence the spectrum of $\CA$ and of $\CS$ in the intervals $(-\infty,c)$ and $(c,+\infty)$ is discrete and accumulates at most at the right end-points $c$ and $+\infty$, respectively. We denote the corresponding sequences of eigenvalues, ordered non-decreasingly and counted with multiplicities, 
by $(\lanm)_{j=1}^{n_1} \subset (-\infty,c)$ and $(\lanp)_{j=1}^{n_2} \subset (c,+\infty)$, respectively, with $n_1$, $n_2 \in \N_0 \cup \{\infty\}$:
\begin{align*}
  &\sigma(\CA) \setminus \{c\} = \sigma_{\rm p}(\CA) \setminus \{c\} =
  \sigma(\CS) = \sigma_{\rm p} (\CS)   = (\lanm)_{j=1}^{n_1} \,\dot\cup\, (\lanp)_{j=1}^{n_2}.
\end{align*}
Here $n_1=\infty$ means that the sequence $(\lanm)_{j=1}^{n_1} \subset [\min \sigma(\CA),c)$ is infinite and  accumulates at $c$, while
$n_2=\infty$ means that the sequence $(\lanp)_{j=1}^{n_2} \subset [c\!+\!\eps,+\infty)$ is infinite and accumulates at~$+ \infty$.
Since $A$ has compact resolvent, we have
\begin{equation}
\label{inf-or-not}
  n_1 < \infty \mbox { if } \dim \Hii < \infty, \quad n_2 < \infty \ \mbox{ if and only if } \dim \Hi < \infty.
\end{equation}
Indeed, if $\dim \Hii < \infty$, then $\sigma_{\rm ess}(\CA) = \sigma_{\rm ess}(C)=\emptyset$ by Proposition \ref{prop:1.1};
the second claim follows because $A$, and hence $\CA$, is bounded if and only if $\dim \Hi < \infty$.

\vspace{-1.5mm}

\begin{center}
\begin{figure}[h]
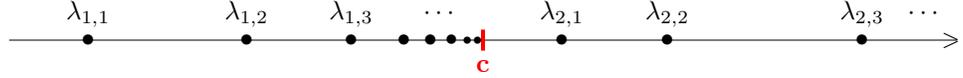

\psset{unit=7mm}\pspicture*(-9,-0.6)(9,1)
\psline[linewidth=.3pt,linecolor=black]{-}(-9,0)(0,0)
\psline[linewidth=.3pt,linecolor=black]{-}(0,0)(9,0)
\rput[C](8.9,0){$>$}
%
\psline[linewidth=1.5pt,linecolor=red]{-}(0,-0.2)(0,0.2)
\rput[C](0,-0.5){\ct{\boldmath{$\bf c$}}}
\rput[C](-7.5,0){$\bullet$}
\rput[C](-7.5,0.5){$\la_{1,1}$}
\rput[C](-4.49,0){$\bullet$}
\rput[C](-4.49,0.5){$\la_{1,2}$}
\rput[C](-2.5,0){$\bullet$}
\rput[C](-2.5,0.5){$\la_{1,3}$}
\rput[C](-1.5,0){$\bullet$}
\rput[C](-1,0){$\bullet$}
\rput[C](-0.6,0){\small$\bullet$}
\rput[C](-0.3,0){\tiny$\bullet$}
\rput[C](-0.1,0){\tiny$\bullet$}
\rput[C](-0.8,0.5){$\cdots$}
\rput[C](1.5,0){$\bullet$}
\rput[C](1.5,0.5){$\la_{2,1}$}
\rput[C](3.5,0){$\bullet$}
\rput[C](3.5,0.5){$\la_{2,2}$}
\rput(7.2,0){$\bullet$}
\rput[C](7.2,0.5){$\la_{2,3}$}
\rput[C](8.4,0.5){$\cdots$}
\endpspicture
\vspace{-3.5mm}
\caption{{\small Eigenvalues $(\la_{1,j})_{j=1}^{n_1} \!\subset\! (-\infty,c)$ and $(\la_{2,j})_{j=1}^{n_2} \!\subset\! (c,+\infty)$ of $\CS$.}}
\end{figure}
\vspace{-3mm}
\end{center}

The condition $\dim \Hii = \infty$ is only necessary for eigenvalue accumulation at $c$, i.e.\ for $n_1=\infty$. 
The following two propositions show that $\dim \Hii = \infty$ and $B$ injective are sufficient conditions for $n_1=\infty$.
Later, with more advanced tools, we will be able to show that the weaker condition $\dim \ov{\Ran B} = \infty$ is both 
necessary and sufficient (see Theorem~\ref{thm:2.3a}~iii)).


\begin{proposition}
\label{inf-at-c}
Let $\,\H$ be a Hilbert space, $\H=\Hi\oplus\Hii$ with Hilbert spaces $\Hi$,~$\Hii$. 
Let $A$ be a self-adjoint operator in $\Hi$ with compact resolvent and bounded from below, 
$B: \Hii \to \Hi$  a non-zero bounded linear operator, and $c\in\R$.

If $\,\dim \Hii = \infty$ and $c$ is not an eigenvalue if infinite multiplicity of $\CA$, 
then $c$ is an accumulation point of eigenvalues 
of $\CA$, and hence of $\CS$, from the left but not from the right.
\end{proposition}

\begin{proof}
If ${\rm dim}\,\Hii = \infty$, Proposition \ref{prop:1.1} shows that $\sigma_{\rm ess}(\CA) = \{c\}$. 
Since $\CA$ is self-adjoint, this implies that $c$ is either an eigenvalue of infinite multiplicity or an accumulation point of eigenvalues. 
The former is excluded by assumption and $(c,c+\varepsilon) \subset \rho(\CA)$ for some $\varepsilon >0$ by Proposition \ref{prop:1.1} iv).
Thus $c$ is an accumulation point of eigenvalues of $\CA$, and hence of $\CS$ by Proposition \ref{prop:2.1} ii), from the left but not from the right.
\end{proof}

\begin{proposition}
\label{prop:1.3a}
Let $\,\H$ be a Hilbert space, $\H=\Hi\oplus\Hii$ with Hilbert spaces $\Hi$,~$\Hii$. 
Let $A$ be a self-adjoint operator in $\Hi$ with compact resolvent and bounded from below, 
$B: \Hii \to \Hi$  a non-zero bounded linear operator, 
and $c\in\R$. 

If $B$ is injective, then $\dim \ker (\CA-c) \le N(A,c)$; in particular, $c$ is not an eigenvalue of infinite multiplicity of $\CA$.
\end{proposition}

\begin{proof}
First we show \vspace{-1mm} that 
\begin{equation}
\label{c-var}
 \wt u\!=\!(u \ \wh u)^{\rm t} \in \ker (\CA- c)\setminus \{0\} \ \implies \ u \ne 0, \quad c= \frac{(Au,u)}{\|u\|^2}.
\vspace{-1mm} 
\end{equation}
In fact, $(\CA-c) \wt u = 0$ is equivalent to
\begin{align}
\label{kerc1i}
  &(A-c)u +B \wh u =0, \\
\label{kerc2i}  
  &B^* u =0.
\end{align}
If $u=0$, then \eqref{kerc1i} implies $B \wh u=0$ and thus $\wh u=0$ since $B$ is injective by assumption, a contradiction to $\wt u\ne 0$; hence $u\ne 0$.
Taking the scalar product with $u$ in \eqref{kerc1i} and using \eqref{kerc2i}, we find that 
\[
0=((A-c)u,u) + (B \wh u,u) = ((A-c)u,u) + (\wh u,B^*u)= ((A-c)u,u), 
\]
which completes the proof of \eqref{c-var}. 
Next we prove that, if  $P: \Hi \oplus \Hii \to \Hi$ denotes the projection onto the first component, then
\begin{equation}
\label{nc}
n_c := \dim {\rm span} \big( P \ker (\CA-c) \big) %
\le N(A,c).
\end{equation}
Assume to the contrary that $n_c \ge N(A,c)+1$. We can characterize the eigenvalues $\nu_1(A) \le \nu_2(A) \le \dots \le \nu_{N(A,c)}(A) \le c < \nu_{N(A,c)+1}(A) \le \dots$
of $A$ counted with multiplicities by the classical min-max variational principle (see e.g.\ \cite[Theorem XIII.1]{MR0493421}).
If we use $P \ker (\CA-c) \subset \dom A$ and \eqref{c-var}, we obtain the estimate
\begin{align*}
  \nu_{j}(A) &= \min_{\ontops{\CL \subset \dom A}{\dim \CL = j}} \,\max_{\ontops{u \in \CL}{u\ne 0}} \frac{(Au,u)}{\|u\|^2} 
            \le  \min_{\ontops{\CL \subset P \ker (\CA-c)}{\dim \CL = j}} \!\!\max_{\ontops{u \in \CL}{u\ne 0}} \frac{(Au,u)}{\|u\|^2}
            = c, \quad j=1,2,\dots,n_c.
\end{align*}
Since $N(A,c)+1 \le n_c$ by assumption, it follows that $\nu_{N(A,c)+1}(A) \le c$, a contradiction.

Since $A$ has compact resolvent, \eqref{nc} implies that $n_c \le N(A,c) < \infty$. 
Thus for every choice of $n_c+1$ elements $\wt u_i = (u_i \ \wh u_i)^{\rm t} \in \ker (\CA- c) \setminus \{0\}$, $i=1,2,\dots,n_c+1$, the first components 
$\{u_1,u_2,\dots,u_{n_c+1}\}$ are linearly dependent. Hence there exists $(\al_i)_{i=1}^{n_c+1} \in \C^{n_c+1} \setminus \{0\}$ such that
$\sum_{i=1}^{n_c+1} \alpha_i u_i = 0$. 
Since $\sum_{i=1}^{n_c+1} \alpha_i \wt u_i \in \ker (\CA-c)$, \eqref{c-var} implies that 
$\sum_{i=1}^{n_c+1} \alpha_i \wh u_i = 0$ and hence $\sum_{i=1}^{n_c+1} \alpha_i \wt u_i = 0$.

Altogether, it follows that $\dim \ker (\CA-c) \le n_c \le N(A,c)<\infty$.
\end{proof}

\begin{remark}
i) \ A necessary condition for $B$ injective is that $\dim \Hii \le \dim \Hi$.

{\rm ii)} \ If $A>c$, then $N(A,c)=0$; in this case $B$ injective implies $c\notin \sigma_{\rm p}(A)$.
\end{remark}

\begin{proposition}
\label{prop:1.3}
If, under the assumptions of Proposition {\rm \ref{prop:1.3a}}, 
$B$ is injective and 
\begin{align}
\label{ass1}
& u \in \dom A  \cap \,\Ker B^*, \ u \ne 0 \ \implies \ (A-c)u \notin  \Ran B,
\end{align}  
then $c\notin \sigma_{\rm p}(\CA)$. Conversely, if $c\notin \sigma_{\rm p}(\CA)$, then \eqref{ass1} holds.
\end{proposition}

\begin{proof}
If $c \!\in\! \sigma_{\rm p}(\CA)$, then \eqref{c-var} and \eqref{kerc1i} show that the first component of every eigenvector at $c$ does not satisfy \eqref{ass1}.
Vice versa, if \eqref{ass1} does not hold, there exists $u\in \dom (A)\cap \ker B^*$, $u\ne 0$, and  $\wh u \in \Hii$ such that $(A-c)u=-B\wh u$. Then the non-zero vector $(u\ \wh u)^{\rm t}$ is an eigenvector of $\CA$ at $c$ and hence $c\in \sigma_{\rm p}(\CA)$.
\end{proof}

\begin{remark}

i) \ \ If $\Ker B^*=\{0\}$, then condition \eqref{ass1} is trivially satisfied.

ii) \ A sufficient (but not necessary) condition for \eqref{ass1} is
\begin{equation}\label{ass1'}
u\in \dom A  \cap \,\Ker B^*, \ u \ne 0 \ \implies \ ((A-c)u,u) \ne 0;
\end{equation}  
in fact, the latter implies that \vspace{1mm} $(A-c)u \notin (\ker B^*)^\perp = \overline{\Ran B}$. 

\end{remark}


In the next subsection, we characterize the eigenvalue sequences $(\lanm)_{j=1}^{n_1}$ and $(\lanp)_{j=1}^{n_2}$ of $\CS$ by min-max variational principles.
The following proposition provides necessary information for this and, in particular, for the index shifts occurring therein; 
here it is crucial that the spectrum of $\CS$ is bounded from below and has a gap
to the right of $c$.

\begin{proposition}
\label{new-est}
Let $\H$ be a Hilbert space, $\H=\Hi\oplus\Hii$ with Hilbert spaces $\Hi$,~$\Hii$. 
Let $A$ be a self-adjoint operator in $\Hi$ with compact resolvent and bounded from below, 
$B: \Hii \to \Hi$  a  non-zero  bounded linear operator, 
$c\in\R$, and 
\[
  N(A,c)= \dim \CL_{(-\infty,c]}(A).
\]
Then the operator function $\Si$ given by \eqref{opfS} has the following properties:
\begin{enumerate}
\item[{\rm i)}] 
$\Si$ satisfies Assumptions {\rm (i)}--{\rm (iv)} of\, \cite{MR2068432} on $(-\infty,c)$ and on $(c,+\infty)$;
\item[{\rm ii)}] there exist $\al_0 \!\in\! (-\infty,c)$, $\al_1 \!\in\! (c,+\infty)$ with $(-\infty,\al_0) \!\subset\! \rho(\Si)$, $(c,\al_1) \!\subset\! \rho(\Si)$; 
\item[{\rm iii)}] the dimensions 
\begin{alignat}{2}
\nonumber
 \dim \CL_{(-\infty,0)} (\Si(\la)) & =: \kam, \quad && \la \in (-\infty,\al_0), \\ 
\label{i-shift} 
 \dim \CL_{(-\infty,0)} (\Si(\la)) & =: \kap, 
 \quad && \la \in (c,\al_1),  
\end{alignat}
are independent of $\la$ and finite with
\[
\kam=0, \quad \kap \le N(A,c);
\]
\item[{\rm iv)}]
if $\,B$~has closed range, 
$\,\PA$ denotes the orthogonal projection onto $\CL_{(-\infty,c]}(A)$ 
and $\PB$ the orthogonal projection onto $\ker B^*$, then
\begin{equation}
\label{kappa2est}
  \kap  \le \rank \big( \PA \PB \big) \le \min\{N(A,c),\dim \ker B^*\}.
\end{equation}
\end{enumerate}
\end{proposition}

\begin{proof}
i) Assumptions (i) and (ii) in \cite{MR2068432} are satisfied since $\dom \Si(\la)\!=\!\dom A$ is independent of $\la$ and $\Si$ is holomorphic on $(-\infty,c)$ and on $(c,+\infty)$.
Because 
\begin{equation}
\label{schur-der}
  \frac{\d}{\d \la} (\Si(\la)u,u) = - \|u\|^2 - \frac{\|B^*u\|^2}{(c-\la)^2} \le - \|u\|^2, \quad u \in  \dom A, \, u\ne 0,
\end{equation}
we have
$\Si'(\lambda)\le -I$ for $\la \in (-\infty,c)$ and for $\la\in(c,+\infty)$. Thus also Assumption (iii) in \cite{MR2068432} is satisfied for $\Si$;
note that the values $-\infty$ and $+\infty$, respectively, of the functional $p$ in \cite[(2.3)]{MR2068432} can be replaced by the end-points $\alpha$ and~$\beta$, respectively of the interval in which the operator function is considered, here $c$ in both cases.
For fixed $\la \in \R\setminus\{c\}$, the operator $\Si(\la)$ is a bounded symmetric perturbation of the self-adjoint operator $A-\la$ which is bounded from below with compact resolvent. This shows that $\Si(\la)$ has finitely many negative eigenvalues and hence Assumption (iv) in \cite{MR2068432} is satisfied.

ii) Due to i) we can apply  \cite[Lemma~2.6]{MR2068432} which yields that the function $\la \mapsto \dim \CL_{(-\infty,0)} (\Si(\la))$ has constant values $\kam$ on $(-\infty,\al_0)$ and $\kap$ on $(c,\al_1)$ for all $\al_0 \in (-\infty,c)$, $\al_1 \in (c,+\infty)$ with $(-\infty,\al_0)\subset \rho(\CS)$, $(c,\al_1)\subset \rho(\CS)$.

Since $(\Si(\la)u,u) \to +\infty$, $u\in \dom A$, $u\ne 0$, for $\la \to - \infty$ and $\Si'(\lambda)\le -I$ for $\la \in (-\infty,c)$, there exists $\al_0 \in (-\infty,c)$ with $\Si(\la)\gg 0$ for $\la \in (-\infty,\al_0)$. Hence  
we have $(-\infty,\al_0) \!\subset\! \rho(\Si)$ and $\kam=\dim \CL_{(-\infty,0)} (\Si(\la)) = 0$ for all $\la \in (-\infty,\al_0)$.

Since $\{c\}=\sigma(C)$, Proposition \ref{prop:1.1} iv) shows that
$(c,c+\varepsilon) \subset \rho(\CA) \setminus\{c\}=\rho(\CS)$ for some $\varepsilon>0$,  without loss of generality $c+\eps < \nu_{N+1}(A)$,
and hence we can choose $\al_1=c+\eps$. 
With $\PA$ defined as in claim iv) and $\PA^\perp := I-\PA$, \vspace{-1mm}  we have  
\begin{equation}
\label{schur-index-0}
  \PA A \PA \le c P, \quad \PA^\perp  A \PA^\perp  \ge \nu_{N+1}(A) P^\perp.
\end{equation}
Then, for $\la \in (c,\al_1)$ and \vspace{-1mm} $u \in \dom A$, 
\begin{align}
\label{schur-index-1}
 (\Si(\la)u,u) & = ( \PA (A \!-\!\la) \PA u,u) \!+\! ( \PA^\perp (A \!-\!\la) \PA^\perp u,u) \!-\! \displaystyle{\frac{(BB^*u,u)}{c-\la}} \\
 &\ge ( \PA (A -\la) \PA u,u)
\nonumber               
\end{align}
since $c<\la < \al_1 <\nu_{N+1}(A)$ which implies that the last two terms in \eqref{schur-index-1} are non-negative. 
By the classical min-max variational principle, this shows that the number of negative eigenvalues of $\Si(\la)$ is less than or equal to the number of negative eigenvalues of $\PA (A -\la) \PA$. The latter is $N(A,c)$ by definition of~$\PA$ and so $\kap=\dim \CL_{(-\infty,0)} (\Si(\la)) \le N(A,c)$ for $\la \in (c,\al_1)$.

iii) In order to show the first inequality in \eqref{kappa2est}, let $\PB^\perp:= I-\PB$. 
We proceed from the equality in \eqref{schur-index-1} above, but now we keep the last term  in \eqref{schur-index-1} to obtain that,
for $\la \in (c,\al_1)$ and $u \in \dom A$, 
\begin{equation}
\label{schur-index-2}
(\Si(\la)u,u) \ge \big( \PA (A -\la) \PA u,u \big) + \frac{(BB^*u,u)}{\la-c}. 
\end{equation}

Since $\la>c$, the operator $\PA (A -\la) \PA$ is non-positive. Hence,
for $u\in \Hi$, $u=(u_1\ u_2)^{\rm t} \in \ker B^* \oplus \Ran B$, 
we can \vspace{-1mm} estimate
\begin{equation}
\label{schur-index-3}
\begin{array}{rl}
  \hspace{-4mm} \big(\PA (A-\la) \PA u,u \big) \hspace{-3mm}
  &= 
  \left(\matrix{cc}{\PB \PA (A-\la) \PA \PB & \PB \PA A \PA \PB^\perp \\ (\PB \PA A \PA \PB^\perp)^* & 
  \PB^\perp \PA (A-\la) \PA  \PB^\perp} \!\!\displaystyle \binom{u_1}{u_2}, \binom{u_1}{u_2} \!\right)\\[4mm]
  & \ge 2 \,\Big( \big( \PB \PA (A-\la) \PA \PB u_1,u_1\big) + \big( \PB^\perp \PA (A-\la) \PA  \PB^\perp u_2,u_2\big) \Big). \hspace{-7mm}
\end{array}
\end{equation}

Since $A \!\ge\! \nu_1(A)$ and $\la \!\le\! \al_1$, we have  $\PA (A-\la) \PA \!\ge\! (\nu_1(A) - \al_1) P$. 
Because $B$ has closed range by assumption and $\PB^\perp$ is the orthogonal projection onto $(\ker B^*)^\perp$, 
there exists $\beta>0$ such that $\PB^\perp B B^* \PB^\perp \ge \beta^2 $ on $\Ran B$. 
Hence,~if we choose $\la \in (c,\al_1)$ such \vspace{-2mm}that 
\begin{equation}
\label{spring}
c < \la  \le c + \frac{\beta^2}{2(\al_1-\nu_1(A))},
\vspace{-2mm}
\end{equation}
then \eqref{schur-index-2}, \eqref{schur-index-3} yield \vspace{-1mm}  that
\begin{align*}
(\Si(\la)u,u) & \ge 2 \, \big( \PB \PA (A-\la) \PA \PB u_1,u_1\big) + 2 (\nu_1(A) - \al_1) \|u_2\|^2 + \frac{\beta^2}{\la-c} \|u_2\|^2\\
&\ge 2 \, \big( \PB \PA (A-\la) \PA \PB u_1,u_1\big).
\end{align*}
Now the classical min-max variational principle shows that, for $\la$ with \eqref{spring},
the number of negative eigenvalues of $\CS(\la)$, is less than or equal to the 
number of negative eigenvalues of the negative operator $\PB \PA (A-\la) \PA \PB$, which is equal to $\rank (\PA \PB)$. 
The second inquality in \eqref{kappa2est} is obvious.
\end{proof}

\subsection{Variational principles.}
\label{subsec:2.2}

In the following we characterize the two eigenvalue sequences $(\lanm)_{j=1}^{n_1} \subset (-\infty,c)$, 
$(\lanp)_{j=1}^{n_2} \subset (c,+\infty)$ 
of $\Si$ by means of min-max variational principles. They are based on variational principles for eigenvalues of analytic operator functions established in \cite{MR2068432}, 
applied to the operator function $\Si$ in the intervals $(-\infty,c)$ and $(c,+\infty)$. 

One of the key properties of the Schur complement ensuring that the assumptions of \cite{MR2068432} are met 
is its Nevanlinna property ensuring monotonicity on $(-\infty,c)$ and $(c,+\infty)$. It guarantees the existence of (at most) one zero of 
the function $\la \mapsto (\Si(\la)u,u)$ on each of these intervals; this zero serves as a generalized Rayleigh functional in the variational principle.

\begin{lemma}
\label{prop:2.2}
Let $\H$ be a Hilbert space, $\H=\Hi\oplus\Hii$ with Hilbert spaces $\Hi$,~$\Hii$. 
Let $A$ be a self-adjoint operator in $\Hi$ with compact resolvent and bounded from below, 
$B: \Hii \to \Hi$  a  non-zero   bounded linear operator, 
and $c\in\R$. 
Then, for every $u\in \dom A$, $u \ne 0$, the function
\begin{equation}
\label{dec2}
 \varphi_u(\la) := (\Si(\la)u,u) = \big((A-\la)u,u\big) - \frac{(BB^*u,u)}{c-\la}, \quad \la \in \R\setminus \{c\},
\end{equation}
has at most one pole at $c$, 
is strictly decreasing on $(-\infty,c)$ and on $(c,+\infty)$. 
Thus~$\varphi_u$ has at most one zero $p_1(u)\!\in\!(-\infty,c)$ and at most one zero $p_2(u)\!\in\!(c,+\infty)$, given~by 
\begin{equation}
\label{ppm}
  p_{1,2} (u) = \frac{1}{2} \bigg( \frac{(Au,u)}{\|u\|^2} + c \bigg) \mp \sqrt{\frac{1}{4} \bigg( \frac{(Au,u)}{\|u\|^2} - c \bigg)^2  + \frac{\|B^{*}u\|^2}{\|u\|^2}}.
\end{equation}
\end{lemma}

\vspace{0.5mm}

\begin{remark}
i) If $B^*u\!=\!0$, then $\varphi_u$ is linear with zero $(Au,u)$; nevertheless, formula \eqref{ppm} is still meaningful and gives  
$p_1(u)\!=\!\min\{(Au,u),c\}$, $p_2(u)\!=\!\max\{(Au,u),c\}$.

ii) The functionals $p_{1,2}$ induced by the numerical range of the Schur complement $\Si$ are related to the functionals $\la_\pm$ 
induced by the quadratic numerical range of the block operator matrix $\CA$ (see \cite[(3.1) and Lemma 3.7]{MR1948455}).
\end{remark}


\begin{proof}
The claimed monotonicity is immediate from inequality \eqref{schur-der}.
If $B^*u\ne 0$, then the claim follows from the fact that the function in \eqref{dec2} has a pole at $c$ with $(\Si(\la)u,u) \to +\infty$ for $\la \to - \infty$ and $\la \searrow c$, $(\Si(\la)u,u) \to -\infty$ for $\la \nearrow c$ and $\la \to +\infty$. If $B^*u=0$, then 
the assertion is immediate from the fact that the function in \eqref{dec2} is linear with precisely one zero $(Au,u)$.

Formula \eqref{ppm} follows from the fact that $\la\in\R \setminus \{c\}$ is a zero of $\varphi_u$ if and only if it is a solution of the quadratic equation
\[
 \big((A-\la)u,u\big)(c-\la) - (BB^*u,u)=0. \qedhere
\]
\end{proof}

\begin{theorem}
\label{thm:2.3a}
Let $\H$ be a Hilbert space, $\H=\Hi\oplus\Hii$ with Hilbert spaces $\Hi$,~$\Hii$. 
Let $A$ be a self-adjoint operator in $\Hi$ with compact resolvent and bounded from below, $B: \Hii \to \Hi$ a  non-zero  bounded linear operator, and $c\in\R$.
Let $N(A,c)= \dim \CL_{(-\infty,c]} \le \dim\Hi$ and let the eigenvalues $(\nu_{j}(A))_{j=1}^{\dim\Hi}$ of~$A$, counted with multiplicities, be ordered non-decreasingly, i.e.
\[
\nu_1(A) \le \dots \le \nu_{N(A,c)}(A) \le c < \nu_{N(A,c)+1}(A) \le \cdots.
\]
Then the spectrum of \,$\CS$ consists of two $($finite or infinite$)$ eigenvalue sequences $(\lanm)_{j=1}^{n_1} \subset (-\infty,c)$ 
and $(\lanp )_{j=1}^{n_2} \subset (c,+\infty)$, $n_1$, $n_2 \in \N_0 \cup \{\infty\}$,
which can be characterized as
\begin{align}
\label{minmax-}
\lanm &= \min_{\ontops{\CL \subset \dom A}{\dim \CL = j}} \,\max_{\ontops{u \in \CL}{u\ne 0}} \ p_1(u),  \quad j=1,2,\dots, n_1,\\
\label{minmax+}
\lanp &= \!\!\!\!\min_{\ontops{\CL \subset \dom A}{\dim \CL = j+\kap }} \!\!\max_{\ontops{u \in \CL}{u\ne 0}} \ p_2(u), \quad j=1,2,\dots, n_2.
\end{align}  
Here $p_{1}$, $p_{2}$ are the functionals in \eqref{ppm} and the index shift $\kap$ defined in \eqref{i-shift} satisfies
$\kap \le N(A,c)$; in particular, $\kap=0$ if $A>c$. Moreover, 
\begin{enumerate}
\item[{\rm i)}] if $\,B$ has closed~range 
and $\PA$ and $\PB$ are the orthogonal projections onto $\CL_{(-\infty,c]}(A)$ and $\ker B^*\!\!$, respectively, then
\[
\kap  \le \rank \big( \PA \PB \big) \le \min\{N(A,c),\dim \ker B^*\};
\]
\item[{\rm ii)}] $n_2\!=\!\infty$ if and only if $\,\dim \Hi \!=\! \infty$; in this case, $(\lanp )_{j=1}^{\infty} \!\subset\! (c,+\infty)$ accumu\-lates $($only$)$ to~$+\infty$;
\item[{\rm iii)}] $n_1\!=\!\infty$ if and only if $\,\dim \ov{\Ran B} \!=\! \infty$; 
in this case, $(\lanm )_{j=1}^{\infty} \!\subset\! (-\infty,c)$ accumulates $($only$)$ to $c$ from the left.
\end{enumerate}
\end{theorem}


\begin{proof}
Proposition~\ref{prop:2.1} i) yields that $\sigma_{\rm p}(\CS) = \sigma_{\rm p}(\CA) \setminus \{c\}$. 
Moreover, $n_1 + n_2 + \dim \CL_{\{c\}}(\CA) = \dim \Hi + \dim \Hii$.

By Proposition \ref{new-est} i), the Schur complement $\Si$ satisfies Assumptions (i) to (iv) of \cite[Theorem~2.1]{MR2068432} on each of the two intervals $(-\infty,c)$ and $(c,+\infty)$. 
Hence the variational characterizations \eqref{minmax-} and \eqref{minmax+} follow from \cite[(2.9]{MR2068432} applied on $(-\infty,c)$ and $(c,+\infty)$, respectively; the claims for the index shifts $\kam$ in $(-\infty,c)$ and $\kap$ in $(c,+\infty)$, in particular i), follow from Proposition \ref{new-est} ii).
%


Claim ii) was proved in \eqref{inf-or-not}. For the proof of claim iii), we first suppose that $\,\dim \ov{\Ran B} \!<\! \infty$. 
Then $\CA$ is a finite rank perturbation of diag$\,(A,c)$ and hence there are at most finitely many eigenvalues in $(-\infty,c)$, i.e.\ $n_1<\infty$.
For the converse,  
assume that $\dim\ov{\Ran B}\!=\!\infty$. If $x\notin\ker B^* = \ker (BB^*)$, then
\begin{equation}\label{inf}
\lim_{\la \nearrow c} (S(\la)x,x) = -\infty.
\end{equation}
Now let $j\in\mathbb N$ be arbitrary.
Then there exists a subspace  $\CL_{j} \subset \dom A$ with $\dim \CL_{j} = j$ and $\CL_{j} \cap \ker B^* = \{0\}$.
In fact, since $\ov{\Ran B}$ is infinite dimensional by assumption, there exists a subspace  $\CL_{j}' \subset \ov{\Ran B}$ with
$\dim \CL_{j}' = j$. Let $\{e_1', e_2', \dots, e_{j}'\}$ be an orthonormal basis of $\CL_{j}'$ and $\eps \in (0,1/2)$. As
$\dom A$ is dense in $\Hi$, there exist  $\{e_1, e_2, \dots, e_{j}\} \subset \dom A$ such that $\|e_k-e_k'\|< \eps$, $k=1,2,\dots, j$.
Then
\begin{equation}
\label{close}
  | (e_i, e_k')| \le \de_{ik} + \eps, \quad i,k=1,2,\dots,j,
\end{equation}
where $\de_{ik}$ is the Kronecker symbol. Since $\eps<1/2$, it is easy to see that \eqref{close} implies that $\{e_1, e_2, \dots, e_{j}\}$ are linearly independent and that 
$\CL_{j}:= {\rm span}\,  \{e_1, e_2, \dots, e_{j}\} \cap \ker B^* = {\rm span}\,  \{e_1, e_2, \dots, e_{j}\} \cap (\ov{\Ran B})^\perp = \{0\}$.

Since $\CL_{j} \cap \ker B^* =\{0\}$ and $\CL_{j}$ is finite dimensional, whence closed, relation \eqref{inf} 
implies that there exists $\la_0\!<\!c$ such that $(\CS(\la_0)x,x)\!<\!0$, $x\in\CL_{j}\setminus\{0\}$. The latter means that 
\[
   p_1(x)<\la_0,\quad x\in\CL_{j}\setminus\{0\}.
\]
It follows that $\max_{x\in\CL_{j}\setminus\{0\}}p_1(x)<\la_0$. Now the variational principle \eqref{minmax-}
implies that $\la_{j} <\la_0$. Therefore there are at least $n$ eigenvalues of $\CS$ below $c$. Since $n$ was arbitrary, 
there must be infinitely many eigenvalues of $\CS$ below $c$.
%
\end{proof}

%
%


\begin{remark}
\label{formdomain}
The variational principles \eqref{minmax-}, \eqref{minmax+} continue to hold if we replace the domain $\dom A$ of the operator $A$ therein by the form domain $\dom \mathfrak{a} = \dom |A|^{1/2}$ of the quadratic form $\mathfrak{a}$ associated with $A$, i.e.\ $\mathfrak{a}[u]:=(Au,u)$ (see \cite[Theorem~3.1, Lemma 3.5]{MR2077211}); 
note that the functionals $p_{1,2}(u)$ in \eqref{ppm} are defined for all $u \in \dom \mathfrak{a}$.
\end{remark}

\begin{remark}
Propositions \ref{inf-at-c} and \ref{prop:1.3a} together imply that if ${\rm dim}\,\Hii = \infty$ and $B$ is injective, then $c$ is an accumulation point of eigenvalues 
of $\CA$, and hence of~$\CS$, from the left. 
Theorem \ref{thm:2.3a} iii) shows that $\dim \ov{\Ran B} = \infty$ ensures eigenvalue accumulation at $c$ from the left. 

The latter result is stronger because ${\rm dim}\,\Hii \!=\! \infty$ and $B$ injective imply that $\dim \ov{\Ran B} \!=\! \infty$. Otherwise, if $\dim \ov{\Ran B} \!<\! \infty$, there exist $n_B\in \N$ and $\wh b_i \!\in\! \Hii$, $b_i \!\in\! \Hi$, $i\!=\!1,2,\dots,n_B$, with $B\!=\! \sum_{i=1}^{n_B} (\cdot, \wh b_i) b_i$. Since 
${\rm dim}\,\Hii \!=\! \infty$, there exists $\wh u_0 \!\in\! \Hii \setminus \{0\}$, $\wh u_0 \perp \wh b_i$, $i\!=\!1,2,\dots,n_B$. Then $B \wh u_0 \!=\! 0$, a contradiction to $B$ being injective.
\end{remark}

%

The estimate for the index shift $\kap$ in Theorem \ref{thm:2.3a} i) is of special interest for numerical approximations 
where infinite dimensional spaces have to be replaced by finite dimensional~ones (see Section \ref{sec:Galerkin}).

\begin{remark}
\label{remark:shift}
If $\Hi$, $\Hii$ are finite dimensional, then  $\dim \ker B^* \!= \dim \Hi - \dim \Ran B^* \!= \dim \Hi - \dim (\ker B)^\perp \ge \dim \Hi- \dim\Hii$.
Hence the condition $\dim \Hi- \dim\Hii \ge N(A,c)$ ensures that $\min\{N(A,c),\dim \ker B^*\}=N(A,c)$.

So to obtain correct information on the index shift $\kap$ of an infinite dimensional problem in $\Hi \oplus \Hii$ with $\dim \ker B^* = \infty$ 
using finite-dimensional approximations in $\Hi_M \oplus \Hii_M$, one should ensure that $\dim \Hi_M -\dim\Hii_M \ge N(A_M,c)$ where
$N(A_M,c)$ is the number of eigenvalues of the corresponding operator $A_M$ in $\Hi_M$ that are $\le c$. 
\end{remark}


\subsection{Two-sided eigenvalue estimates.}
\label{subsec:2.3}

Next we show that the min-max variational principles in Theorem~\ref{thm:2.3a} provide two-sided estimates for all the eigenvalues 
of $\Si$.
The bounds are expressed in terms of the eigenvalues of the left upper entry $A$ of $\CA$ and the norm of the off-diagonal entry~$B$. 

\begin{theorem}
\label{thm:2.3}
Let $\H$ be a Hilbert space, $\H=\Hi\oplus\Hii$ with Hilbert spaces $\Hi$,~$\Hii$. 
Let $A$ be a self-adjoint operator in $\Hi$ with compact resolvent and bounded from below, $B: \Hii \to \Hi$ a  non-zero  bounded linear operator, and $c\in\R$.
Let $N(A,c)= \dim \CL_{(-\infty,c]} \le \dim\Hi$ and let the eigenvalues $(\nu_{j}(A))_{j=1}^{\dim\Hi}$ of~$A$, counted with multiplicities, be ordered non-decreasingly, i.e.
\begin{equation}
\label{defN}
\nu_1(A) \le \dots \le \nu_{N(A,c)}(A) \le c < \nu_{N(A,c)+1}(A) \le \cdots.
\end{equation}
Then the eigenvalues $(\lanm)_{j=1}^{n_1}\! \subset\! (-\infty,c)$, 
$(\lanp )_{j=1}^{n_2}\! \subset\! (c,+\infty)$, $n_1, n_2 \!\in\! \N_0 \!\cup\! \{\infty\}$, of $\CS$ 
satisfy the two-sided estimates
\begin{equation}
\label{twosidedestimates}
  \lanmL \le \lanm \le \lanmU, \quad 
  \la_{2,j+\kap }^L \le \lanp \le \la_{2,j+\kap }^U
\end{equation}
for $j=1,2,\dots,n_1$ and $j=1,2,\dots, n_2-\kap$, respectively, where
\begin{align}
\label{upbound-}
\lanmL &:= \frac{\nu_{j}(A)+c}2
  -\sqrt{\biggl(\frac{\nu_{j}(A)-c}2\biggr)^{\!2}\!\!+\|B\|^2}, \\ 
\label{lobound-}
\lanmU &:= \frac{\nu_{j}(A)+c}2
  -\sqrt{\biggl(\frac{\nu_{j}(A)-c}2\biggr)^{\!2}\!\!+\!\min\sigma(BB^*)} \le \min\{ \nu_{j}(A), c\}
\vspace{-1mm}  
\end{align}
\vspace{-1mm}and 
\begin{align}
\label{upbound+}
\hspace{-2mm}
\la_{2,j+\kap }^L &\!\!:=\! \frac{\nu_{j+\kap\!}(A)\!+\!c}2
  \!+\!\sqrt{\!\biggl(\frac{\nu_{j+\kap\!}(A)\!-\!c}2\biggr)^{\!2}\!\!\!+\!\min\sigma(BB^*)} \ge \max\{ \nu_{j+\kap}(A),c\}, \hspace{-2mm}\\
\label{lobound+}
\hspace{-2mm}
\la_{2,j+\kap }^U &\!\!:=\! \frac{\nu_{j+\kap\!}(A)\!+\!c}2
  \!+\!\sqrt{\!\biggl(\frac{\nu_{j+\kap\!}(A)\!-\!c}2\biggr)^{\!2}\!\!\!+\|B\|^2}. 
\end{align}
Here $\kap$ and $n_1$, $n_2$ satisfy 
\begin{enumerate}
\item[{\rm i)}] $\kap \le N$, and $\kap \le \min\{N(A,c),\dim \ker B^*\}$ if $B$ has closed range;  
\item[{\rm ii)}] $n_2 = \infty$ if and only if $\dim \Hi=\infty$;
\item[{\rm iii)}] $n_1=\infty$ if and only if $\,\dim \ov{\Ran B} \!=\! \infty$;
\item[{\rm iv)}] $n_1 + n_2 + \dim \CL_{\{c\}}(\CA) = \dim \Hi + \dim \Hii$, \\
$n_1 \ge N(A,c)$, $n_2 \ge \dim \Hi -N(A,c)$.
\end{enumerate}
\end{theorem}

\begin{proof}
In order to prove \eqref{twosidedestimates}, we use the variational characterizations \eqref{minmax-}, \eqref{minmax+} of $\lanm$, $\lanp$ proved in Theorem \ref{thm:2.3} with $p_{1,2}$ as defined in \eqref{ppm}; note that for $(\lanm)_{j=1}^{n_1}$ the index shift is $0$ by Theorem~\ref{thm:2.3}.
 
For estimating $p_{1,2}$ we note that the functions 
\begin{align}
\label{mon1}
f(s,\beta)&:=\frac{s+c}2+\sqrt{\Big(\frac{s-c}2\Big)^2+\beta^2}, \quad s \in \R, \ \beta \in [0,+\infty), \\
\label{mon2}
 g(s,\beta)&:=\frac{s+c}2-\sqrt{\Big(\frac{s-c}2\Big)^2+\beta^2}, \quad s \in \R, \ \beta \in [0,+\infty),
\end{align}
are both increasing in $s$; moreover, $f$ is increasing in $\beta$, while $g$ is decreasing in $\beta$.


These monotonicity properties, together with the classical min-max characterization of the eigenvalues $\nu_j(A)$ and 
the two-sided bounds 
\[
 0 \le \min\sigma(BB^*) \le \frac{\|B^*u\|^2}{\|u\|^2} \le \|B^*\|^2 = \|B\|^2, \quad u\ne 0,
\] 
yield the claimed estimates \eqref{twosidedestimates}; note that the leftmost bound 
yields the estimates in \eqref{lobound-}, \eqref{upbound+} since \vspace{-1mm} e.g.
\[
\lanmU \le \frac{\nu_{j}(A)+c}2  -\sqrt{\biggl(\frac{\nu_{j}(A)-c}2\biggr)^{\!2}} = \min\{ \nu_{j}(A), c\}. 
\vspace{-1mm}
\]

Claim i) follows from Proposition \ref{new-est}. Claims ii), iii) were proved in Theorem \ref{thm:2.3a} ii), iii). The first claim in iv) is obvious; 
the lower bounds for $n_1$ and $n_2$ follow from the upper estimate in \eqref{dec1a} and the lower estimate in \eqref{dec1d}.
\end{proof}

The following corollary shows how to obtain two-sided computable bounds for the eigenvalues $(\lanm)_{j=1}^{n_1}$ and $(\lanp)_{j=1}^{n_2}$ of $\CA$. 
It is an immediate consequence of Theorem \ref{thm:2.3} and of the monotonicity of the functions $f$, $g$ in \eqref{mon1}, \eqref{mon2}.

\begin{corollary}
\label{boundsmonotonic}
Suppose that, for some $\wt N \le \dim \Hi$, the first $\wt N$ eigenvalues of $A$ admit two-sided estimates 
\begin{equation}
\label{twosided-computable}
\nu_{j}^L(A) \le \nu_{j}(A) \le \nu_{j}^U(A), \quad j=1,2,\dots, \wt N.
\end{equation}
Then
\begin{equation}
\label{twosidedestimates-computable}
  \wh \la_{1,j}^L \le \lanm \le \wh \la_{1,j}^U, \quad 
  \wh \la_{2,j+\kap }^L \le \lanp \le \wh \la_{2,j+\kap }^U, \quad j=1,2,\dots, \wt N,
\end{equation}
where $\wh \la_{1,j}^L$, $\wh \la_{2,j}^L$ are obtained from $\lanmL$, $\lanpL$ by replacing $\nu_{j}(A)$ by its lower bound $\nu_{j}^L(A)$, while
$\wh \la_{1,j}^U$, $\wh \la_{2,j}^U$  are obtained from $\lanmU$, $\lanpU$ by replacing $\nu_{j}(A)$ by its upper bound $\nu_{j}^U(A)$.
\end{corollary}

\begin{remark}
If the two-sided bounds for $\nu_{j}(A)$ in \eqref{twosided-computable} are computable, then so are the two-sided bounds for $\lanm$, $\lanp$ in \eqref{twosidedestimates-computable}.
Computable upper bounds for $\nu_{j}(A)$ may be obtained by the classical min-max variational principle (see e.g. \cite{MR738929}),
e.g.\ from Galerkin approximations since, for every $\wt N$-dimensional subspace $\Hi_{\wt N}$ of $\dom A \subset \Hi$, 
\[
 \nu_{j}(A) = \min_{\ontops{\CL \subset \dom A}{\dim \CL = j}} \,\max_{\ontops{u \in \CL}{u\ne 0}} \ (Au,u) 
        \le \min_{\ontops{\CL \subset \Hi_{\wt N}}{\dim \CL = j}} \,\max_{\ontops{u \in \CL}{u\ne 0}} \ (Au,u) =: \nu_{j}^U(A),  
            \quad j=1,2,\dots, \wt N.
\]
\end{remark}

\begin{proposition}
\label{beh-bounds}
Let the assumptions of Theorem {\rm \ref{thm:2.3}} hold, set $\nu_0(A):=-\infty$, \vspace{-2mm} and 
\begin{equation}
\label{arctan}
\de_B(t):= \|B\| \tan \left( \frac 12 \arctan \frac{2\|B\|}{|t-c|} \right), \quad t\in (-\infty,+\infty).
\end{equation}
Then the eigenvalues $(\lanm)_{j=1}^{n_1} \!\subset\! (-\infty,c)$ of $\CA$ 
satisfy the estimates
\begin{alignat}{3}
\label{dec1a}
\hspace{6mm} \nu_{j}(A) - \de_B(\nu_{j}(A)) & \le \lanm && \le \nu_{j}(A), \quad && j=1,2,\dots,N(A,c), \\
\label{dec1b}
\hspace{6mm} c - \de_B(\nu_{j}(A)) & \le \lanm && \le c, \quad && j=N(A,c)+1,\dots,n_1, 
\end{alignat}
and the eigenvalues  $(\lanp)_{j=1}^{n_2} \!\subset\! (\nu_1(A),+\infty)$ of $\CA$ 
satisfy
\begin{alignat}{3}
\label{dec1c}
c & \le \lanp && \le c + \de_B(\nu_{j+\kap }(A)) , \quad && j\!=\!1,2,\dots,N(A,c)-\kap,\\
\label{dec1d}
\nu_{j+\kap }(A) & \le \lanp && \le \nu_{j+\kap }(A) +  \de_B(\nu_{j+\kap }(A)), \quad &&  j\!=\!N(A,c)\!-\!\kap\!+\!1,\dots,n_2\!-\!\kap.
\end{alignat}
\end{proposition}

\begin{proof} 
The upper bound for $\lanm$ follows from the definition of $\lanmU$ in \eqref{lobound-}; the lower bound for $\lanp$ follows from the definition of $\la_{1,j+\kap }^L$ in \eqref{upbound+} 
and the estimate $\min\sigma(BB^*)\ge 0$ which yield
\[
 \la_{1,j+\kap }^L  
 \ge \max\{ \nu_{j+\kap }(A),c\} = \begin{cases} \ \ c, \ & \quad j=1,2, \dots, N(A,c)-\kap ,\\ 
 \nu_{j+\kap }(A),  & \quad j=N(A,c)-\kap +1,\dots,n_2-\kap. \\ \end{cases}
\]
The lower bound for $\lanm$ and the upper bound for $\lanp$ follow from the following alternative formulas for the solutions of quadratic equations on the right hand sides of \eqref{upbound-}, \eqref{lobound+} 
(see e.g.\ \cite[Lemma~5.1 and (5.1)]{CT09},
\begin{alignat}{2}
\label{tanarctan}
\la_{1,j}^L & = \min \{ \nu_{j}(A),c\} - \de_B(\nu_{j}(A)), \quad && j=1,2,\dots,n_1, \\
\la_{2,j+\kap }^U & = \max \{ \nu_{j+\kap }(A),c\} + \de_B(\nu_{j+\kap }(A)), \quad && j=1,2,\dots,n_2-\kap, 
\end{alignat}
together with the definition of $N(A,c)$ (see \eqref{defN}). 
\end{proof}

\begin{remark}
\label{re:order}
Note that, if $n_1=\infty$ and $n_2=\infty$, respectively, \eqref{dec1b} and \eqref{dec1d} 
also yield the order of convergence  of $\lanm \!\nearrow\! c$ 
and  of $\lanp - \nu_{j+\kap }(A) \!\searrow\! 0$ for  $j\to\infty$,
\begin{alignat*}{3}
c\!-\! \lanm & \stackrel{n\to\infty}{=}\! {\rm O}(\nu_{j}(A)^{-1}), \quad   && j\!\ge\!N(A,c)\!+\!1,
\\
\lanp \!-\! \nu_{j+\kap}(A)  & \stackrel{n\to\infty}{=}\!{\rm O}(\nu_{j+\kap }(A)^{-1}), \quad  && j\!\ge\!N(A,c)\!-\!\kap \!+\!1,
\end{alignat*}
since a Taylor series expansion shows that $\de_B(t)= {\rm O}(t^{-1})$, $t\to \infty$.
\end{remark}


Classical perturbation theory for the block operator matrix $\CA$ in \eqref{eq:A} with diagonal entries $A$, $c$ and off-diagonal entries $B$, $B^*$ shows that 
\begin{alignat}{2}
\label{ill1}
  &\|B\|< \frac{c-\nu_{N(A,c)}(A)}2   && \implies  \sigma(\CA) \cap \big(\nu_{N(A,c)}(A) +\|B\|,c -\|B\| \big) = \emptyset,  \\
\label{ill2}  
  &\|B\|< \frac{\nu_{N(A,c)+1}(A)-c}2 && \implies  \sigma(\CA) \cap \big( c+\|B\|,\nu_{N(A,c)+1}(A) -\|B\| \big) = \emptyset. \hspace{-5mm}
\end{alignat}
However, the off-diagonal structure of the perturbation allows for stronger inclusions. More precisely,  
the two-sided estimates in Theorem \ref{thm:2.3} and Proposition \ref{beh-bounds} provide tighter estimates 
for subintervals of $(\nu_{N(A,c)}(A),c)$ and $(c,\nu_{N(A,c)+1}(A))$ to be free of eigenvalues.

\begin{proposition}
\label{thm:1.4}
Let $\H$ be a Hilbert space, $\H=\Hi\oplus\Hii$ with Hilbert spaces $\Hi$,~$\Hii$. 
Let $A$ be a self-adjoint operator in $\Hi$ with compact resolvent and bounded from below, $B: \Hii \to \Hi$ a  non-zero  bounded linear operator, and $c\in\R$.
If 
\[ N:= N(A,c)= \dim \CL_{(-\infty,c]}(A)\] 
and the eigenvalues $(\nu_{j}(A))_{j=1}^{\dim \Hi}$ of $A$, counted with multiplicities, 
are ordered non-decreasingly $($see \eqref{defN}$)$,
then the eigenvalues $(\lanm)_{j=1}^{n_1} \subset (-\infty,c)$ 
and  $(\lanp )_{j=1}^{n_2} \subset (c,+\infty)$, $n_1$, $n_2 \in \N_0 \cup \{\infty\}$, of $\,\CS$ 
satisfy the inclusions
\begin{alignat}{2}
\label{dec1e}
 & (\lanm)_{j=1}^N \!\subset\! \big( \nu_1(A) \!-\! \de_B(\nu_1(A)), \nu_N(A) \big), \ \ && (\lanm)_{j=N+1}^{n_1} \!\subset\! \big(c\!-\!\de_B(\nu_{N+1}(A)),c\big),
 \hspace{-4mm} \\
\label{dec1f} 
 & (\lanp)_{j=1}^{N-\kap } \!\subset\! \big( c, c\!+\! \de_B(\nu_N(A))\big), \quad  && (\lanp)_{j=N-\kap +1}^{n_2} \!\subset\! \big(\nu_{N+1}(A), \infty\big). 
 \hspace{-4mm}
\end{alignat}
In particular, if $\|B\| < \sqrt{(c\!-\!\nu_N(A))^2\!+\!(c\!-\!\nu_N(A))(c\!-\!\nu_{N+1}(A))}$, then
\begin{align}
\label{umea2a}
  \sigma(\CA) \cap \big( \nu_N(A), c-\de_B(\nu_{N+1}(A)) \big) = \emptyset, \\
\intertext{and if $\|B\| < \sqrt{(c\!-\!\nu_{N+1}(A))^2\!+\!(c\!-\!\nu_N(A))(c\!-\!\nu_{N+1}(A))}$, then}
\label{umea2b}  
  \sigma(\CA) \cap \big( c + \de_B(\nu_N(A)), \nu_{N+1}(A) \big) = \emptyset;
\end{align}
if $\|B\|< \sqrt 2 \min\{ c-\nu_N(A),\nu_{N+1}(A)-c\}$, then both \eqref{umea2a} and \eqref{umea2b} hold.
\end{proposition}

\smallskip

Note that the lower bound $\min\bigl\{\nu_1(A), c \bigr\}-\de_B(\nu_1(A))$ for the whole spectrum of $\CA$ in \eqref{dec1e} coincides with the bound in Proposition \ref{stillill} since therein $\min\sigma(A)=\nu_1(A)$, $\min\sigma(C)=c$, 
and $\de=\de_B(\nu_1(A))$.

\smallskip

The eigenvalue estimates \eqref{umea2a}, \eqref{umea2b} are better than the classical perturbation results \eqref{ill1}, \eqref{ill2} for two reasons;
firstly, they show that the spectral gaps $\big(\nu_{j}(A),c\big)$, $\big(c,\nu_{N+1}(A)\big)$ between the diagonal elements $A$ and $c$ may shrink 
only from one side and, secondly, the norm constraint on the perturbation is not only improved from a factor $\frac 12$ to $1$ as a result of 
the one-sided shrinking, but even to~$\sqrt 2$. 

\vspace{2mm}

\noindent
\emph{Proof of Proposition} \ref{thm:1.4}.
The first inclusion in \eqref{dec1e} follows from \eqref{dec1a}, the fact that the function $t\mapsto t- \de_B(t)$ is increasing on $(-\infty,+\infty)$, and $\nu_1(A) \le \nu_{j}(A) \le \nu_N(A)\le c$, $j=1,2,\dots,N$. The second inclusion in \eqref{dec1e} follows from \eqref{dec1b}, the fact that the function $t\mapsto \de_B(t)$ is decreasing on $(c,+\infty)$, and $c < \nu_{N+1}(A) \le \nu_{j}(A)$, 
$j=N+1,N+2,\dots$. The inclusions in \eqref{dec1f} follow from \eqref{dec1c}, \eqref{dec1d} by similar arguments.

It is not difficult to check that $\de_B(\nu_{N+1}(A)) < c-\nu_N(A)$ if and only if $\|B\| < \sqrt{(c\!-\!\nu_N(A))^2\!+\!(c\!-\!\nu_N(A))(c\!-\!\nu_{N+1}(A))}$.
Hence the latter condition guarantees that the two intervals in \eqref{dec1e} are disjoint and thus \eqref{umea2a} follows. In a similar way, claim \eqref{umea2b} follows from \eqref{dec1f}. Finally, for the last claim it suffices to note~that e.g.\ $ \sqrt 2 \min\{ c\!-\!\nu_N(A),\nu_{N+1}(A)\!-\!c\} \!\le\! \sqrt{(c\!-\!\nu_N(A))^2\!+\!(c\!-\!\nu_N(A))(c\!-\!\nu_{N+1}(A))}$.
\qed
%

\begin{remark}
\label{cor:1.5}
If $c\in\sigma_{\rm p}(A)$, 
we have $c=\nu_N(A)$ since $N=N(A,c)$ and hence $\de_B(\nu_N(A)) = \|B\|  \lim_{\tau\to\infty} \tan \big( \frac 12 \arctan \tau \big) = \|B\|$.
In this case, \eqref{dec1e} is not applicable and \eqref{arctan} \vspace{-1mm} becomes
\[
  \sigma(\CA) \cap  \big(c+\|B\|,\nu_{N+1}(A)\big) = \emptyset.
\]
\end{remark}

As we have seen above, many of the previous results can be strengthened if $A>c$. In this case no eigenvalue of $A$ lies below $c$ and hence $N=N(A,c)=0$, which implies that also $\kap  \, (\le N(A,c)) = 0$ for the index shift $\kap $. The following remark summarizes these stronger results. The proposition below adds a so-called half range basis result for this particular case.

\begin{remark}
\label{rem:1.4a}
Suppose that, in addition to the assumptions of Theorem \ref{thm:2.3} (and hence of Proposition \ref{thm:1.4}), we have $A>c$. Then
the point spectrum of $\CA$ consists of two sequences of eigenvalues of finite multiplicities, 
$(\lanm)_{j=1}^{n_1} \!\subset\! (-\infty,c)$, tending to $c$ if $n_1=\infty$, and $(\lanp)_{j=1}^{n_2} \!\subset\! (\nu_1(A),+\infty)$, tending to $+\infty$ if $n_2=\infty$, which satisfy the two-sided estimates \eqref{twosidedestimates} with $N=\kap =0$, in particular,
\[
\sigma(\CA) \cap  \big(c,\nu_1(A)\big) = \emptyset.
\]
\end{remark}

In the following proposition we formulate basis properties for the first components of the eigenvectors of $\CA$. Observe that, for eigenvalues different from $c$, the latter coincide with the eigenvectors of $\CS$. 

\begin{proposition}
\label{thm:1.4a}
Let $\H$ be a Hilbert space, $\H=\Hi\oplus\Hii$ with Hilbert spaces $\Hi$,~$\Hii$. %
Let $A$ be a self-adjoint operator in $\Hi$ with compact resolvent and bounded from below, 
$(\nu_{j}(A))_{j=1}^{\dim \Hi}$ the sequence of eigenvalues of $A$ ordered non-decreasingly, $B: \Hii \to \Hi$  a  non-zero  bounded linear operator, 
and $c\in\R$ such that $A>c$.
Then each of the following sets can be chosen to form a Riesz basis in~$\Hi$:
\begin{enumerate}
\item[{\rm i)}] the first components 
of the eigenvectors of $\CA$ corresponding to eigenvalues in $(-\infty,c]$, i.e.\ corresponding to
$(\lanm)_{j=1}^{n_1}$ and to $c$ if $c\in \sigma_{\rm p}(\CA)$;
\item[{\rm ii)}] the first components 
of the eigen\-vectors of $\CA$ corresponding to eigenvalues in $(c,+\infty)$, i.e.\ corresponding to $(\lanp)_{j=1}^{n_2}$. 
\end{enumerate}  
\end{proposition}

\begin{proof}
%
The second claim follows directly from \cite[Theorem 3.5]{MR1354980}. The first claim follows from \cite[Theorem 2.3 (iv)]{MR1354980}, which shows that the restriction $\CA_-$ of the block operator matrix $\CA$ to the spectral subspace corresponding to the spectrum in $(-\infty,c)$ is unitarily equivalent to a self-adjoint operator in some Hilbert space, and from the fact that $\sigma(\CA_-)=\sigma(\CA)\cap (-\infty,c) = (\lanm)_{j=1}^{n_1}$ is discrete.
\end{proof}

%% file: Part_H_add_C.tex
\section{\bf Multi-pole case: variational principles}
\label{sec:4}

In this section we consider the case that the rational operator function $\Si$ in the Hilbert space $\Hi$
in \eqref{schur1} is of the \vspace{-1.5mm} form
\begin{equation}
\label{H1}
\hspace{4mm}
\Si(\la)=A-\la-\sum_{\ell=1}^L\dfrac{B_\ell B_\ell^*}{c_\ell-\la}, 
\quad \dom \Si(\la) = \dom A, \quad \la \in \C \setminus \{c_1,c_2,\dots, c_L\},
\end{equation}
with $c_1,c_2, \dots, c_L\in\R$, $c_1 < c_2 < \dots < c_L$, and non-zero
bounded linear operators $B_\ell$ from Hilbert spaces $\Hii_\ell$ to $\Hi$, $\ell=1,2,\dots,L$. 
As before, $A$ is supposed to be self-adjoint and bounded from below with compact resolvent. 

In the sequel, we set $c_0\!:=\!-\infty$, $c_{L+1}\!:=\!+\infty$ and introduce the bounded operators
\[
 \Gamma_\ell:= B_\ell B_\ell^*: \Hi \to \Hi, \quad \ell=1,2,\dots,L.
\]

\subsection{Spectral properties of $\CS$.}
\label{subsec:3.1}
Let $\Hii\!:=\!\Hii_1 \oplus \Hii_2 \oplus \dots \oplus \Hii_L$. Then $\Si$ is the Schur com\-plement of 
the operator matrix $\CA$ in \eqref{eq:A-gen} in 
$\H\!=\! \Hi \oplus \Hii$ given~by
\begin{equation}
\label{eq:A-multi}
\CA\!=\!
\matrix{cc}{\!A &\! B \\\! B^* &\! C }\!=\!
\matrix{c|cccc}{A&B_1\!\!&\!\!B_2\!\!&\!\!\cdots\!\!&\!\!B_L\\ \hline
B_1^*&c_1\!\!&\!\!0\!\!&\!\!\cdots\!\!&\!\!0\\
B_2^*&0\!\!&\!\!c_2\!\!&\!\!\cdots\!\!&\!\!0\\
\vdots&\vdots\!\!&\!\!\vdots\!\!&\!\!\ddots\!\!&\!\!\vdots\\
B_L^*&0\!\!&\!\!0\!\!&\!\!\cdots\!\!&\!\!c_L
}\!, \quad  \dom \CA \!=\! \dom A \oplus \Hii,
\vspace{-1mm}
\end{equation}
i.e.\ $B\!=\!\big(B_1\, B_2\, \dots \,B_L \big)\!:\! \Hii_1 \oplus \Hii_2 \oplus \dots \oplus \Hii_L \!\to\! \Hi$, 
$C\!=\!{\rm diag\,} \big( c_1 I_{\Hii_1},c_2 I_{\Hii_2},\dots ,c_L I_{\Hii_L} \big)$.
%
%

Since $\CA$ is self-adjoint and bounded from below, the spectra $\sigma(\CA)$ and $\sigma(\CS)=\sigma(\CA) \setminus \{c_1,c_2,\dots,c_L\}$ 
are real and bounded from below. By Propositions \ref{prop:2.1} and~\ref{prop:1.1}, we have
\begin{align}
\label{feb18}
  &\sigma_{\rm ess}(\CA) = \sigma_{\rm ess}(C)=\{c_\ell: \dim\,\Hii_\ell = \infty,\  \ell=1,2,\dots,L\}, \\
\nonumber  
  &\sigma_{\rm ess}(\CS) =  \sigma_{\rm ess}(\CA) \setminus \sigma(C) = \emptyset,
\end{align}
and $(c_ L,c_L+\varepsilon_L)\subset\rho(\CA)$ for some $\varepsilon_L>0$. 

Hence the spectrum of $\CS$ outside of its poles $c_\ell$, i.e.\ in the intervals $(-\infty,c_1)$, $(c_{\ell-1},c_{\ell})$, $\ell=2,3,\dots,L$, and
$(c_L,+\infty)$, is discrete and accumulates at most at the points $c_1,c_2,\dots, c_L$ and at $+\infty$, respectively.
Moreover, since $A$ is bounded from below and $\max \sigma_{\rm ess}(C)=c_L$, Proposition \ref{prop:1.1} shows that eigenvalues of $\CA$ cannot accumulate from the right at the points $c_0=-\infty$ and $c_L$. 


Unlike the one pole case, we now also have to study the question of eigenvalue accumulation between two poles. 
Our first result shows that there are gaps in the spectrum to the right of \emph{all} points $c_\ell$, $\ell=0,1,\dots,L$,
and hence the eigenvalues in each interval $(c_{\ell-1},c_{\ell})$ cannot accumulate at $c_{\ell-1}$ from the right, $\ell=1,2,\dots, L$.

\begin{proposition} 
\label{pro1}
Let $\Hi$ be a Hilbert space, $A$ a self-adjoint operator in $\Hi$, bounded from below and with compact resolvent,  let $\Gamma_\ell$, $\ell=1,2,\dots,L$,
be non-zero bounded self-adjoint non-negative operators in $\Hi$, $c_1,c_2, \dots, c_L\in\R$ with $c_1<c_2<\dots<c_L$, and set $c_0=-\infty$, $c_{L+1}=+\infty$. 
Then the operator function $\CS$ given by 
\begin{equation}\label{Meise}
\CS(\la)=A-\la-\sum_{\ell=1}^L\dfrac{\Gamma_\ell}{c_\ell-\la}, \quad \dom\CS(\la)=\dom A, \quad \la \in \C \setminus \{c_1,c_2,\dots, c_L\}, \hspace{-4mm}
\end{equation}
has the following properties:
\begin{enumerate}
\item[\rm i)] $\Si$ satisfies Assumptions {\rm (i)}--{\rm (iv)} of\, \cite{MR2068432} on $(c_{\ell-1},c_{\ell})$, $\ell=1,2,\dots,L\!+\!1$;
\item[\rm ii)] 
there exist $\al_{\ell-1} \in (c_{\ell-1},c_\ell)$ such that $(c_{\ell-1},\al_{\ell-1}) \!\subset\! \rho(\CS)$, $\ell=1,2,\dots,L+1$; \vspace{-1.5mm} further,
\begin{alignat*}{2}
\dim \CL_{(-\infty,0)} (\Si(\la)) & =: \kam, \quad && \la \!\in\! (-\infty,\al_0), \\ 
\dim \CL_{(-\infty,0)} (\CS(\la)) & =: \kappa_\ell, \ \ && \la \!\in\! (c_{\ell-1},\alpha_{\ell-1}), \quad \ell=2,3,\dots,L+1, \hspace{-10mm}
\end{alignat*}
are independent of $\la$ and finite \vspace{-1.5mm} with
\[
  \kam = 0, \quad \kappa_\ell \le N\Big(A\!-\!\!\!\sum_{j=\ell}^L \!\frac{\Gamma_j}{c_j\!-\!c_{\ell-1}},c_{\ell-1}\Big), \quad \ell=2,3,\dots,L+1,
\vspace{-1.5mm}  
\]
in particular, $\kappa_{L+1} \le N(A,c_L)$; \vspace{1mm}
\item[\rm iii)] 
if $\Gamma_{\ell-1}$ has closed range  for some $\ell\in\{2,3,\dots,L\}$, $P_{\ell-1}$ denotes the orthogonal projection onto $\CL_{(-\infty,c_{\ell-1}]}(A)$ and $Q_{\ell-1}$ 
the orthogonal projection onto $\ker \Gamma_{\ell-1}$, \vspace{-2mm} then 
\begin{equation}
\label{iml}
   \ka_\ell \le N\Big(Q_{\ell-1}\Big(P_{\ell-1}AP_{\ell-1}-\sum_{j=\ell}^L\dfrac{\Gamma_j}{c_j-c_{\ell-1}}\Big)Q_{\ell-1},c_{\ell-1}\Big);
   \vspace{-2mm}
\end{equation}
in particular, 
\begin{equation}\label{amsel}
\ka_{L+1}\le{\rm rank}\,(P_L Q_L) \le \min \big\{ N(A,c_L), \dim \ker \Gamma_L \big\}.
\end{equation}
\end{enumerate}
\end{proposition}

\begin{proof}
i) The proof that $\CS$ satisfies Assumptions {\rm (i)}--{\rm (iv)} of\, \cite{MR2068432} on each $(c_{\ell-1},c_{\ell})$, $\ell=1,2,\dots,L+1$,
is completely analogous to the proof of Proposition \ref{new-est} i). 
Here we note that, since $\Gamma_\ell$, $\ell=1,2,\dots,L$, is non-negative and bounded, we still have $\Si'(\lambda)\le -I$ and the operator $\CS(\la)$ is a bounded symmetric perturbation of $A-\la$ for $\la \in (c_{\ell-1}, c_{\ell})$. 

ii) Due to i) we can apply \cite[Theorem 2.1 and Lemma~2.6]{MR2068432} which yield 
that there are $\alpha_{\ell-1} \in (c_{\ell-1},c_\ell)$ with $(c_{\ell-1},\alpha_{\ell-1}) \subset \rho(\CS)$ and 
that the function $\la \mapsto \dim \CL_{(-\infty,0)} (\Si(\la))$ has a constant value
$\kappa_\ell$ on $(c_{\ell-1},\alpha_{\ell-1})$, $\ell=1,2,\dots,L+1$.

The proof of $\kappa_1=0$ is the same as in the one pole case since we again have $(\Si(\la)u,u)$ $\!\to\! +\infty$,
$u\in \dom A$, $u\ne 0$, for $\la \!\to\! - \infty$, $\Si'(\lambda)\!\le\! -I$ for $\la \!\in\! (-\infty,c_1)$, and so $\Si(\la) \gg 0$ for $\la \!\in\! (-\infty,\al_0)$.

In order to prove the upper bound for $\kappa_\ell$, $\ell=2,3\dots,L+1$, let
$\la \in (c_{\ell-1},c_{\ell})$. 
Then the \vspace{-1mm} estimate 
\begin{equation}
\label{ps}
\CS(\la)=A-\la-\sum_{j=1}^L\dfrac{\Gamma_j}{c_j-\la}\ge
A-\la-\sum_{j=\ell}^L\dfrac{\Gamma_j}{c_j-\la}=: \CT_0(\la)
 \vspace{-1mm} 
\end{equation}
shows that
the number of negative eigenvalues of $\CS(\la)$ is less than or equal to the number of negative eigenvalues of $\CT_0(\la)$. 
%
Choosing $\la=c_{\ell-1} + \eps$ with $\eps >0$ arbitrarily small,  we find
\begin{equation}
\label{spring-new}
 \CT_0(c_{\ell-1}+\eps) = A-c_{\ell-1} \!-\! \sum_{j=\ell}^L \frac{\Gamma_j}{c_j\!-\!c_{\ell-1}} - \eps \Big( I +\! \sum_{j=\ell}^L \frac{\Gamma_j}{(c_j\!-\!(c_{\ell-1}\!+\!\eps))(c_j\!-\!c_{\ell-1})}\Big).
\end{equation}
If we let $\eps\searrow 0$, the bound for $\kappa_\ell$ claimed in ii) follows.

iii) Let $Q_{\ell-1}^\perp:=I-Q_{\ell-1}$ be the orthogonal projection onto $\Ran \Gamma_{\ell-1}$. Without loss of generality, we may assume that $\eps_{\ell-1}>0$ is so
small that $(c_{l-1},c_{l-1}+\eps_{\ell-1}] \cap \sigma(A) = \emptyset$.
Then, for  $\la\in (c_{\ell-1},c_{\ell-1}+\eps_{\ell-1})$, we have the estimate
\begin{equation}
\label{ps1}
\CS(\la)\ge P_{\ell-1}(A-\la)P_{\ell-1}-\sum_{j=\ell}^L\dfrac{\Gamma_j}{c_j-\la}+\dfrac{\Gamma_{\ell-1}}{\la - c_{\ell-1}}=:\CT_1(\la)+\dfrac{\Gamma_{\ell-1}}{\la - c_{\ell-1}}.
\end{equation}
Since $c_{\ell-1}<\la<c_\ell$,  $P_{\ell-1}(A-\la)P_{\ell-1}$ is non-positive while $\Gamma_j / (c_j-\la)$ is non-negative, $j=\ell,\ell+1,\dots,L$. Thus $\CT_1(\la)$ is non-positive and hence, for 
$u\in \CH$, $u=(u_1\ u_2)^{\rm t} \in \ker \Gamma_{\ell-1}^*\oplus\Ran \Gamma_{\ell-1}$, we can estimate
\begin{equation}
\label{Peteseger}
\begin{array}{rl}
\big(\CT_1(\la)u,u\big)\!\!\!\!&=
\left( \begin{pmatrix}
Q_{\ell-1} \CT_1(\la)Q_{\ell-1}& Q_{\ell-1} \CT_1(\la) Q_{\ell-1}^\perp
\\[2mm] 
\big(Q_{\ell-1} \CT_1(\la)Q_{\ell-1}^\perp\big)^*& Q_{\ell-1}^\perp\CT_1(\la) Q_{\ell-1}^\perp \end{pmatrix} \displaystyle{\binom{u_1}{u_2}, \binom{u_1}{u_2}} \right)\\[4mm]
&\ge 2 \Big( \big( Q_{\ell-1} \CT_1(\la)Q_{\ell-1} u_1,u_1 \big) + \big( Q_{\ell-1}^\perp\CT_1(\la) Q_{\ell-1}^\perp  u_2,u_2\big) \Big). 
\end{array} \hspace{-5mm}
\end{equation}
Since $A\ge \nu_1(A)$, $\la<c_{\ell-1}+\eps_{\ell-1}$, and $P_{\ell-1}=0$ if $\nu_1(A) > c_{\ell-1}+\eps_{\ell-1}$, we have 
\begin{align*}
P_{\ell-1}(A-\la)P_{\ell-1} &\ge \min \{(\nu_1(A)-(c_{\ell-1}+\eps_{\ell-1})),0\} P_{\ell-1}, \\ 
-\sum_{j=\ell}^L\dfrac{\Gamma_j}{c_j-\la} &\ge - \sum_{j=\ell}^L\dfrac{\|\Gamma_j\|}{c_j-(c_{\ell-1}+\eps_{\ell-1})}.
\end{align*}
Therefore, for $u_2 = Q_{\ell-1}^\perp u_2 \in \Ran\Gamma_{\ell-1}$,
\begin{align}
\nonumber
& \big( \CT_1(\la) u_2, u_2 \big)  \\[-1mm]
\label{22.2.}
& \ge - \max \{c_{\ell-1}\!+\!\eps_{\ell-1}\!-\!\nu_1(A),0\} \|P_{\ell-1} u_2\|^2 \!-\! \sum_{j=\ell}^L\dfrac{\|\Gamma_j\|}{c_j\!-\!(c_{\ell\!-\!1}+\eps_{\ell-1})}  \| u_2\|^2 \\[-2mm]
\nonumber
& \ge - \Big(  \max \{c_{\ell-1}\!+\!\eps_{\ell-1}\!-\!\nu_1(A),0\} \!+\! \sum_{j=\ell}^L\dfrac{\|\Gamma_j\|}{c_j\!-\!(c_{\ell-1}+\eps_{\ell-1})} \Big) \| u_2\|^2 \!=: - \gamma \|u_2\|^2 \hspace{-3mm}
\end{align}
where $\gamma\ge 0$. 
Because $\Gamma_{\ell-1}$ has closed range by assumption and $Q_{\ell-1}^\perp$ is the or\-tho\-gonal projection onto $(\ker \Gamma_{\ell-1})^\perp\!$, there exists $\gamma_{\ell-1} >0$ with $Q_{\ell-1}^\perp\Gamma_{\ell-1} Q_{\ell-1}^\perp \ge \gamma_{\ell-1}$ on $\Ran \Gamma_{\ell-1}$. 
Altogether, \eqref{ps1}, \eqref{Peteseger}, and \eqref{22.2.}, yield that 
\begin{align*}
  (\Si(\la) u,u) 
  &\ge 2 \big( Q_{\ell-1} \CT_1(\la)Q_{\ell-1} u_1,u_1 \big) - 2 \gamma \|u_2\|^2 + \frac{\gamma_{\ell-1}}{\la-c_{\ell-1}} \|u_2\|^2.
\end{align*}
Hence, if we choose $\la \in (c_{\ell-1},c_{\ell-1}+\eps_{\ell-1})$ such \vspace{-0.75mm} that
\begin{equation}
\label{march22}
 c_{\ell-1} < \la < c_{\ell-1} + \frac{\gamma_{\ell-1}}{2\gamma} \quad \text{if } \gamma>0,
\vspace{-1.75mm} 
 \end{equation}
we have the \vspace{-1mm} estimate 
\begin{align*}
  (\Si(\la) u,u) 
  &> 2 \big( Q_{\ell-1} \CT_1(\la)Q_{\ell-1} u_1,u_1 \big).
\end{align*}
Now the classical min-max variational principle for semi-bounded operators shows that, for $\la$ as in \eqref{march22}, the number $\kappa_\ell$ of negative eigenvalues of $\Si(\la)$ is less than or equal to the number of negative eigenvalues of $Q_{\ell-1} \CT_1(\la)Q_{\ell-1}$. Since $\kappa_\ell$ is independent of~$\la$, we may choose $\la=c_{\ell-1}+\eps$ with $\eps >0$ arbitrarily small and proceed in the same way as in \eqref{spring-new} to see that, in the limit $\eps\searrow 0$, the number of negative eigenvalues of $Q_{\ell-1} \CT_1(\la)Q_{\ell-1}$ is given by the right hand side of \eqref{iml}.

The last claim for $\ell=L+1$ follows since in this case $\CT_1(\la)=A-\la$ and the operator $Q_L P_L (A-\la) P_L Q_L$ is negative for $\la>c_L$ and thus has $\rank (P_LQ_L)$ negative eigenvalues.
\end{proof}

\begin{remark}
If there is just one pole, i.e.\ $L=1$, then the claims in Proposition \ref{pro1} coincide with those in Proposition \ref{new-est}; note that $\ker \Gamma_\ell = \ker B_\ell^*$.
\end{remark}

\vspace{-3.25mm}

\begin{center}
\begin{figure}[h]
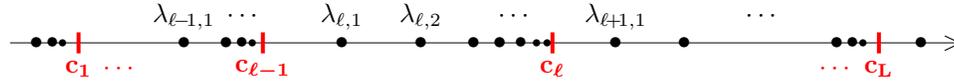

\psset{unit=7mm}\pspicture*(-10.3,-0.7)(7.7,1)
\psline[linewidth=.3pt,linecolor=black]{-}(-10.3,0)(0,0)
\psline[linewidth=.3pt,linecolor=black]{-}(0,0)(7.7,0)
\rput[C](7.6,0){$>$}
%
\rput[C](-9.8,0){$\bullet$}
\rput[C](-9.5,0){\small$\bullet$}
\rput[C](-9.3,0){\tiny$\bullet$}
\psline[linewidth=1.5pt,linecolor=red]{-}(-9,-0.2)(-9,0.2)
\rput[C](-9,-0.5){\ct{\boldmath{$\bf c_1$}}}
\rput[C](-8.2,-0.5){\ct{$\cdots$}}
\rput[C](-7,0){$\bullet$}
\rput[C](-7,0.5){$\la_{\ell\!-\!1,1}$}
\rput[C](-6.2,0){$\bullet$}
\rput[C](-5.9,0){\small$\bullet$}
\rput[C](-5.7,0){\tiny$\bullet$}
\rput[C](-5.85,0.5){$\cdots$}
\psline[linewidth=1.5pt,linecolor=red]{-}(-5.5,-0.2)(-5.5,0.2)
\rput[C](-5.5,-0.5){\ct{\boldmath{$\bf c_{\ell-1}$}}}
\psline[linewidth=1.5pt,linecolor=red]{-}(0,-0.2)(0,0.2)
\rput[C](0,-0.5){\ct{\boldmath{$\bf c_{\ell}$}}}
\rput[C](-4,0){$\bullet$}
\rput[C](-4,0.5){$\la_{\ell,1}$}
\rput[C](-2.5,0){$\bullet$}
\rput[C](-2.5,0.5){$\la_{\ell,2}$}
\rput[C](-1.5,0){$\bullet$}
\rput[C](-1,0){$\bullet$}
\rput[C](-0.6,0){\small$\bullet$}
\rput[C](-0.3,0){\tiny$\bullet$}
\rput[C](-0.1,0){\tiny$\bullet$}
\rput[C](-0.7,0.5){$\cdots$}
\rput[C](1.2,0){$\bullet$}
\rput[C](1.2,0.5){$\la_{\ell\!+\!1,1}$}
\rput[C](2.5,0){$\bullet$}
\rput[C](4,0.5){$\cdots$}
\rput[C](5.4,0){$\bullet$}
\rput[C](5.7,0){\small$\bullet$}
\rput[C](5.9,0){\tiny$\bullet$}
\psline[linewidth=1.5pt,linecolor=red]{-}(6.2,-0.2)(6.2,0.2)
\rput[C](6.2,-0.5){\ct{\boldmath{$\bf c_L$}}}
\rput[C](5.4,-0.5){\ct{$\cdots$}}
\rput[C](7,0){$\bullet$}
\endpspicture
\vspace{-2.1mm}
\caption{{\small Eigenvalues $(\la_{\ell,j})_{j=1}^{n_\ell} \!\subset\! (c_{\ell-1},c_{\ell})$ of $\CS$.}}
\end{figure}
\vspace{-4.25mm}
\end{center}

According to Proposition \ref{pro1}, the spectrum of $\CA$ and of $\CS$ in \emph{all} intervals $(c_{\ell-1},c_{\ell})$ may accumulate at most at the right end-points 
$c_{\ell}$ for $\ell=1,2,\dots,L+1$. We denote the corresponding sequences of eigenvalues, ordered non-decreasingly and counted with multiplicities, 
by $(\la_{\ell,j})_{j=1}^{n_\ell} \subset (c_{\ell-1},c_{\ell})$ with \vspace{-2.5mm} $n_\ell \in \N_0 \cup \{\infty\}$:
\begin{align*}
  &\sigma(\CA) \setminus \{c_1,c_2,\dots, c_L\} = \sigma_{\rm p}(\CA) \setminus \{c_1,c_2,\dots, c_L\} 
  = \sigma(\CS) = \sigma_{\rm p} (\CS)  
  = \bigcup_{\ell=1}^{L+1} ( \la_{\ell,j} )_{j=1}^{n_\ell}.
\vspace{-4.5mm}  
\end{align*}
Here $n_\ell=\infty$ means that the sequence $(\la_{\ell,j})_{j=1}^{n_\ell} \subset (c_{\ell-1},c_{\ell})$ is infinite and  accumulates at $c_{\ell}$. It is not difficult to see that
\begin{equation}
\label{inf-or-not-ell}
\begin{array}{ll}
  n_{\ell} < \infty & \mbox { if } \dim \Hii_{\ell} < \infty, \ \ell=1,2,\dots,L, \\
  n_{L+1} < \infty & \mbox{ if and only if } \dim \Hi < \infty.
\end{array}  
\end{equation}
Both claims are consequences of the fact that $A$ has compact resolvent.
Indeed, if $\dim \Hii_{\ell} < \infty$, then $c_{\ell} \notin \sigma_{\rm ess}(\CA)$ by Proposition \ref{prop:1.1};
the second claim follows because $A$, and hence $\CA$, is bounded if and only if $\dim \Hi < \infty$.

The condition $\dim \Hii_{\ell} = \infty$ is only necessary for eigenvalue accumulation at~$c_{\ell}$, i.e.\ for $n_{\ell}=\infty$.
%
The following two propositions show that $\dim \Hii_\ell = \infty$ and $B_\ell$ injective are sufficient conditions for $n_\ell=\infty$, 
in analogy to Propositions \ref{inf-at-c} and~\ref{prop:1.3a} where $L=1$.
In Theorem~\ref{th_j}~iii) below, we show that the weaker condition $\dim \ov{\Ran B_{\ell}} = \dim \ov{\Ran \Gamma_{\ell}} = \infty$ is both necessary and sufficient.



\begin{proposition}
\label{inf-at-c-ell}
Let $\,\H$ be a Hilbert space, $\H=\Hi\oplus \Hii_1 \oplus \Hii_2 \oplus \dots \oplus \Hii_L$ with Hilbert spaces $\Hi$,~$\Hii_1$, $\Hii_2$, \dots, $\Hii_L$. 
Let $A$ be a self-adjoint operator in $\Hi$ with com\-pact resolvent and bounded from below, 
$B_\ell: \Hii_\ell \to \Hi$, $\ell=1,2,\dots,L$, bounded linear operators, and
$c_1, c_2, \dots, c_L\in\R$ with $c_1 < c_2 < \dots < c_L$. 

If $\,\dim \Hii_\ell = \infty$ for $\ell\in \{1,2,\dots,L\}$ and $c$ is not an eigenvalue of infinite multiplicity of $\CA$, 
then $c_\ell$ is an accumulation point of eigenvalues 
of $\CA$, and hence of $\CS$, from the left but not from the right.
\end{proposition}

\pagebreak

\begin{proof}
The proof of Proposition \ref{inf-at-c-ell} is analogous to the proof of Proposition \ref{inf-at-c} if one replaces Proposition \ref{prop:1.1} iv) by Proposition \ref{pro1} ii).
\end{proof}

\begin{proposition}
\label{prop:4.2a}
Let $\,\H$ be a Hilbert space, $\H=\Hi\oplus \Hii_1 \oplus \Hii_2 \oplus \dots \oplus \Hii_L$ with Hilbert spaces $\Hi$,~$\Hii_1$, $\Hii_2$, \dots, $\Hii_L$. 
Let $A$ be a self-adjoint operator in $\Hi$ with com\-pact resolvent and bounded from below, 
$B_\ell: \Hii_\ell \to \Hi$, $\ell=1,2,\dots,L$, bounded linear operators, and
$c_1, c_2, \dots, c_L\in\R$ with $c_1 < c_2 < \dots < c_L$.

If $B_\ell$ is injective for $\ell\in \{1,2,\dots,L\}$, then $\dim \ker (\CA-c_\ell) \le N(A,c_\ell)$; in particular, $c_\ell$ is not an eigenvalue of infinite multiplicity of $\CA$.

\end{proposition}

\begin{proof}
The proof of Proposition \ref{prop:4.2a} is completely analogous to the proof of Proposition~\ref{prop:1.3a} 
and thus left to the reader.
\end{proof}


\begin{proposition}
\label{prop:4.2}
If, under the assumptions of Proposition {\rm \ref{prop:4.2a}}, 
$B_\ell$ is injective for $\ell\in \{1,2,\dots,L\}$ \vspace{-1mm} and 
\begin{equation}\label{ass1-m}
u \in \dom A  \cap \,\Ker B_\ell^*, \ u \ne 0 \ \implies \ (A-c_\ell)u \notin  \Ran \big(B_1 \ B_2 \,\dots\, B_L \big).
\end{equation}  
then $c_\ell\notin \sigma_{\rm p}(\CA)$. Conversely, if $c_\ell\notin \sigma_{\rm p}(\CA)$, then \eqref{ass1-m} holds.
\end{proposition}

\begin{proof}
The proof of Proposition \ref{prop:4.2} is completely analogous to the proof of Proposition~\ref{prop:1.3} 
and thus left to the reader.
\end{proof}

\subsection{Variational principles.} 
In this subsection we establish min-max variational principles for the eigenvalues of the operator function $\CS$ in \eqref{Meise} 
in the intervals $(c_{\ell-1},c_{\ell})$, $\ell=1,2\dots,{L+1}$. 

\begin{theorem}
\label{th_j} 
Let $\Hi$ be a Hilbert space, $A$ a self-adjoint operator in $\Hi$, bounded from below and with compact resolvent,  let $\Gamma_\ell$, $\ell=1,2,\dots,L$,
be non-zero bounded self-adjoint non-negative operators in $\Hi$, 
$c_1, c_2, \dots, c_L\in\R$ with $c_1<c_2<\dots<c_L$, and set $c_0=-\infty$, $c_{L+1}=+\infty$.

Then the operator function $\CS$ in \eqref{Meise} is strictly decreasing in $\R\setminus\{c_1,c_2,\dots, c_L\}$ and hence, 
for each $u \in \dom(A)=\dom (\CS(\la))$, the function $\la \mapsto (S(\la) u,u)$ has at most one zero in each interval
$(c_{\ell-1},c_{\ell})$, $\ell=1,2,\dots, L+1$. If we define $p_\ell(u) \in [c_{\ell-1},c_{\ell}]$ for $u \in \dom(A)=\dom (\CS(\la))$ by
\begin{equation}
\label{zeros}
  p_\ell(u):= \begin{cases} 
               \la_\ell(u) & \mbox{ if } \ (S(\la_\ell(u))u,u)=0 \ \mbox{ for } \ \la_\ell(u) \in (c_{\ell-1},c_{\ell}), \\
               \ \, c_{\ell-1} & \mbox{ if } \ (S(\la) u,u)<0 \ \mbox{ for all } \ \la \in (c_{\ell-1},c_{\ell}), \\
               \ c_{\ell} & \mbox{ if } \ (S(\la) u,u)>0 \ \mbox{ for all } \ \la \in (c_{\ell-1},c_{\ell}),
              \end{cases}
\end{equation}
then the spectrum of \,$\CS$ consists of $L+1$ $($finite or infinite$)$ eigenvalue sequences  $(\la_{\ell,j})_{j=1}^{n_\ell} \subset (c_{\ell-1},c_{\ell})$, $n_\ell \in \N_0\cup\{\infty\}$,  which may be characterized as
\begin{equation}
\label{minmax}
\la_{\ell,j} = \!\!\!\!\min_{\ontops{\CL \subset \dom A}{\dim \CL = j+\kappa_\ell}} \!\!\max_{\ontops{u \in \CL}{u\ne 0}} \ p_{\ell}(u),\quad j=1,2,\dots,n_\ell, \ \ell = 1,2,\dots, L+1.
\end{equation}
Here the index shifts $\kappa_\ell$ defined as in Proposition {\rm \ref{pro1}} are finite and satisfy the estimates therein; in particular, $\kam=0$ and $\kappa_{L+1} \le N(A,c_L)$. Moreover,
\begin{enumerate}
\item[\rm i)] 
if $\,\Gamma_L$ has closed range, $P_L$ is the orthogonal projection onto $\CL_{(-\infty,c_L]}(A)$ and $Q_L$ 
the orthogonal projection onto $\ker \Gamma_L$, then 
\[
\ka_{L+1}\le{\rm rank}\,(P_L Q_L) \le \min \big\{ N(A,c_L), \dim \ker \Gamma_L \big\};
\]
\item[{\rm ii)}] $n_{L+1}\!=\!\infty$ if and only if $\,\dim \Hi \!=\! \infty$; in this case, $(\la_{L+1,j})_{j=1}^{\infty} \!\subset\! (c_L,+\infty)$ accumu\-lates $($only$)$ to~$+\infty$;
\item[{\rm iii)}] $n_{\ell}\!=\!\infty$ if and only if $\,\dim \ov{\Ran \Gamma_\ell} \!=\! \infty$, $\ell=1,2,\dots,L$; 
in this case, $(\la_{\ell,j} )_{j=1}^{\infty} \!\subset\! (c_{\ell-1},c_\ell)$ accumulates $($only$)$ to $c_\ell$ from the left.
\end{enumerate}
\end{theorem}

\begin{proof}
According to Proposition~\ref{prop:2.1} i), we have $\sigma(\CS)\!=\!\sigma_{\rm p}(\CS) \!=\! \sigma_{\rm p}(\CA) \!\setminus\! \{c_1,c_2,..,c_L\}$. 

By Proposition \ref{pro1} i), the Schur complement $\Si$ satisfies Assumptions (i) to (iv) of \cite[Theorem~2.1]{MR2068432} on each interval $(c_{\ell-1}, c_{\ell})$, $\ell=1,2,\dots,L+1$. 
Hence the variational characterizations \eqref{minmax} follow from \cite[(2.9]{MR2068432} applied on $(c_{\ell-1}, c_{\ell})$, respectively; note that the values $-\infty$ and $+\infty$, respectively, of the functional $p$ in \cite[(2.3)]{MR2068432} can be replaced by the end-points $c_{\ell-1}$ and~$c_{\ell}$, respectively of the interval $(c_{\ell-1}, c_{\ell})$.

The claims for the index shifts $\kappa_\ell$ in $(c_{\ell-1}, c_{\ell})$ follow from Proposition \ref{pro1} ii).

Claim i) is the case $\ell=L+1$ in Proposition \ref{pro1} iii). Claim ii) was proved in \eqref{inf-or-not-ell};
the proof of claim iii) is analogous to the proof of Theorem \ref{thm:2.3a} iii).
\end{proof}
 
\begin{remark}
Define the monic operator \vspace{-1.4mm} polynomial  
\begin{equation}
\label{defP}
 \CP(\la):=(-1)^{L-1} \bigg(\displaystyle \prod_{\ell=1}^L (c_\ell-\la) \bigg)\Si(\la), \quad \dom \CP(\la) = \dom \CS(\la) = \dom A,
\vspace{-1.4mm} 
\end{equation}
of degree $L+1$. One can show that $\CP$ is \emph{weakly hyperbolic} (comp.\ e.g.\ \cite{MR0344927,MR0350477}, \cite[\S~31]{MR971506} 
for operator polynomials with bounded operator values); in particular, for every $u\in \dom A$, the function $\la \mapsto (\CP(\la)u,u)$ has $L+1$ real zeros 
which coincide with the numbers $p_1(u) \le p_2(u) \le \dots p_{L+1}(u)$ defined in~\eqref{zeros}.
\end{remark}

The two-sided eigenvalue estimates for the one pole case in Subsection \ref{subsec:2.3} were based on the solution formula for quadratic equations and  
do not readily generalize to the case of several poles. However, under stronger assumptions on the operators~$B_\ell$, we can estimate the gap $(c_\ell,c_\ell+\varepsilon_\ell)$ in the spectrum of $\CS$ in Proposition~\ref{pro1}.

\begin{proposition}
\label{prop:4.1}
Let $\,\H$ be a Hilbert space, $\H=\Hi\oplus \Hii_1 \oplus \Hii_2 \oplus \dots \oplus \Hii_L$ with Hilbert spaces $\Hi$,~$\Hii_1$, $\Hii_2$, \dots, $\Hii_L$. Let $A$ be a self-adjoint operator in $\Hi$ with compact resolvent and bounded from below, $A\ge a$ for some $a\in\R$, 
$B_\ell: \Hii_\ell \to \Hi$  bounded linear operators with the same closed range, $\Ran B_\ell =: \Hi_0 \subset \Hi$, $\ell=1,2,\dots,L$, and $c_1,c_2,\dots, c_L\in\R$ with $c_1 < c_2< \dots < c_L$. 
Then there exist $\ga_\ell^\pm>0$ such that 
\begin{equation}\label{ppm-new}
  \ga_\ell^- \| u_0 \|^2 \le \|B_\ell^* u\|_\ell^2 \le \ga_\ell^+ \|u_0 \|^2,\quad u=u_0+u_0' \in \Hi_0 \oplus \Hi_0^\perp, \ \ \ell=1,2,\dots,L.
\end{equation}
Furthermore, for $\ell\in\{1,2,\dots,L\}$ the following hold: 
\noindent\begin{enumerate}
\item[{\rm i)}]
If $a>c_{\ell+1}$ and $\eta_\ell$ is the unique zero of the \vspace{-1mm}function
\[
 \psi_\ell(\la) := \sum_{k=1}^\ell\dfrac{\ga_k^-}{c_k-\la}+\sum_{k={\ell+1}}^L\dfrac{\ga_k^+}{c_k-\la}, \quad \la \in \R\setminus \{c_1,c_2,\dots,c_L\},
\]
in  $(c_{\ell},c_{\ell+1})$, then $(c_\ell,\eta_\ell)\subset\rho(\CS)$.
\item[{\rm ii)}]
If $\Hi_0=\Hi$ and $\eta_\ell'$ is the unique zero of the function $\la \mapsto a-\la-\psi_\ell(\la)$ in $(c_{\ell},c_{\ell+1})$, then  $(c_\ell,\eta_\ell')\subset\rho(\CS)$.
\end{enumerate}
\end{proposition}

%
%
%
\begin{proof}
Let $\ell\in\{1,2,\dots,L\}$. By definition of $\Hi_0$, the operator $B_\ell:\Hii_\ell \to \Hi_0$ is surjective. If we decompose $\Hi=\Hi_0 \oplus \Hi_0^\perp$ with 
$\Hi_0= \Ran B_\ell$, $\Hi_0^\perp = (\Ran B_\ell)^\perp = \ker B_\ell^*$, then $B_\ell^* u = B_\ell^*u_0$  for $u=u_0+u_0' \in \Hi_0 \oplus \Hi_0^\perp$
and $B_\ell^*|_{\Hi_0}: \Hi_0\to \Hii_\ell$ is injective. Further, since $\Ran B_\ell \!=\! \Hi_0$ is closed by assumption, so is $\Ran (B_\ell^*|_{\Hi_0})$ by the Closed Range theorem. 
Hence, by the Closed Graph theorem, $(B_\ell^*|_{\Hi_0})^{-1}$ is bounded on $\Ran (B_\ell^*|_{\Hi_0})$, and the existence of the numbers $\ga_\ell^\pm$ in \eqref{ppm-new} follows.

For $u\in\dom A$, $u\ne 0$, we consider the \vspace{-1mm}function
\[
\varphi_u(\la):=(\Si(\la)u,u)=((A-\la)u,u) -\sum_{\ell=1}^L\dfrac{\|B_\ell^*u\|_\ell^2}{c_\ell-\la}, \quad \la \in \R\setminus\{c_1,c_2,\dots,c_L\},
\vspace{-1mm}
\]
where $\|\cdot\|_\ell$ denotes the norm in $\Hii_\ell$. The claims i) and ii) follow if we show that in a right neighbourhood of $c_\ell$ all functions $\varphi_u(\la),\,u\in\dom A$, 
are uniformly positive. For $\la\in (c_{\ell},c_{\ell+1})$ we can \vspace{-1.5mm} estimate 
\begin{align*}
\varphi_u(\la)
&=((A-\la) u,u)-\sum_{k=1}^\ell\dfrac{\|B_k^*u\|_k^2}{c_k-\la}-\sum_{k={\ell+1}}^L\dfrac{\|B_k^*u\|_k^2}{c_k-\la}\\[-2mm]
&\ge ((A-\la)u,u)-\bigg(\sum_{k=1}^\ell\dfrac{\ga_k^-}{c_k-\la}+\sum_{k={\ell+1}}^L\dfrac{\ga_k^+}{c_k-\la}\bigg)\|u_0\|^2\\
&= ((A-\la)u,u)-\psi_\ell(\la)\|u_0\|^2.
\end{align*}
Since 
$\psi_\ell$ is strictly increasing in $(c_{\ell},c_{\ell+1})$ from $-\infty$ to $+\infty$, it has pre\-cisely one zero $\eta_\ell\!\in\!(c_{\ell},c_{\ell+1})$ and it is negative on $(c_\ell, \eta_\ell)$.
This implies that, in \vspace{-0.5mm} case~i),
\[
\varphi_u(\la)  \ge (a-c_{\ell+1}) \|u\|^2, \quad \la \in (c_\ell,\eta_\ell),
\]
and, in case ii), with $\eta_\ell'$ as defined \vspace{-0.5mm} there,  
\[
\big(a-\la-\psi_\ell(\la)\big)\|u\|^2>0\quad \la \in (c_\ell,\eta_\ell').
\qedhere
\]
\end{proof}

\vspace{2mm}

%% file: Part_Photonic_Crystals_v2-ccc.tex
\newcommand\Ccheck[1]{\marginpar{\cg{\emph{#1}}}}

\section{\bf Application to photonic crystals}
\label{Application}

In this section, we consider electromagnetic waves $(E,H)$ propagating in a nonmagnetic medium with relative permittivity $\epsilon$. The permittivity (or dielectric) function $\epsilon$ depends on the spatial coordinates $x_1$, $x_2$ as well as on the frequency~$\omega$, but not on the spatial coordinate $x_3$. The electromagnetic wave $(E,H)$ is then decomposed into transverse electric (TE) polarized waves $(E_1, E_2,0,0,0,H_3)$  and transverse magnetic (TM) polarized waves $(0,0,E_3, H_1,H_2,0)$ \cite[Chapter 1]{MR1036731}, \cite{MR1417473}. This decomposition reduces Maxwell's equations to scalar equations for $H_3$ and $E_3$, respectively. 

The discussion below focuses on the TM case. We will apply the operator theoretic results developed in the preceding sections to a physically relevant spectral problem with periodic permittivity function. This operator formulation of the photonics crystal problem with $\lambda$-dependent material properties is used as a base for numerical approximations in \vspace{-2mm} Section~\ref{sec:Galerkin}.


\subsection{The physical problem.}

Let $\Delta:=\dfrac{\partial^2}{\partial x_1^2}+\dfrac{\partial^2}{\partial x_2^2}$ denote the formal \vspace{-2mm} Laplace operator in $\R^2$.
The equation
\begin{equation}\label{eq:basTM}
	\CL E_3:=-\Delta E_3-\omega^{2}\epsilon(x,\omega) E_3=0, \quad x:=(x_1,x_2)\in\mathbb{R}^{2},
\end{equation}
models transverse magnetic (TM) polarized waves with frequency $\omega$. The permittivity function $\epsilon$ will only depend on $\omega^2\in\mathcal D \subset \C$ and we therefore define $\lambda:=\omega^2$. Let $\Gamma$ denote the lattice $\mathbb{Z}^{2}$ and $\Omega:=(0,1]^{2}$ the unit cell of the lattice $\Gamma$. The dual lattice to $\Gamma$ is $\Gamma^{*}:=2\pi\Z^2$ and we define the fundamental domain (the Brillouin zone) of the dual lattice $\Gamma^{*}$ as the set $\CK:=(-\pi,\pi]^{2}$. 
In this paper, we consider Bloch solutions of (\ref{eq:basTM}), i.e.~non-zero solutions of the form $E_3(x)=\eu^{\iu \langle k, x \rangle}u(x)$,  $x\in\mathbb{R}^{2}$, where $u$ is a $\Gamma$-periodic function, $k\in\CK$, and $u$ is a function of the variable $x$ on the torus $\T^{2}=\R^2/\Z^2$ \cite[p.\ 104]{MR1232660}. 
Here $\left <\cdot,\cdot \right >$ and $|\cdot |$ denote the Euclidean inner product and norm in $\R^2$, respectively.
Further, let $\nabla$ denote the gradient with respect to the space variable $x\in\R^2$, $\left <v,\nabla \right >$ the directional derivative in the direction $v\in\R^2$, $v\ne 0$, and $H^{s}(\T^{2})$ the Sobolev space of order $2$ associated with $\Ltwo$. Note that
functions $u \in H^{s}(\T^{2})$, $s>0$, can be characterized in terms of their Fourier series with coefficients  \vspace{-1mm}  $\hat u(n)\in\C$, $n\in\Z^2$, 
\[
u\in H^{s}(\T^{2}) \iff
\sum_{n\in\Z^2}(1+|2\pi n|)^{2s}|\hat u(n)|^2<\infty.
 \vspace{-1.5mm} 
\]

For fixed $k\in\CK=(-\pi,\pi]^2$ the shifted Laplace operator $\Delta_k:\Ltwo\rightarrow\Ltwo$ is defined as 
\begin{equation}
	\Delta_k:=\left <\nabla+\iu k,\nabla+\iu k\right >=\Delta+2\iu \left <k,\nabla \right >-|k|^2,
	\quad \dom (\Delta_k)=H^{2}(\T^2).
\end{equation}
The operator $\Delta_k$ has compact resolvent and $\sigma (-\Delta_k)=\{|2\pi n+k|^2\,:\,n\in \Z^2\}$, \cite[p.\ 161-164]{MR1232660}. Since $\nabla (\eu^{\iu \langle k, x \rangle}u(x))=\eu^{\iu \langle k, x \rangle}(\nabla+\iu k) u(x)$, the Bloch solutions $E_3$ are formally determined by the solutions of $\CT_k(\lambda)u=0$ with
\begin{equation}
	\CT_k(\lambda):=-\Delta_k-\lambda\epsilon(\cdot,\lambda),\quad k\in\CK=(-\pi,\pi]^2,\quad \lambda\in\mathcal{D}\subset \C,
	\label{eq:bas}
\end{equation}
over the torus $\T^2$. The solutions $u$ and eigenvalues $\lambda$ depend on $k\in\CK$ but we will not write this out explicitly. A given $\lambda=\omega^2$ is called a band gap frequency if equation \eqref{eq:basTM}, independent of the parameter $k\in\CK$, has no non-zero Bloch solution.



General analytic properties of the permittivity function $\epsilon(x,\cdot)$  are discussed in \cite[Chapter 1]{MR1409140}, while \cite{MR2718134} considers basic spectral properties of bounded holomorphic operator functions with applications to photonic crystals. 

Here, we apply the theory developed in the previous sections to a rational operator function with periodic permittivity $\epsilon$. Let $\Omega  =\Omega_1 \dot\cup \cdots \dot\cup \Omega_M$, $M\in\N$, denote a partitioning of $\Omega= (0,1]^2$ and let $\chi_{\Omega_m}$ denote the characteristic function (indicator function) of the subsets $\Omega_m \subset (0,1]^2$, $m=1,2,\dots,M$. In most applications the function $\epsilon$ is  piecewise constant in $x$, 
\begin{equation}\label{eq:epsSum}
	\epsilon(x, \lambda) := \sum_{m=1}^M \epsilon_m(\lambda)\chi_{\Omega_m} (x), \quad x\in\Omega,\ \lambda\in\mathcal{D}\subset \C.
\end{equation}
The $\lambda$-independent case $\epsilon_m:= \aaa_m$ leads to the extensively studied linear problem \cite{MR1417473,MR1232660}.  However, a common model for solid materials with dispersion is the Lorentz model
 \begin{equation}\label{eq:epsilonAnsatz}
	\epsilon_m(\lambda) := \aaa_m + \sum_{\ell = 1}^{L_m} \frac{\bb_{m,\ell}}{\cc_{m, \ell} - \lambda}, 
	\quad \la \in \C \setminus \{ \cc_{m,l}: \ell=1,2,\dots,L_m\},
\end{equation}
with $L_m\in \N$ and positive
constants $\aaa_m$, $\bb_{m,\ell}$, and $\cc_{m, \ell}$ \cite{MR1409140}. Here 
the permittivity function in \eqref{eq:epsSum} takes the form
\begin{equation}\label{eq:epsModel}
	\epsilon(\cdot, \lambda) := \sum_{m=1}^M a_m \chi_{\Omega_m} (\cdot)+\sum_{m=1}^{\widehat{M}} \sum_{\ell = 1}^{L_m} \frac{\bb_{m,\ell}}{\cc_{m, \ell} - \lambda} \chi_{\Omega_m} (\cdot)
\end{equation}
with $\wh M \in \{1,2,\dots,M\}$ and $\epsilon_m:= \aaa_m$, $m=\wh{M}+1,\wh{M}+2,\dots,M$. By a polynomial long division of $\lambda\epsilon (\cdot,\lambda)$ with the Lorentz model  \eqref{eq:epsModel}, we obtain  
\begin{equation}\label{eq:rationalTerm}
	\lambda\epsilon (\cdot,\lambda) = \lambda\sum_{m=1}^{M}a_m \chi_{\Omega_m} (\cdot)-\sum_{m=1}^{\wh{M}}\sum_{\ell=1}^{L_m}\bb_{m,\ell}\chi_{\Omega_m} (\cdot)+\sum_{m=1}^{\wh{M}}\sum_{\ell=1}^{L_m}\frac{\cc_{m,\ell}\bb_{m,\ell}}{\cc_{m,\ell}-\lambda}\chi_{\Omega_m} (\cdot)
\end{equation}
for $\la \in \C\setminus \{c_{m,\ell}: \ell=1,2,\dots, L_m, m=1,2,\dots, \wh M\}$. The operator
\begin{equation}
\label{W}
  W:\Ltwo\rightarrow\Ltwo, \quad W:=\sum_{m=1}^{M}a_m \chi_{\Omega_m},
\end{equation}
is bounded, self-adjoint, and bijective since $a_m\!>\!0$. Moreover, $\sum_{m=1}^{M}\chi_{\Omega_m}\!=\! I_{\Ltwo}$ is the identity operator in $\Ltwo$ and $(W^{*})^{-\frac 12}=W^{-\frac 12}=\sum_{m=1}^{M}a_m^{-\frac 12} \chi_{\Omega_m}$. 

For $\lambda\in\C \setminus \{\cc_{m,\ell}:  \ell=1,2,\dots, L_m, m=1,2,\dots, \wh M \}$, we define the operator $\CS_k(\lambda):\Ltwo\rightarrow\Ltwo$ \vspace{-3.5mm}by
\begin{equation}\label{eq:Sk}
\begin{aligned}
	\CS_k(\lambda) 	& :=W^{-\frac 12}\CT_k(\lambda)W^{-\frac 12}
					= A_k-\lambda-\!\sum_{m=1}^{\wh{M}}\sum_{\ell=1}^{L_m} (\cc_{m,\ell}-\lambda)^{-1}\frac{\cc_{m,\ell}\bb_{m,\ell}}{a_m}\chi_{\Omega_m}, \hspace{-4mm}
\\[-2mm]
\end{aligned} 
\end{equation}
\vspace{-3.5mm} where 
\begin{equation}\label{eq:A0+A1}
  A_k\!:=\!A_k^{(0)}+A^{(1)}\!, \quad 
	A_k^{(0)}\!\!:=\!-W^{-\frac 12}\Delta_k W^{-\frac 12}\!,\quad A^{(1)}\!\!:=\!\sum_{m=1}^{\wh{M}}\frac{1}{a_m}\sum_{\ell=1}^{L_m}\bb_{m,\ell}\chi_{\Omega_m},
\vspace{-3.5mm}
\end{equation}
with domains 
\[
  \dom \CS_k(\lambda)= \dom A_k = \dom A_k^{(0)} = W^{\frac 12} \dom \Delta_k= W^{\frac 12}H^{2}(\T^2). 
\]  
The operator $A_k^{(0)}$ is self-adjoint with discrete spectrum, while $A^{(1)}$ is self-adjoint and bounded.  
Hence, $A_k$ is self-adjoint and the spectrum $\sigma (A_k)$ is discrete as well.

\subsection{Block operator matrix formulation.}

In this subsection, we consider the minimal linearization of the rational operator function
$\CS_k$ in \eqref{eq:Sk}. Here the characteristic function  $\chi_{\Omega_m}$ will be viewed as a multiplication operator $P_m\!=\!\chi_{\Omega_m} \!\cdot $ bet\-ween $\Ltwo$ and its~range 
\begin{equation}
	\Ran P_m=\left \{u\in\Ltwo\,:\, u|_{\Omega \setminus\Omega_m} \equiv 0\right \}.
\end{equation}
Then $\CS_k$ can be written in the form
\begin{equation}
	\CS_k(\lambda)= A_k-\lambda-\!\sum_{m=1}^{\wh{M}}\sum_{\ell=1}^{L_m} B_{m,\ell}(\cc_{m,\ell}-\lambda)^{-1}B_{m,\ell}^{*},\quad B^{*}_{m,\ell}:=\sqrt{\frac{\cc_{m,\ell}\bb_{m,\ell}}{a_m}}P_m, \hspace{-4mm}
\end{equation}	
with $A_k$ as 
in \eqref{eq:A0+A1}. 
The theory in the previous sections assumes that the poles of $\CS_k$ are disjoint and ordered increasingly. Therefore, we denote by $L$ the number of disjoint poles of $\CS_k$,
\[
	L:=\#\left\{c_{m,\ell}:m=1,\dots,\wh{M},\ell=1,\dots,L_m\right\}, 
\]
and we denote these disjoint poles by $c_1<c_2<\dots<c_L$. For $i\!=\!1,2,\dots,L$, the~set
\[
	M_i\!:=\!\left\{m\in\{1,\dots,\wh{M}\}: \exists\, \ell=\ell^i(m)\in \{1,2,\dots,L_m\}, c_{m,\ell}=c_i \right\}=:\left\{\mu_1^{i},\dots,\mu_{m_{i}}^{i}\right\},
\]
consists of all indices $m$ such that the permittivity $\epsilon_m$ in \eqref{eq:epsilonAnsatz} has a pole at $c_i$. 
Let
\[
  \Hi:= \Ltwo, \quad  \Hii:=\Hii_1 \oplus \Hii_2\cdots \oplus \Hii_L, \quad
  \Hii_i:=\!\bigoplus_{m\in M_i}\!\Ran P_m,
\]   
and define the operators $B\!:\Hii_1 \oplus  \Hii_2 \oplus\dots \oplus \Hii_L \to \Ltwo$ and $C:\Hii\to\Hii$ by
\begin{align}\label{eq:B}
	B\! &:=\!\big((B_{m,\ell^{1}(m)})_{m\in M_1}\, (B_{m,\ell^{2}(m)})_{m\in M_2}\, \dots \,(B_{m,\ell^{L}(m)})_{m\in M_L} \big),\\
\label{eq:C}	
	C\! &:={\rm diag\,} \big( c_1 I_{\Hii_1} \,c_2 I_{\Hii_2} \,\dots \,c_L I_{\Hii_L} \big),
\end{align}
where
\begin{equation}
	B_i:=(B_{m,\ell^{i}(m)})_{m\in M_i}:=\left ( B_{\mu_1^{i},\ell^i (\mu_1^{i})}\,\, B_{\mu_2^{i},\ell^i (\mu_2^{i})}\,\dots B_{\mu_{m_i}^{i},\ell^i (\mu_{m_i}^{i})} \right ).
\end{equation}
Then $\CS_k$ is the first Schur complement of the block operator matrix $\CA_k$ in $\H:=\Ltwo\oplus\Hii$ given by
\begin{equation}
\label{eq:Ak}
\CA_k\!=\!
\matrix{cc}{\!A_k &\! B \\\! B^* &\! C },  \quad \dom \CA_k \!=\! \dom A_k \oplus \Hii.
\end{equation}
Hence, the entries in the block operator matrix $\CA_k$ are
\[
\CA_k\!=\!
\matrix{c|cccc}{A_k&(B_{m,\ell^{1}(m)})_{m\in M_1}\!\!&\!\!(B_{m,\ell^{2}(m)})_{m\in M_2}\!\!&\!\!\cdots\!\!&\!\!(B_{m,\ell^{L}(m)})_{m\in M_L}\\ \hline
(B_{m,\ell^{1}(m)}^*)_{m\in M_1}&c_1I_{\Hii_1}\!\!&\!\!0\!\!&\!\!\cdots\!\!&\!\!0\\
(B_{m,\ell^{2}(m)}^*)_{m\in M_2}&0\!\!&\!\!c_2 I_{\Hii_2}\!\!&\!\!\cdots\!\!&\!\!0\\
\vdots&\vdots\!\!&\!\!\vdots\!\!&\!\!\ddots\!\!&\!\!\vdots\\
(B_{m,\ell^{L}(m)}^*)_{m\in M_L}&0\!\!&\!\!0\!\!&\!\!\cdots\!\!&\!\!c_L I_{\Hii_L}\ }\!.
\]

\medskip
\noindent
\textbf{Example}. Let $\Omega=\Omega_1\dot\cup \Omega_2\dot\cup \Omega_3$, i.e.\ $M=3$, and consider a Lorentz model
\begin{equation}
\begin{aligned}
	\epsilon(\cdot, \lambda) 	:= &a_1 \chi_{\Omega_1}  (\cdot)+a_2 \chi_{\Omega_2}  (\cdot)+a_3 \chi_{\Omega_3}  (\cdot)+\\
							& \frac{\bb_{1,1}}{\cc_{1, 1} - \lambda} \chi_{\Omega_1} (\cdot)+\frac{\bb_{1,2}}{\cc_{1, 2} - \lambda} \chi_{\Omega_1} (\cdot)+\frac{\bb_{2,1}}{\cc_{2, 1} - \lambda} \chi_{\Omega_2}  (\cdot)
\end{aligned}
\end{equation}
with $c_{1,2} < c_{1,1} = c_{2,1}$. Here $\wh M=2$, $L=2$, $c_1=c_{1,2}$, $c_2=c_{1,1} = c_{2,1}$, $M_1=\{1\}$, 
\linebreak
$\ell^1 (1)=2$, $M_2=\{1,2\}$, $\ell^2 (1)=1$, $\ell^2 (2)=1$, $(B_{m,\ell^{1}(m)})_{m\in M_1}=B_{1,2}$, and 
$(B_{m,\ell^{2}(m)})_{m\in M_2}=(B_{1,1}\,B_{2,1})$. Hence, the operator function $\CS_k$ is the first Schur complement of the block operator matrix
\begin{align*}
\CA_k\!&=\!
\left( \begin{array}{c|c:cc}
A_k& B_{1,2}\!&\!B_{1,1}\!\!&\!B_{2,1}\\ \hline
B^*_{1,2}&c_1\!\!&\!\!0\!\!&\!\!0\\ \hdashline
B^*_{1,1}&0\!\!&\!\!c_2\!\!&\!0\\
B^*_{2,1}&0\!\!&\!\!0\!&\!\!c_2\
\end{array} \right), \quad \dom \CA_k = \dom A_k \oplus (\Hii_1 \!\oplus\! \Hii_2 ),
\end{align*}
in $\Hi \!\oplus\! (\Hii_1 \!\oplus\! \Hii_2 ) \!=\! \Ltwo \!\oplus\! ( \Ran \chi_{\Omega_1} \!\oplus\! (\Ran \chi_{\Omega_1} \!\oplus\! \Ran \chi_{\Omega_2}) )$.

\subsection{Spectral properties}
In this subsection, we analyze the accumulation of eigenvalues at the poles of the permittivity function by means of the results of Section \ref{sec:2}.

Let $\big\{\nu_j(A_k^{(0)})\big\}_{j=1}^\infty$, $\big\{\nu_j(A^{(1)})\big\}_{j=1}^\infty$ denote the eigenvalues of $A_k^{(0)}$, $A^{(1)}$ in \eqref{eq:A0+A1}, respectively. Since $a_{m,\ell}$, $b_{m,\ell}> 0$ and $\sum_{m=1}^{\widehat M} P_m \le I_{\Ltwo}$, we \vspace{-2mm} have
\begin{align*}
  0 \le A^{(1)} & \le \| A^{(1)} \|=  \sup_{\| u\|=1}  \sum_{m=1}^{\widehat M} \Big(\frac 1{a_m} \sum_{\ell=1}^{L_m} b_{m,\ell} P_m u,u \Big) \\
   & \le \max_{m=1}^{\widehat M} \left( \frac 1{a_m}\sum_{\ell=1}^{L_m} b_{m,\ell} \right) \sup_{\| u\|=1}  \sum_{m=1}^{\widehat M} (P_m u,u)\\
    & \le \max_{m=1}^{\widehat M} \left( \frac 1{a_m}\sum_{\ell=1}^{L_m} b_{m,\ell} \right).  
\end{align*}
Hence, by means of the classical min-max variational principle \cite[Theorem XIII.1]{MR0493421}, 
the eigenvalues of the operator $A_k$ defined in \eqref{eq:A0+A1} can be estimated by
\begin{equation}\label{eq:estAk}
	|k|^2 \le\nu_j(A_k^{(0)})\leq \nu_j(A_k) 
	\leq\nu_j(A_k^{(0)})+\max_{m=1}^{\widehat M} \left( \frac 1{a_m}
\sum_{\ell=1}^{L_m} b_{m,\ell} \right).
\end{equation}
In order to bound the norm of $B$ defined in \eqref{eq:B}, we use the inequality
\[  
	\|(B_{m,\ell^{i}(m)})_{m\in M_i}\|\le \max_{m\in M_i} \sqrt{\frac{\cc_{m,\ell (m)}\bb_{m,\ell (m)}}{a_m}}\sum_{m\in M_i} \| P_m \|, 
\]
and $\sum_{m\in M_i} \| P_m \|\le 1$ to obtain
\begin{equation}\label{eq:estB}
	\| B\|\le \sum_{i=1}^{L}\|(B_{m,\ell^{i}(m)})_{m\in M_i} \|\le  \sum_{i=1}^{L}\max_{m\in M_i} \sqrt{\frac{\cc_{m,\ell (m)}\bb_{m,\ell (m)}}{a_m}}.
\end{equation}
Clearly, dim $\,\Hii_{m}=\infty$ and the operators $P_m^{*}:\Ran P_m \to \Ltwo$ are injective. 
Take $i \in \{1,2,\dots,L\}$ arbitrary and assume $u_i=(u_m)_{m\in M_i} \in \ker B_i$, 
i.e.
\[
0=\big(B_{m,\ell^{i}(m)}\big)_{m\in M_i}u_i= \sum_{m\in M_i}  B_{m,\ell^{i}(m)} u_m =  \sum_{m\in M_i} \sqrt{\frac{\cc_{m,\ell (m)}\bb_{m,\ell (m)}}{a_m}} P_m^* u_m. 
\]
Then $P_m^* u_m=0$ since the coefficients of $P_m^* u_m$ are positive, the support of $P_m^* u_m$ is contained in $\Omega_m$, 
and the subsets $\Omega_m$ are pairwise disjoint. Hence $u_m=0$ for all $m\in M_i$, i.e.\ $u_i=0$, which implies that $B_i$, $i=1,2,\dots, L$, is injective. 

Now it follows from Propositions \ref{inf-at-c-ell} and \ref{prop:4.2a} that all poles of the permittivity function $\epsilon$ in \eqref{eq:epsModel}, i.e.\ all points 
$c_1<c_2<\dots<c_L$, are accumulation points of eigenvalues of $\CA_k$. Moreover, Proposition \ref{pro1} implies that, for fixed $k$, there are gaps in the spectrum of $\CA_k$ to the right of all $\cc_{\ell}$. Note that only $A_k^{(0)}$ depends on $k$ and the dependence is analytic.

In general, we will not have a band gap above the poles for all material models \eqref{eq:epsModel} since the width of the gaps may tend to $0$ if $k$ varies, i.e.\ we may have $\la_{\ell+1,1}(k)\to c_{\ell}$, $k\rightarrow k_0$, for some $k_0\in (-\pi,\pi)^{2}$ and $\ell \in \{1,2,\dots,L\}$.  
An example for this 
is given in Section \ref{sec:Galerkin} where Figure \ref{Fig:BandTwoProjections} shows that numerically we do not find a band gap above $c_1=c_{1,1}$. 

However, Proposition \ref{prop:4.1} provides a  concrete estimate for a gap above a pole $c_\ell$ in the case when $A_k>c_{\ell+1}$ and all 
operators $B_{m,\ell}$, $\ell=1,2,\dots,L$, have the same closed range. Since 
the range of $B_{m,\ell}$ equals the range of $P_m^{*}$, which is
\begin{equation}
	\Ran B_{m,\ell}= \Ran P_m^{*}=\left \{u\in\Ltwo\,:\, u|_{\Omega_m} \equiv 0\right \}^{\bot},
\end{equation}
operators $B_{m,\ell}$, $B_{n,\ell}$ with $m\neq n$ do not have the same range. Thus Proposition~\ref{prop:4.1} only applies to permittivity functions of the form
\begin{equation}\label{eq:OneP}
	\epsilon(\cdot, \lambda) := \sum_{m=1}^M a_m \chi_{\Omega_m} (\cdot)+\sum_{\ell = 1}^{L_1} \frac{\bb_{1,\ell}}{\cc_{1, \ell} - \lambda} \chi_{\Omega_1} (\cdot)
\end{equation}
and to poles $c_{1,\ell}$ such that $A_k > c_{1,\ell+1}$. 
Since $\epsilon(\cdot,\la)$ and hence the operators $B_{m,\ell}$, do not depend on $k$, Proposition \ref{prop:4.1} i) shows that 
if $A_k > c_{1,\ell+1}$, $k \in \CK=(-\pi,\pi]^2$, 
then there exists an $\eta_\ell > c_\ell$, independent of $k$, such that 
\begin{equation}
\label{band-gap}
  (c_{\ell},\eta_{\ell}) \subset \rho(\CA_k), \quad k \in \CK=(-\pi,\pi]^2,
\end{equation}
which means that in this case our abstract results guarantee a band gap above $c_{\ell}$. 

\medskip
\noindent
\emph{Explicit bounds on the band gap above $c_{\ell}$}.
In this subsection, we consider a case where an estimate of the band gap in \eqref{band-gap} can be derived explicitly. Assume the 
permittivity function has the form \eqref{eq:OneP} with $M=2$,  $L_1=2$, and set $c_1:=c_{1,1}<c_{1,2}=:c_2$.

Let $u=u_0+u'_0\in\Hii_1 \oplus \Hii_1^\perp$. By the definition of $B_{1,\ell}^*$, 
\begin{equation}
	 \|B_{1,\ell}^* u\|^2= \frac{\bb_{1,\ell}\cc_{1,\ell}}{\aaa_1} \|u_0\|^2,\quad \ell=1,2,
\end{equation}
which implies that the constants in Proposition \ref{prop:4.1} are
\begin{equation}
	\gamma_1^-=\frac{\bb_{1,1}\cc_{1,1}}{\aaa_1},\quad \gamma_2^+=\frac{\bb_{1,2}\cc_{1,2}}{\aaa_1}.
\end{equation}
According to  Proposition \ref{prop:4.1}, we let
\begin{equation}
	\psi_2(\lambda)=\frac{\gamma_1^-}{c_{1,1}-\lambda}+\frac{\gamma_2^+}{c_{1,2}-\lambda}, \quad \la \in \R\setminus\{c_{1,1}, c_{1,2}\}.
\end{equation}
It is easy to see that the unique zero $\eta_1$ of $\psi_2$ in the interval $(c_{1,1},c_{1,2})$ is given by
\begin{equation}
	\eta_1=\cc_{1,1}+\frac{\cc_{1,1}\bb_{1,1}}{\cc_{1,2}\bb_{1,2}+\cc_{1,1}\bb_{1,1}}(\cc_{1,2}-\cc_{1,1}).
\end{equation}
Hence, for $c_{1,2}$ such that $A_k>c_{1,2}$ for  all $k \in \CK=(-\pi,\pi]^2$, Proposition \ref{prop:4.1} ii) implies that
we have a band gap above $c_{1,1}$, namely
\begin{equation}\label{PC:TwoPoles}
	\left ( \cc_{1,1},\cc_{1,1}+\frac{\cc_{1,1}\bb_{1,1}}{\cc_{1,2}\bb_{1,2}+\cc_{1,1}\bb_{1,1}}(\cc_{1,2}-\cc_{1,1}) \right ) \subset\rho(\CA_k), 
	\quad k \in \CK=(-\pi,\pi]^2.
\end{equation}
The accumulation of eigenvalues at the poles of the permittivity function from the left and two-sided estimates for them will be 
discussed and illustrated in Section~\ref{sec:Galerkin}. 

\subsection{Space independent permittivity}

The case $\epsilon(x,\lambda)=\epsilon (\lambda)$, where the eigenvalues of $\CA_k$ can be calculated explicitly, gives additional insight into what we can expect from the numerical calculations in Section \ref{examples}.  In particular, it shows that all the abstract two-sided eigenvalue estimates in \eqref{twosidedestimates} are optimal. 

Let $a$, $b$, and $c$ be positive \vspace{-0.5mm} constants,
\begin{equation}
\label{eps-space-indep}
	\epsilon (\lambda):=\aaa+ \frac{\bb}{\cc - \lambda}, \quad \la\in \C\setminus\{c\}. 
\end{equation}
Here $M=\wh M=1$, $P_1=I_{\Ltwo}$, and $W=a I_{\Ltwo}$ since $\Omega=\Omega_1$. Hence $A_k=- \frac 1a \Delta_k + \frac ba$ and $B=B_1= \sqrt{\frac{cb}a} I_{\Ltwo}$  so that the estimate in \eqref{eq:estB} is, 
\vspace{-1mm} in fact, an equality, $\|B\|=\sqrt{\frac{cb}a}$, and we know the eigenvalues of $A_k$ and of $\CA_k$ explicitly. With the notation $\{\nu_j(-\Delta_k)\}_{j=1}^\infty =\{|2\pi n+k|^2:n\in \Z^2\}$ introduced earlier for the eigen\-values of $-\Delta_k$ in increasing order, we have
\begin{equation}
	\{\nu_j (A_k)\}_{j=1}^\infty = \Big\{\frac{\nu_j(-\Delta_k)+b}{\aaa}\Big\}_{j=1}^\infty=\big\{\nu_j (A_k^{(0)})\big\}_{j=1}^\infty+\frac{b}{a}
\end{equation}
and the two sequences of eigenvalues of $\CA_{k}$ accumulating at $c$ and at $\infty$ are \vspace{-1mm} given~by 
\[
\begin{aligned}
	\{\la_{1,j}\}_{j=1}^\infty &= 	\bigg\{\frac 12 \Big(\nu_j(A_k) \!+\! c \Big) \!-\! \sqrt{\frac 14 \Big( \nu_j(A_k) \!-\! c \Big)^2 \!+\!   \frac{cb}a}
	\bigg\}_{j=1}^{\infty}, \hspace{-3mm}\\
	\{\la_{2,j}\}_{j=1}^\infty &= \bigg\{\frac 12 \Big(\nu_j(A_k) \!+\! c \Big) \!+\! \sqrt{\frac 14 \Big( \nu_j(A_k) \!-\! c \Big)^2 \!+\! \frac{cb}a} 
	\bigg\}_{j=1}^{\infty}. \hspace{-3mm}
\end{aligned}
\]
The order of convergence of $\la_{1,j}\to c$ and of $\la_{2,j}\to\infty$ is ${\rm O}(\nu_j (A_k)^{-1})$ as expected from Remark~\ref{re:order}. 

To compare with the two-sided eigenvalue estimates in Theorem~\ref{thm:2.3}, we need to determine
the index shift $\kappa_2$ from Proposition \ref{new-est} which is defined as the number of negative eigenvalues 
of $\CS_k(c+\delta)$ for any $\delta\!>\!0$ such that $c\!<\!\delta+c\!<\!\la_{2,1}$. It is easy to see that the eigenvalues $\mu_j$ 
of $\CS_k(\la)= A_k - \la - \frac{cb}a \frac 1{c-\la}$ for $\la\!=\!c+\delta$ \vspace{-0.5mm} are given~by 
\begin{equation}
	\{\mu_j\}_{j=1}^\infty = 
	\Big\{ \nu_j(A_k) - (c+\delta)  + \frac{cb}a \frac 1\delta \Big\}_{j=1}^{\infty}
\end{equation}
which implies that all eigenvalues are positive if $\delta>0$ is chosen small enough and hence $\kappa_2=0$.

Since $B=B_1= \sqrt{\frac{cb}a} I_{\Ltwo}$, we have $\min \sigma(BB^*) = \|B\|^2 = \frac{cb}a$, and
the two-sided estimates \eqref{twosidedestimates} all become equalities, i.e.
\begin{equation}\label{eq:EmptyBounds}
	\la^{L}_{1,j}=\la_{1,j} = \la^{U}_{1,j}, 
	\quad 
	\la^{L}_{2,j}= \la_{2,j}=\la^{U}_{2,j},
\end{equation}
which proves that they are optimal.
%

%% file: Part_Galerkin-cc.tex
\section{\bf Galerkin approximations}\label{sec:Galerkin}

In this section, we consider  Galerkin finite element approximations of eigenpairs $(\la,u)$ 
of the 
operator function \eqref{eq:bas} with permittivity function \eqref{eq:epsModel} or, equivalently, eigenpairs of the  Schur complement \eqref{eq:Sk}.  
Since, in Section \ref{conforming}, \ref{SIP} below,  we use results that were derived for bounded forms, 
the forms associated with the operators \eqref{eq:bas} will be considered in $\Hone$ rather than in $\Ltwo$.
Moreover, here it is more convenient to use the operator $W$ in \eqref{W} in a new inner product and work with $\CT_k$ instead of $\CS_k$. 

In the space $\Hone\times\Hone$, for fixed $k\in\CK=(-\pi,\pi]^2$,  we consider the bounded sesquilinear \vspace{-1mm} forms  
\begin{alignat}{2}
\nonumber	
	\mathfrak{t}^{(0)}_k [u, v] &:=\int_{\Omega}(\nabla+\iu k) u\cdot\overline{(\nabla+\iu k)v}\, \mrm{dx},\quad 
	& \mathfrak{t}^{(1)} [u, v] &:=\sum_{m=1}^{\wh{M}}\sum_{\ell=1}^{L_m}\bb_{m,\ell}\mathfrak{b}_m [u, v],\\
\label{def:sesq}        
	\mathfrak{t}_k [u, v] &:= \mathfrak{t}^{(0)}_k [u, v]+\mathfrak{t}^{(1)} [u, v], & \\
\nonumber
        \mathfrak{b}_m [u, v] &:=\int_{\Omega_m}u\overline{v}\, \mrm{dx},\quad & (u,v)_w&:=\sum_{m=1}^{M}a_m\mathfrak{b}_m [u, v],
\end{alignat}
and denote by $\| u \|_{w}=\sqrt{(u,u)_w}$ the corresponding weighted $L_2$-norm. Then the eigenvalues of $\CT_k$ are determined by the following variational problem: 

Find $u\in\Hone\backslash\{0\}$ and $\lambda\in\R\backslash\{c_{m,\ell}: m=1,2,\dots,\wh M, l=1,2,\dots,L_m\}$ such that for all $v\in\Hone$
\begin{equation}\label{eq:sG}
	\mathfrak{t}_k(\lambda)[u, v]:=\mathfrak{t}_k[u, v]-\lambda (u,v)_w-\sum_{m=1}^{\wh{M}}\sum_{\ell=1}^{L_m}\frac{\cc_{m,\ell}\bb_{m,\ell}}{\cc_{m,\ell}-\lambda}\mathfrak{b}_m[u, v]=0.
\end{equation}
From the preceding sections, we know that the eigenvalues of $\CS_k$, and hence of $\CT_k$, are isolated and of finite multiplicity and hence so will be the values $\la$ for which the above problem admits a solution (comp.\ e.g.\  \cite[p.\ 131]{MR2373954}. 

Let $\Hi_{\wt N}$ denote an arbitrary $\wt N$-dimensional subspace of $\Hone$. Then the discrete (Galerkin) formulation of problem \eqref{eq:sG} is to seek the $\wt N$ eigenpairs 
$(\wt{\la}_{j}, \wt {u}_{j})\in (\R \backslash\{c_{m,\ell}: m=1,2,\dots,\wh M, l=1,2,\dots,L_m\}) \times \Hi_{\wt N}$ such that for all $v\in\Hi_{\wt N}$
\begin{equation}\label{eq:sGalerkin}
	\mathfrak{t}_k(\wt \la_{j})[\wt u_{j},v]=0, \quad j=1,2,\dots, \wt N.
\end{equation}
For a sequence of subspaces $\Hi_{\wt N} \subset \Hone$, the corresponding sequences of eigenpairs $\big( \wt \la_{j}, \wt u_{j} \big)$, $j=1,2,\dots, \wt N$, 
are viewed as approximations to the true eigenpairs $(\la_{j},u_{j})$ of the spectral problem $\CS_k(\la) u = 0$ (comp.\ e.g.\ \cite{MR962210}).

According to Theorem \ref{th_j} all eigenvalues of $\CS_k$, and hence of $\CT_k$, can be determined by means of the variational principles \eqref{minmax-}, \eqref{minmax+}.
For a material model \eqref{eq:epsModel} with $M=2$, $\wh M=1$, and $L_1=1$, 
the functionals $p_{1,2}$ in \eqref{ppm} 
are given by
\begin{equation}
\label{p12}
	p_{1,2}(u)
	=\frac{1}{2}\left (\frac{\mathfrak{t}_k [u, u]}{\| u \|^2_{w}}+c_{1,1}\right )\mp\sqrt{\frac{1}{4}\left (\frac{\mathfrak{t}_k [u, u]}{\| u \|^2_{w}}-c_{1,1}\right )^2
	+\frac{b_{1,1}c_{1,1}\mathfrak{b}_1 [u, u]}{\| u \|^2_{w}}}
\end{equation}
for $u\in \Hone$. Note that, in analogy to Remark \ref{formdomain}, the variational principles \eqref{minmax-}, \eqref{minmax+} hold with 
$\dom \Delta_k=\Htwo$ replaced by 
$\dom {\mathfrak t}_k =\Hone$. 


In the same way as for the eigenvalues of $\CS_k$, we divide the Galerkin eigenvalues into two groups. By
\begin{align}\label{eq:Gal_low}
	&\wt \la_{1,1}\leq \wt \la_{1,2}\leq \dots \leq \wt \la_{1,\wt N_1} < c_{1,1} \\ 
\intertext{we denote the $\wt N_1 $ $(\leq \wt N)$ eigenvalues of \eqref{eq:sGalerkin} below $c_{1,1}$ and by}
\label{eq:Gal_high}
	&\cc_{1,1} < \wt \la_{2,1}\leq \wt \la_{2,2}\leq \dots  \leq \wt \la_{2,\wt N_2} 
\end{align}
the $\wt N_2$ $(=\wt N-\wt N_1)$ eigenvalues above $\cc_{1,1}$; the corresponding eigenvectors are denoted by $(\wt u_{1,j})_{j=1}^{\wt N_1}$, $(\wt u_{2,j})_{j=1}^{\wt N_1} \subset \Hi_{\wt N}$.
Then the variational principle \eqref{minmax-} for $\mathfrak{t}_k$ together with  Remark \ref{formdomain} implies that
\begin{align} \label{GalerkinMinmax}
  \lanm & \!=\!\! \min_{\ontops{\CL \subset \Hone}{\dim \CL = j}} \!\! \max_{\ontops{u \in \CL}{u\ne 0}} \ p_1(u) 
  		  \le \max_{\ontops{\!u\in{\rm span}\{\!\wt u_1,\dots,\wt u_{j}\!\}\!}{u\ne 0}} p_{1}(u) 
		   \le  \wt \la_{1,j}, \quad j=1,2,\dots,\wt N_1,
\end{align}
and, similarly, the variational principle \eqref{minmax+} implies $\lanp\le  \wt \la_{2,j}$, $j=1,2,\dots,\wt N_2$. 
So, in analogy to the classical min-max (or Rayleigh-Ritz) variational principle, we have the chain of inequalities
\begin{equation}
\label{June5}
 \la_{1,j} \le   \wt \la_{1,j} < c_{1,1} < \la_{2,k} \le   \wt \la_{2,k}, \quad j=1,2,\dots, \wt N_1, \ k=1,2,\dots, \wt N - \wt N_1, 
\end{equation}
for the eigenvalues of $\CT_k$ and the Galerkin eigenvalues to the left and to the right of the pole $c_{1,1}$.

Regarding the eigenfunctions, let $(\la_{1,j}, u_{1,j}) \in \R\times \Htwo$, $j\in\N$, be a sequence of eigenpairs of $\CS_k$, here given by
\begin{equation}
	\CS_k(\lambda)= A_k-\lambda-B_{1,1}(\cc_{1,1}-\lambda)^{-1}B_{1,1}^{*},\quad B^{*}_{1,1}:=\sqrt{\frac{\cc_{1,1}\bb_{1,1}}{a_1}}P_1,
\end{equation}	
such that $\la_{1,j}\nearrow c_{1,1}$, $j\to\infty$, and $\|u_{1,j}\|=1$, $j\in\N$. 
Then the corresponding eigenvectors are of the form $(u_{1,j} \ \wh u_{1,j})^{\rm t}$ with $\wh u_{1,j} = - (c_{1,1}-\la)^{-1} B^* u_{1,j}$, $j\in\N$, and
$\big( (u_{1,j} \ \wh u_{1,j})^{\rm t} \big)_{j=1}^\infty$
is a singular sequence of $\CA_k$ at $\cc_{1,1}$, i.e.
\[
  \|(\CA_k\!-\!\cc_{1,1}) (u_{1,j} \, \wh u_{1,j})^{\rm t} \|^2\!\! =\|(A_k\!-\!\cc_{1,1}) u_{1,j} \!+\!B \wh u_{1,j} \|_{\Ltwo}^2\!+\!\|B^*u_{1,j}\|_{\Hii}^2 \!\to\! 0,\ j\to \infty.
\]
This implies that
\begin{equation}
	\|B^* u_{1,j} \|_{\Hii}^2=\bb\cc\int_{\Omega_1}| u_{1,j}(x) |^2 \d x\to 0,\quad j\to \infty.
\end{equation}
Thus we may expect that, for large $j$, the $L_2$-norm of the eigenfunctions $u_{1,j}$ of $\CS_k$ on $\Omega_1$ is very small. 
We mention that while the order of convergence for the eigenvalues $\la_{1,j}\nearrow c_{1,1}$ is ${\rm O}(\nu_j (A_k)^{-1})$ according to Remark \ref{re:order},
the order of convergence $\|B^* u_{1,j}\|_{\Hii}\to 0$ will depend on the geometry and other parameters. 
Furthermore, since $\big(\cc_{1,1},\cc_{1,1}+\varepsilon \big)\subset\rho(\CA_k)$ for some~$\varepsilon >0$, we may not expect that the $||B^* u_{2,j}||_{L^2(\Omega_1)}$ 
is small for eigenvectors $u_{2,j}$ of $\CS_k$ corresponding to eigenvalues $\la_{2,j}$ above $\cc_{1,1}$. 

Figure \ref{Fig:EigVector} illustrates this by numerical computations for a particular example where $\cc_{1,1}=12.7367$. 
The constants in the material model \eqref{eq:epsModel} are $M=2$, $\wh M=1$, $L_1=1$, and $\Omega_1$ is a disk of radius $r=0.475$. 
Numerically the convergence $\|B^* u_{1,j}\|_{\Hii}\to 0$ is faster than $(c_{1,1}-\lambda_{1,j})^{-1}\to\infty$ in this case, 
which can be seen in the third plot from the left. The numerical method used here will be outlined in the next two \vspace{-4mm}  sections.

\begin{figure}[h]
\begin{center}
\includegraphics[width=3cm]{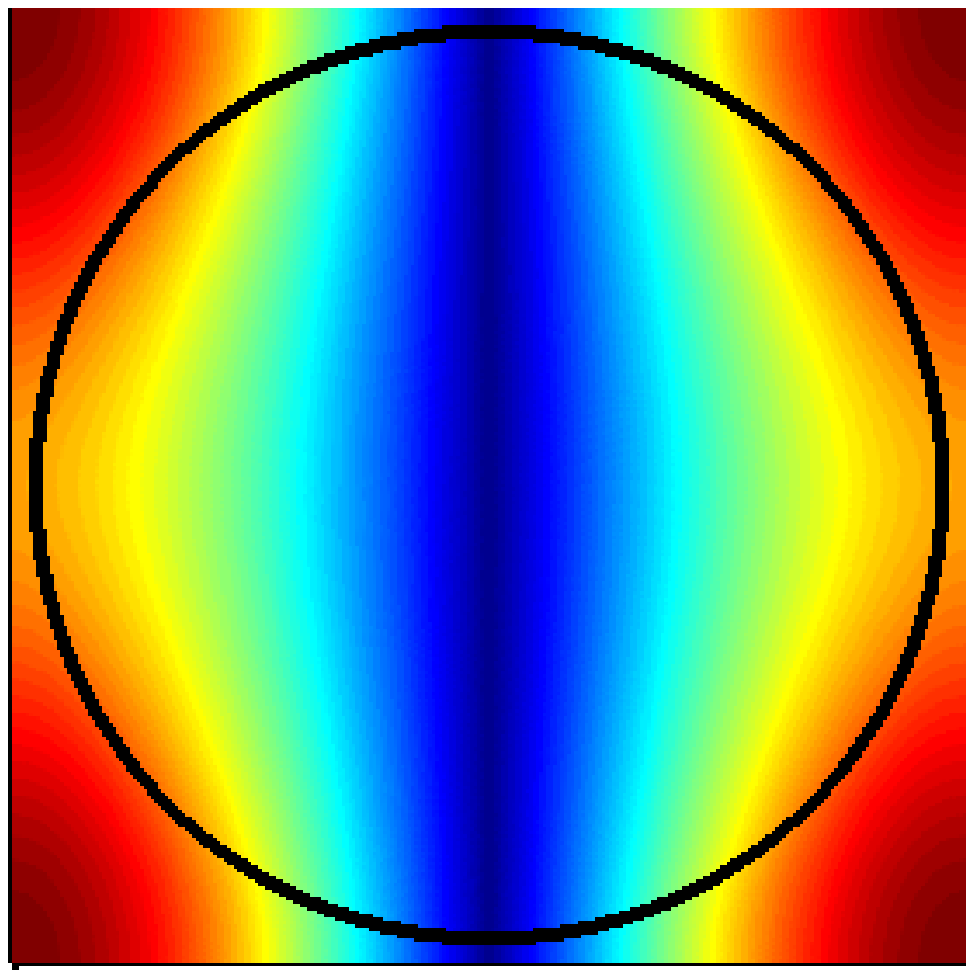}
\includegraphics[width=3cm]{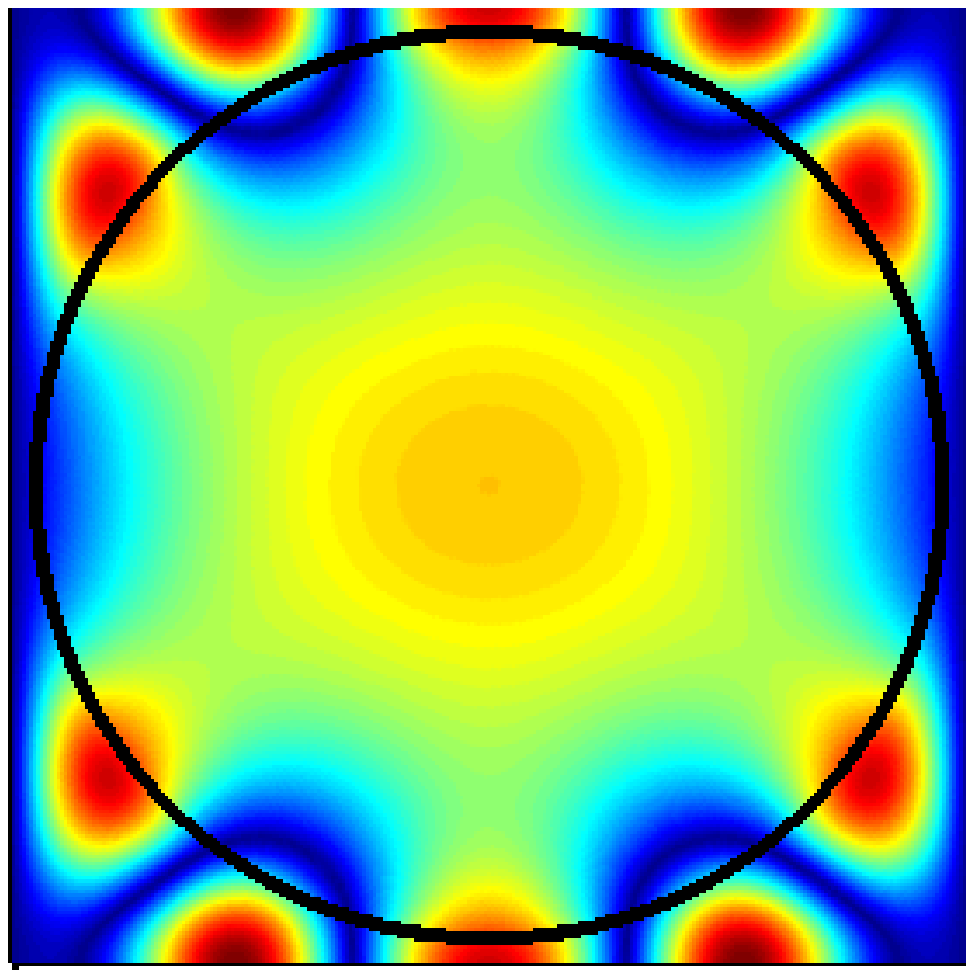}
\includegraphics[width=3cm]{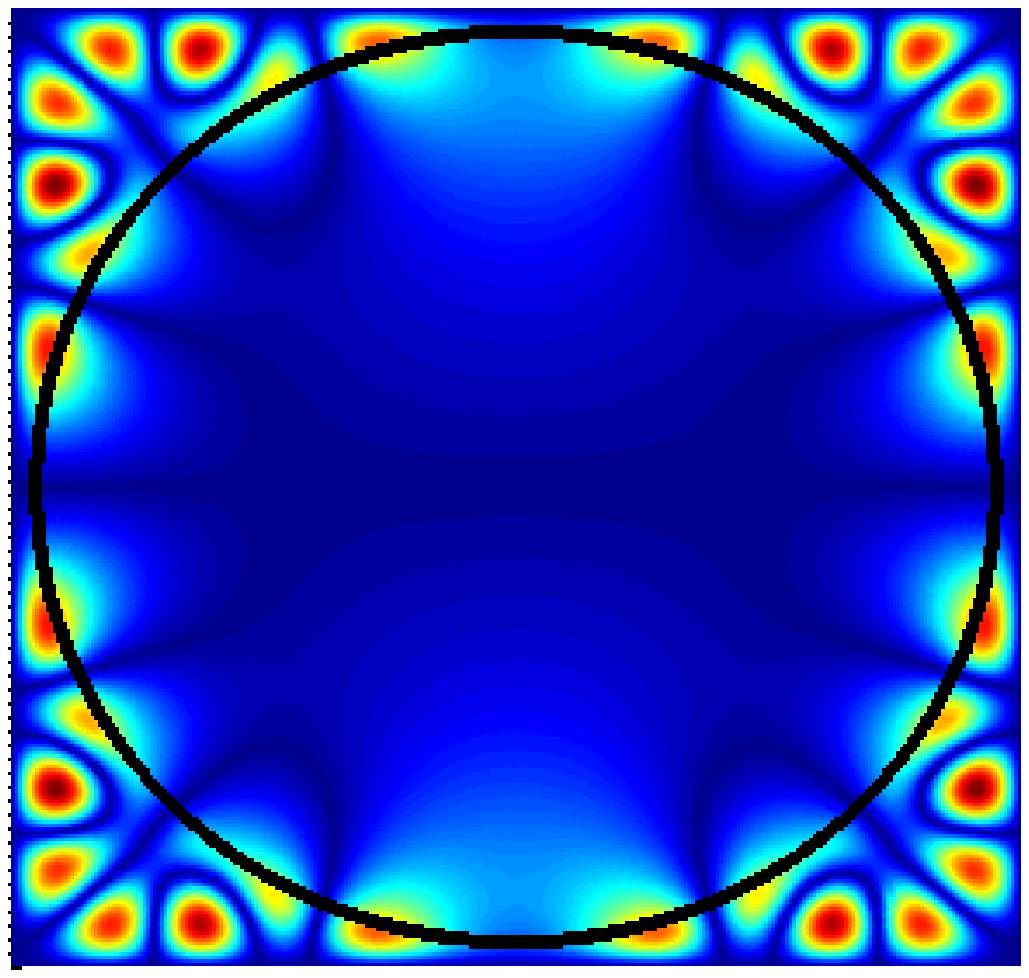}
\includegraphics[width=3cm]{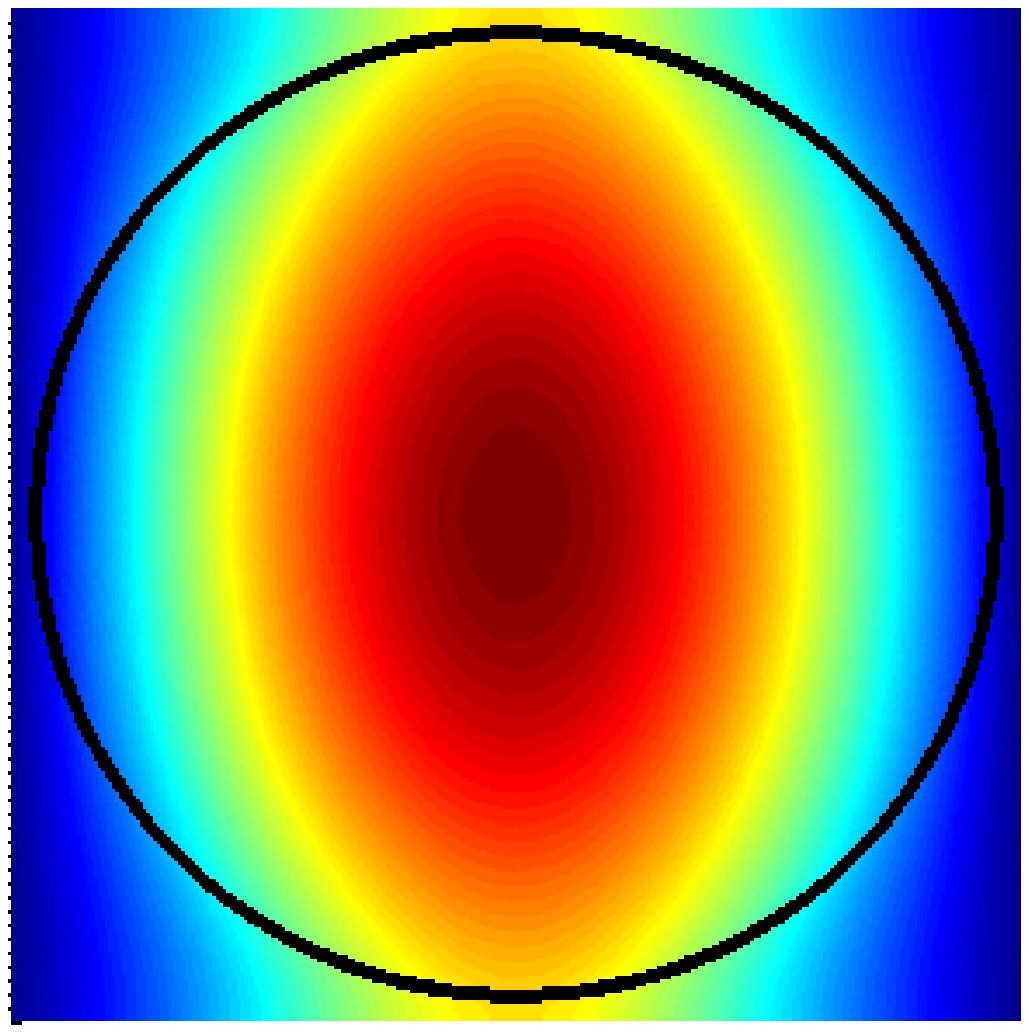}
\caption{\small $\Omega_1$ disk of radius $r=0.475$ (black circle) and  pole $\cc_{1,1}=12.7367$. The absolute value of the approximate eigenvectors $u$ 
corresponding to (counting from the left) the eigenvalues $\wt\la_{1,1}=1.4012$, $\wt\la_{1,2}=11.6218$, $\wt\la_{1,3}=12.613$, and $\wt\la_{2,1}=14.3949$. 
Blue corresponds to $|u(x)|$ close to zero and red to large function \vspace{-4mm} values.}
\label{Fig:EigVector}
\end{center}
\end{figure}

In the following numerical computations with space-dependent permittivity function $\epsilon$ we will approximate the analytic bounds on the eigenvalues of $\CA_k$ established in Theorem \ref{thm:2.3}. 
The estimates \eqref{upbound+}, \eqref{lobound+} only involve the eigenvalues of the positive densely defined unbounded self-adjoint operator $A_k$, 
which has compact resolvent;  note that we already know $\min \sigma(BB^*) =0$ and $\|B\|$, or can estimate the latter.
The spectrum of the operator $A_k$ can be obtained from the following variationally posed eigenvalue problem: 

Find $u\in\Hone\backslash\{0\}$ and $\nu\in\R$ such that for all $v\in\Hone$ 
\begin{equation}\label{eig:a}
	\mathfrak{t}_k[u, v]=\nu (u,v)_{\mrm{w}}.
\end{equation}
The Galerkin approximations $\wt \nu_{j}$ of $\nu_{j}(A_k)\in\sigma(A_k)$, $j=1,2,\dots, \wt N$, with $\wt \nu_1\leq \wt\nu_2\leq \dots \leq \wt\nu_{\wt N}$
and the corresponding eigenvectors are eigenpairs $(\wt \nu, \wt u)\in \R \times \Hi_{\wt N}\backslash\{0\}$ 
such that for all $\wt v\in\Hi_{\wt N}$  
\begin{equation}
\label{eig:aG}
 \mathfrak{t}_k[\wt u, \wt v]=\wt \nu (\wt u,\wt v)_{\mrm{w}}.
\end{equation}
The classical min-max variational principle \cite[Theorem XIII.1]{MR0493421} applies and shows that the Galerkin eigenvalues are upper bounds for true eigenvalues
\cite[Theorem XIII.3]{MR0493421},
\begin{equation}\label{eq:A_Min_Max}
	\nu_{j}(A_k)\le \wt{\nu}_{j},\quad j=1,2,\dots,\wt N.
\end{equation}	
Therefore, by Remark \ref{boundsmonotonic}, since the bounds $\la_{1,j}^{U}$ and $\la_{2,j+\kappa_2}^{U}$ are increasing functions of $\nu_{j}(A_k)$, the eigenvalues satisfy the estimates
\begin{align} 
  \lanm&\le  \lanmU\le  \TlanmU, \\
  \la_{2,j+\kappa_2} &\le  \la_{2,j+\kappa_2}^{U}\le  \wt{\la}_{2,j+\kappa_2}^{U},
\end{align}
where $\lanm$ and $\lanp$ are defined  by \eqref{upbound+} and \eqref{lobound+}, respectively, and
where $\wt \la_{1,j}^U$, $\wt \la_{2,j}^U$  are obtained from $\lanmU$, $\lanpU$ by replacing $\nu_{j}(A)$ by its upper bound $\wt \nu_{j}$.

\subsection{The conforming finite element method (FEM)}
\label{conforming}

The conforming Galerkin FEM is characterized by a family of triangulations $\mathcal{Q}_{h}$ of $\Omega$ with quadrilateral elements~$Q$. Let $F_{Q}$ be a bijective mapping of a chosen 
reference square $Q_{\mrm{ref}}$ onto an element~$Q$ and let $\mathbb{P}^{p}$ 
denote the space of polynomials on $\R^2$ of degree $\le p$, i.e. $\mathbb{P}^{p}:=\text{span}\,\{x^{n}_1 x^{m}_2, 0\leq n,\, m\leq p\}$. 
Then the space of polynomials on $\R^2$ of degree $\le p$ that are piecewise polynomials on the triangulation of $\Omega$ is defined as
\begin{equation}
 \mathcal{P}_{h}^{p}:=\{u\in\Hone:\,u|_{Q}\circ F_{Q}\in \mathbb{P}^{p}\};
\end{equation} 
here $\left.u\right|_Q$ denotes the restriction of $u$ to $Q$. The topology of the torus is imposed by mapping the parallelogram edges situated on one boundary on the corresponding 
parallelogram edge on the opposite side. 
From the Cauchy--Schwarz inequality it follows that for $t\in [0,1]$  and $n\in\Z^2$ the inequality $|2\pi n+k|^2\geq (2\pi)^2|n|^2(1-t)+|k|^2$ holds when $|k|\leq t\pi$. 
Hence, for $|k|\leq t\pi$ with $t\in [0,1]$, the inequality 
\begin{equation}
	\mathfrak{t}_k[u, u]\geq\int_{\Omega}|(\nabla+\iu k) u|^2\, \mrm{dx}\geq (1-t)||\nabla u||^2+|k|^2||u||^2, \quad u \in \Hone, 
\end{equation}
follows. Thus, for $k\neq 0$, $\sqrt{\mathfrak{t}_k[u,u]}$ is a norm equivalent to the standard norm on $\Hone$ and $\mathfrak{t}_k[u,v]$ can be used as an inner product. 
As above, let $\{\nu_j(A_k)\}_{j=1}^\infty$ 
denote the eigenvalues of $A_k$ ordered increasingly and let 
$\{u_j\}_{j=1}^\infty$ be a system of eigenvectors that are orthonormal with respect to the inner product $\mathfrak{t}_k[\cdot,\cdot]$,
i.e.\  $\mathfrak{t}_k[u_i, u_j]=\delta_{ij}$, $i,j\in\N$. 

The corresponding finite dimensional problem is \eqref{eig:aG} with the special choice $\Hi_{\wt N} = \mathcal{P}_{h}^{p}$, i.e.
to seek $\wt u\in\mathcal{P}_{h}^{p}\backslash\{0\}$ and $\wt \nu\in\R$ such that for all $\wt v\in\mathcal{P}_{h}^{p}$
\eqref{eig:aG} holds.
Let $0\leq\wt \nu_1\leq \wt \nu_2\leq\dots\leq \wt \nu_{\wt N}$ denote the eigenvalues of \eqref{eig:aG} 
with corresponding eigenvectors $\wt u_1,\wt u_2,\dots ,\wt u_M$ orthornormal with respect to $\mathfrak{t}_k[\cdot,\cdot]$, 
i.e.\  $\mathfrak{t}_k[\wt u_i, \wt u_j]=\delta_{ij}$, $i,j\in\N$. Moreover, let $P$ denote the orthogonal projection of $\Hone$ onto $\mathcal{P}_{h}^{p}$. 
Then, the well-known estimate \cite{MR0400004,MR2206452} 
 \begin{equation}
 	0\leq\frac{\wt \nu_{j}-\nu_{j}(A_k)}{\wt \nu_{j}}\leq\sum_{i=1}^{j}\|u_i-Pu_i\|_{\mathfrak{t}_k}^{2}
 \end{equation}
holds which, under the assumption that $\sum_{i=1}^{j}\|u_i-Pu_i\|_{\mathfrak{t}_k}^{2}<1$, can be written in the equivalent \vspace{-2mm}  form
 \begin{equation}\label{eq:estA}
 	0\leq\frac{\wt \nu_{j}-\nu_{j}(A_k)}{\nu_{j}(A_k)}\leq\frac{\sum_{i=1}^{j}\|u_i-Pu_i\|_{\mathfrak{t}_k}^{2}}{1-\sum_{i=1}^{j}\|u_i-Pu_i\|_{\mathfrak{t}_k}^{2}}.
 \end{equation}
Note that (\ref{eq:estA}) depends on the approximation errors for all eigenfunctions  $u_1,u_2$, $\dots ,u_{j}$. 
It is possible to derive error estimates that depend mainly on the approximation error for $u_{j}$ \cite{MR2206452} 
but this will not be important in our setting since all eigenvectors are (piecewise) analytic. 
In the following we will express the convergence rates in terms of the number of degrees of freedom $\wt N$. 
If $u_{j}$ is analytic, the convergence rate for the $p$-version of the finite element method is exponential,
\begin{equation}
\label{Rate_p}
	\| u_i-Pu_i\|_{\mathfrak{t}_k}\leq C\eu^{-\gamma \wt N^{1/2}},
\end{equation}
for some positive constants $C$ and $\gamma$ \cite{MR954788}. Moreover, the estimate (\ref{Rate_p}) 
holds when the eigenfunctions are analytic in each subdomain up to the interfaces and the interfaces are exactly resolved by curvilinear cells \cite{MR954788}.

In the following, we discretize the operator $A_k$ and  the block operator matrix (\ref{eq:A-multi}) 
using the implementation available in the software package Concepts \cite{MR1955543}. The shape functions are based on 
Jacobi polynomials and a blending technique is applied to construct element mappings.  
Concepts uses curvilinear quadrilateral cells that resolve the curved analytic material interfaces in the numerical examples. Further implementation details can be found in \cite{MR1955543,Schmidt+Kauf2009}. 

Let $\{\phi_1,\phi_2,\dots,\phi_{\wt N}\}$ be a basis of $\mathcal{P}_{h}^{p}$. For $u\in\mathcal{P}_{h}^{p}$, the approximate eigenfunctions $\wt u$ are of the \vspace{-2mm} form
\begin{equation}
	\wt u=\sum_{j=1}^{\wt N}\alpha_{j}\phi_{j}.
	\vspace{-2mm} 
\end{equation}
The elements in the matrices $\widetilde{T_k}$, $\widetilde{W}$, and $\widetilde{B}_m$ are
\begin{equation}\label{eq:elements}
	(\widetilde{T_k})_{ij}=\mathfrak{t}_k[\phi_{j},\phi_{i}],\quad	
	(\widetilde{W})_{ij}=(\phi_{j},\phi_{i})_{w},\quad	
	(\widetilde{B}_m)_{ij}=\mathfrak{b}_m[\phi_{j},\phi_{i}],
\end{equation}
where the sesquilinear forms are defined by (\ref{def:sesq}). The finite element approximation of the rational eigenvalue problem \eqref{eq:sG} then is
\begin{equation} \label{eq:RathatSh}
	\widetilde{\CT_k}(\lambda) x=0,\quad x=(\alpha_1,\alpha_2,\dots,\alpha_{\wh N})^t,
\end{equation}
where, for $ \lambda\in\C \setminus \{\cc_{m,\ell}: m=1,2,\dots,\wh M, \, \ell=1,2,\dots,L_m \}$,
\begin{equation} \label{eq:hatSh}
	\widetilde{\CT_k}(\lambda):=\widetilde{T_k}-\lambda  \widetilde{W}-\sum_{m=1}^{\wh{M}}\sum_{\ell=1}^{L_m}\frac{\cc_{m,\ell}\bb_{m,\ell}}{\cc_{m,\ell}-\lambda}\widetilde{B}_m.
\end{equation}
The eigenvalues of $\widetilde{A_k}:=\widetilde{W}^{-1/2} \widetilde{T_k} \widetilde{W}^{-1/2}$ are used to approximate the eigenvalues of $A_k$ and hence the bounds on the eigenvalues of $\CA_k$ in \eqref{upbound+}, \eqref{lobound+}.  
The rational eigenvalue problem $\widetilde{\CS_k}(\lambda)x:=\widetilde{W}^{-1/2} \widetilde{\CT_k}(\lambda) \widetilde{W}^{-1/2} x=0$ with $\wh M=1$ can be written in the form \eqref{eq:A-multi} and we will calculate the approximate eigenvalues using this block matrix form. 

\begin{remark}
	Note that $\widetilde{B}_m$ is not block diagonal for conforming methods and the method outlined in this section can only be used when $\wh M=1$. However, the approximation spaces for discontinuous Galerkin (DG) methods are localised in each element and $\widetilde{B}_m$ is then block diagonal. A DG method that can handle $\wh M>1$ is outlined in Section \ref{SIP}.
\end{remark}

The variational characterization of the eigenvalues (\ref{GalerkinMinmax}) shows that (ignoring rounding errors) the approximate eigenvalues are upper bounds on the eigenvalues of $\CA_k$ (see \eqref{June5}). 

Given two closed subspaces $\CH_1$ and $\CH_2$ of a Hilbert space $\CH$ the proximity of the spaces is measured in terms of the containment gap \cite[Section IV.2.1]{MR1335452}
\[
        \wh \delta(\CH_1,\CH_2)\!:=\!\max \big\{ \delta(\CH_1,\CH_2), \delta(\CH_2,\CH_1) \big\}, \quad 
	\delta(\CH_1,\CH_2)\!:=\!\sup_{u_1\in \CH_1}\!\inf_{u_2\in \CH_2}\!\frac{\|u_2-u_1\|_{\CH}}{\|v_1\|_{\CH}}.
\]
For the gap $\wh\delta(E_{h}^{p}(\la),E(\la))$  between the discrete and continuous eigenspaces $E_{h}^{p}(\la)$ and $E(\la)$ of $A_k$, respectively,
it is known that $\wh \delta(E_{h}^{p}(\la),E(\la))\to 0$ 
when $h\!\to\!0$ or $p\!\to\!\infty$ \cite{MR0383117}.
This implies several decisive properties including non-pollution of the spectrum \cite{Davies+Plum2004}. 
Hence, without risk of spectral pollution, we can calculate accurate approximations of the estimates  $\lanmL$ in  \eqref{upbound-} and $\lanpL$ in \eqref{lobound-}.  
We mention that a complete convergence theory for $\CA_k$ will be studied in a forthcoming paper; a block operator formulation of a related problem was considered in \cite{MR3164142}.
Since all eigenfunctions are analytic in each subdomain up to the interfaces we expect exponential convergence. The presented continuous finite element method 
will be used in Section \ref{examples} to compute the eigenvalues but the discontinuous Galerkin method described in Section \ref{SIP} below is a strong alternative.

\subsection{The symmetric interior penalty method.}
\label{SIP}

In the following, we discretize \eqref{eq:sG} with a discontinuous finite element method called the symmetric interior penalty method (SIP) \cite{MR2220929}. 
More specifically, we use the $p$-version of SIP 
\pagebreak
in this paper since in the numerical examples all eigenvectors are analytic in each subdomain and the subdomains have 
analytic interfaces. Discontinuous Galerkin methods, such as SIP, are interesting partly because they are more flexible in the choice of basis functions and in the mesh design \cite{MR2946411}. 
Moreover, the mass matrix is block diagonal and can therefore be inverted at low computational cost. This is frequently used to obtain explicit time integration, but a block diagonal mass matrix 
is also an advantage for the solution of our non-linear eigenvalue problem.

Let $\mathcal{T}_{h}$ denote the triangulation of $\Omega$, let 
$F_{T}$ be a bijective mapping of a chosen reference triangle 
$T_{\mrm{ref}}$ onto an element $T$, and let $\mathcal{V}_{h}^{p}$ denote the space of polynomials on $\R^2$ of degree $\le p$ that are piecewise constant on the triangulation of $\Omega$,
\begin{equation}
	\mathcal{V}_{h}^{p}:=\{u\in L^{2}(\Omega)\,:\, u|_{T}\circ F_{T}\in \mathbb{P}^{p} \};
\end{equation}
here $\left.u\right|_T$ denotes the restriction of $u$ to $T$. Consider two adjacent
triangles $T_{+}, T_{-} \in \mathcal{T}_{h}$ 
with outward pointing normals $\vec{n}_{+}$, $\vec{n}_{-}$ on the shared edge $\partial T_{+}\cap\partial T_{-}$. The symbols $\nabla_h$ and $\nabla_h\cdot$ denote the elementwise (broken) gradient operator and divergence operator, respectively. 
The solutions of \eqref{eq:sG} are in $H^{2}(\Omega)$. Hence, for all $T\in \mathcal{T}_{h}$ the traces on the edges are well-defined. The averages $\{\cdot\}$ and jumps $[\cdot]$ of $w$ and $\nabla_h w$ on $T_+\cup T_-$ are then defined as
\begin{alignat}{2}
	\{ w \} &:= \tfrac{1}{2} (w_{+} + w_{-}), &\qquad [w] &:= w_{+}\vec{n}_{+} + w_{-}\vec{n}_{-},\\
	\{ \nabla_h w \} &:= \tfrac{1}{2} (\nabla_h w_{+} + \nabla_h w_{-}), &\qquad 
\end{alignat}
where $w_{\pm}$ and $\nabla_h w_{\pm}$ denote the traces on $\partial T_{\pm}$. Let $\mathcal{E}$ denote the set of all edges $e$ of $\mathcal{T}_{h}$. We will use the convention
\begin{equation}
	\int_{\Omega}u_h \mrm{dx}:=\sum_{T \in\mathcal{T}_{h}}\int_{T}u_h \mrm{dx},\quad \int_{\mathcal{E}}u_h \mrm{ds}:=\sum_{e \in\mathcal{E}}\int_{e}u_h \mrm{ds}.
\end{equation}
For $u_h, v_h \in \mathcal{V}_{h}^{p}$ we consider the sesquilinear forms
\begin{equation} \label{eq:bilinearform}
	\begin{aligned}
		\wh{\mathfrak{t}}_k^{(0)}[u_h,v_h] &:= \int_{\Omega}\left <(\nabla_h+\iu k) u_h,\overline{(\nabla_h+\iu k)v_h }\right >\mrm{dx}+ \int_{\mathcal{E}} \left <\beta [u_h],[\bar{v}_h]\right >\mrm{ds}\\
		\phantom{a} & \phantom{{} = b}-\int_{\mathcal{E}}\left <\{(\nabla_h+\iu k) u_h\},[\bar{v}_h]\right >+\left <\{\overline{(\nabla_h+\iu k) v_h}\},[u_h]\right >\mrm{ds},\\
	\wh{\mathfrak{t}}_k [u_h,v_h] &:=\wh{\mathfrak{t}}_k^{(0)}[u_h,v_h]+\mathfrak{t}^{(1)}[u_h,v_h].
\end{aligned}
\end{equation}
Let $\{\varphi_1,\varphi_2,\dots,\varphi_{\wt N}\}$ be a basis of $\mathcal{V}_{h}^{p}$. The finite element matrices and 
the corresponding rational matrix function $\wh \CT_k$ are then formed as 
in \eqref{eq:elements}--\eqref{eq:hatSh}, with 
\begin{equation}
	(\widehat{T_k})_{ij}=\wh{\mathfrak{t}}_k[\varphi_{j},\varphi_{i}],\quad	
	(\widehat W)_{ij}=(\varphi_{j},\varphi_{i})_{w},\quad	
	(\widehat B _m)_{ij}=\mathfrak{b}_m[\varphi_{j},\varphi_{i}].
\end{equation}
The periodic boundary conditions are imposed weakly by identifying opposite sides of the unit cell and enforcing periodicity of the solution via the corresponding penalty terms in~\eqref{eq:bilinearform}. To generate the mesh, we use the software package $Emc^{2}$~
\cite{Saltel+Hecht1995} and for the discretization we build on the {\sc Matlab} version~\cite{MR2372235} of \texttt{NUDG++} (\texttt{www.nudg.org}). In order to preserve the high accuracy of the $p$-version of 
SIP it is essential to use curvilinear elements. A Gordon--Hall blending \cite{MR0451775} is used to preserve the approximation properties of the basis.  One advantage of SIP is the 
block diagonal mass matrix which directly splits into mass matrices for $\Omega_m$, $m=1,2,\dots,M$. Hence, $\widehat{B}_m$, $m=1,2,\dots,\wh{M}$, are block diagonal matrices. 
Let, for example, $\wh{M}=2$ and write $\widehat{B}_1$, $\widehat{B}_2$ in the form $\widehat{B}_1=\text{diag}\, (D_1,0)$ and $\widehat{B}_2=\text{diag}\, (0,D_2)$, respectively. 
The matrices $D_1\in\R^{n_1\times n_1}$ and $D_2\in\R^{n_2\times n_2}$ are positive definite, where $n_1$, $n_2$ correspond to the number of basis functions supported in domains 
$\Omega_1$, $\Omega_2$. Let $D_{j}=:L^{*}_{j} L_{j}$ denote the Cholesky decomposition of $D_{j}$, $j=1,2$. Then the Cholesky decomposition of the block diagonal matrix $\widehat{W}$ \vspace{-1mm} is
\begin{equation}
\widehat{W}:= L_{\widehat W}^{*} L_{\widehat W}:=\matrix{cc}{\sqrt{\aaa_1} L_1^* & 0 \\ 0 & \sqrt{\aaa_2} L_2^* }\matrix{cc}{\sqrt{\aaa_1} L_1 & 0 \\ 0 & \sqrt{\aaa_2} L_2}.
 \end{equation}
Note that, for a given polynomial degree $p$,  $D_{j}$, $j=1,2$, are block diagonal with block sizes $(p+1)(p+2)/2$ ~\cite[p.\ 171]{MR2372235} 
and $L_{j}$ therefore has only triangular blocks on the diagonal. Hence, the inverses of $L_{j}$ and of $L_{\widehat W}$ can be computed at low costs. 
Let $I_{n_j}$, $j=1,2$, denote the identity matrix on $\R^{n_j\times n_j}$ and define
\[
	\widetilde{\Si}_k(\lambda) \!:=\! (L^{*}_{\widehat W})^{-1}\widehat{\CT_k}(\lambda)(L_{\widehat W})^{-1}\!=\!\widetilde{A_k}-\lambda-\widetilde B_{1,1}(\cc_{1,1}-\lambda)^{-1}
	\widetilde B_{1,1}^*-\widetilde B_{2,1}(\cc_{2,1}-\lambda)^{-1}\widetilde B_{2,1}^*
\]	
where $\widetilde A_k\!:=\!(L^{*}_{\widehat W})^{-1}\widehat T_k(L_{\widehat W})^{-1}\!$, $\widetilde B_{1,1}\!:=\!\sqrt{\frac{\cc_{1,1}\bb_{1,1}}{\aaa_1}}(I_{n_1}\, 0)$, \vspace{-1mm}
$\widetilde B_{2,1}\!:=\!\sqrt{\frac{\cc_{2,1}\bb_{2,1}}{\aaa_2}}(0\, I_{n_2})$. The function $\widetilde \Si_k$ is the first Schur complement of 
\begin{equation}
\widetilde \CA_k =\!
\matrix{c|cc}{\widetilde A_k& \widetilde B_{1,1}\!\!&\!\!\widetilde B_{2,1}\\ \hline \\[-3.5mm]
\widetilde B^*_{1,1}&c_{1,1}\!\!&\!\!0\\
\widetilde B^*_{2,1}&0\!\!&\!\!c_{2,1}
}\!.
\vspace{-1mm}
\end{equation}
Hence, the eigenvalues can be computed from the $2(n_1+n_2)\times 2(n_1+n_2)$ matrix $\widetilde \CA_{k}$. 
For the operator $A_k$ convergence in the gap of the  eigenspaces  is known \cite{MR2220929,MR2974168}, which implies non-pollution of the spectrum. 
However, a convergence analysis of the discretization of $\CA_k$ is beyond the scope of the current~paper.  The main aim of 
our calculations in Section \ref{examples} is to illustrate the general theory in the previous sections and to show how our abstract results apply in concrete~cases. 
Note that, for both finite element methods, convergence theory is known for the operator $A_k$, but the bounds \eqref{June5} only apply to the conforming finite~element~method.

\subsection{Numerical examples.}
\label{examples}

In the last twenty years physicists and engineers have studied dispersive ($\lambda$-dependent) materials in periodic structures extensively \linebreak \cite{PhysRevB.49.11080,PhysRevLett.90.196402,Huang2004,ADMA:ADMA200600106}. In particular, polaritonic materials have received much interest, mainly because they exhibit a strong resonance at infrared frequencies \cite{2040-8986-14-5-055103}. A common model for polaritonic materials is
\begin{equation}\label{eq:polaritonic}
	\epsilon (\lambda):=\epsilon_{\infty}\left (1+\frac{\omega_{\mrm{L}}^2-\omega_{\mrm{T}}^2}{\omega_{\mrm{T}}^2-\lambda} \right ).
\end{equation}
Contrary to ordinary non-dispersive structures, the physical lattice constant $d$ (the physical unit cell is $d\times d$) plays an important role. Let $d=21\mrm{\mu m}$ denote the lattice constant and $\nu$ the speed of light. The material constants of  gallium arsenide (GaAs) in units of $2\pi \nu/a$ then are \cite{2040-8986-14-5-055103},
\begin{equation}\label{eq:GaAs}
   \epsilon_{\infty}=10.9, \quad \omega_{\mrm{L}}=0.612,\quad \omega_{\mrm{T}}=0.568. 
\end{equation} 
In all examples the standard IRA algorithm as provided by ARPACK \cite{MR1621681} was used to compute the eigenvalues. Moreover, in the one pole case residual inverse iteration \cite{MR799120} based on the Rayleigh functionals \eqref{ppm} was used to verify the computations of $\wt{\la}_{1,1}$ and $\wt{\la}_{2,1}$ defined by \eqref{eq:Gal_low}, \eqref{eq:Gal_high}.

\subsection{One pole case.}

The conforming finite element method in Section \ref{conforming} was used to study the polaritonic material model \eqref{eq:polaritonic} when the constants in the material model \eqref{eq:epsModel} are $M=2$, $\wh M=1$, and $L_1=1$.  The material in $\Omega_1$ is polaritonic with constants \eqref{eq:GaAs} and $\Omega_2$ is filled 
\pagebreak   
with air, $\epsilon=1$. Hence, the constants in the material model
\[
	\epsilon (x,\la)=\left (a_1+\frac{b_{1,1}}{\la-c_{1,1}}\right )\chi_{\Omega_1}(x)+a_2\chi_{\Omega_2}(x), \quad x\!\in\! \Omega\!=\!\Omega_1\dot\cup\Omega_2,
\vspace{-2mm}	
\]
are
\begin{equation}
	\aaa_1=\epsilon_{\infty},\quad \aaa_2=1,\quad \bb_{1,1}=\epsilon_{\infty}(\omega_{\mrm{L}}^2-\omega_{\mrm{T}}^2),\quad \cc_{1,1}=\omega_{\mrm{T}}^2
\end{equation}
with $\epsilon_\infty$, $\omega_{\mrm{L}}$, and $\omega_{\mrm{T}}$ as in \eqref{eq:GaAs}.

Table \ref{MinMaxA} shows that the approximations $\wt{\nu}_1$, $\wt{\nu}_2$ to the first two eigenvalues $\nu_1(A_k)$, $\nu_2(A_k)$ of $A_k$  decrease for higher polynomial degrees $p$, which is expected from the classical min-max principle (\ref{eq:A_Min_Max}). The numerical calculations indicate that the convergence is exponential and that the estimate (\ref{eq:estA}) with (\ref{Rate_p}) holds. Moreover, Table \ref{MinMaxA} suggests the same convergence behaviour for the approximations to the lowest eigenvalue $\la_{1,1}$ of the operator matrix $\CA_k$ and to the first eigenvalue $\la_{2,1}$ above $\cc_{1,1}$. This behaviour is expected from the variational principle (\ref{minmax-}), (\ref{GalerkinMinmax}).

\begin{table}[ht]
\caption{\small Polaritonic material model \eqref{eq:polaritonic}. Eigenvalues of $\widetilde{A}_k$ and $\widetilde \CA_k$ when $k\!=\!(\pi,0)$, $\aaa_1\!=\!10.9$, $\aaa_2\!=\!1$, $\bb_{1,1}\!=\!22.3419$, $\cc_{1,1}\!=\!12.7367$. 
The bold numbers show the digits in common with the solutions for $p\!=\!10$.
\vspace{-2mm}}
\begin{center}
\begin{tabular}{c c l l l l}
$p$ & $\wt N$ &  $\wt{\nu}_1$ & $\wt{\nu}_2$ &  $\wt{\la}_{1,1}$ & $\wt{\la}_{2,1}$ \\
\hline\\[-2.5mm]
4	& 320     &  \textbf{3.544}840275   &    \textbf{5.16}8021445 	&	\textbf{1.402}494821   &     \textbf{14.396}21821\\ 
6 	& 720     &  \textbf{3.54459}6578   &    \textbf{5.166}804859 	&   \textbf{1.40223}9773   &     \textbf{14.3960}3122\\ 
8 	& 1280   &  \textbf{3.5445959}50   &    \textbf{5.1667979}31 	&  	\textbf{1.40223895}8   &     \textbf{14.396029}80\\
10 	& 2000   &  \textbf{3.544595948}   &    \textbf{5.166797908}   	&    \textbf{1.402238956}   &     \textbf{14.39602979}  
\end{tabular}
\vspace{-1mm}
\end{center}
\label{MinMaxA}
\end{table}

Table \ref{Tabel:pol} shows the numerical approximations of the bounds  \eqref{lobound-}--\eqref{lobound+}  for a few eigenvalues when $\dim \Hi_{\wt N}=7320$ and $\dim\Hii_{\wt N}=3504$. Recall that  $N(\widetilde A_{k},c)$ denotes the number of eigenvalues of $\widetilde A_{k}$ in $\Hi_{\wt N}$ less than or equal to $c$. In the example we have $N(\widetilde A_k,c_{1,1})=4$. Hence, for the matrix problem the condition $\dim \Hi_{\wt N} -\dim\Hii_{\wt N} \ge N(\widetilde A_{k} ,c_{1,1})$  in Remark \ref{remark:shift} ensures that the shift $\kappa_2$ for the finite dimensional problem is computed correctly. In this example, numerical calculations show that $\widetilde{\CS}(c_{1,1}^{+})$ does not have any negative eigenvalues, which implies that the index shift $\kappa_2$ is zero.
 


Note that the chosen material parameters \eqref{eq:GaAs} give $\|B\|=\sqrt{\bb_{1,1}\cc_{1,1}/\aaa_1}\approx 5.1$, which is relatively large;  
in this case we cannot expect the bounds  \eqref{twosidedestimates} to be~tight.

\begin{table}[h]
\caption{\small Polaritonic material model \eqref{eq:polaritonic}. Bounds on $\wt{\la}_{1,j}$, $\wt{\la}_{2,j}$ for $\Omega_1$ disk of radius $r=0.475$, $k=(\pi,0)$, 
$N(\widetilde A_{k},c_{1,1})=4$, $\aaa_1=10.9$, $\aaa_2=1$, $\cc_{1,1}=12.73668500$, and here $\bb_{1,1}=22.3419$, $\kappa_2=0$.
\vspace{-2mm}}
\begin{center}
\begin{tabular}{c c c c c c c}
$n$ &  $\wt{\la}_{1,j}^{L}$ & $\wt{\la}_{1,j}$ & $\wt{\la}_{1,j}^{U}$ &  $\wt{\la}_{2,j}^{L}$ & $\wt{\la}_{2,j}$ & $\wt{\la}_{2,j}^{U}$\\
\hline\\[-2.5mm]
1	&	1.25       	&   	1.39		&	3.53 		&	$c_{1,1}$		&	14.38	&	15.01\\
2	&	2.57       	&   	3.09     	&     	5.14		&	$c_{1,1}$   	&   	14.82     	&     15.31\\
3	&	5.93       	&   	6.04     	&      9.76		&	$c_{1,1}$     	&   	16.44     	&     16.57\\
4	&	7.00		&	7.30		&	11.55	&	$c_{1,1}$		&	16.96	&	17.29\\
5	&	9.42	 	&	9.56		&	$c_{1,1}$		&	17.31	&	20.48	&	20.63\\
\end{tabular}
\label{Tabel:pol}
\vspace{-1mm}
\end{center}
\end{table}

Table \ref{Tabel:Low} shows the numerical approximations of the bounds  \eqref{twosidedestimates}  for a few eigenvalues when $\bb_{1,1}=1$, which gives a less strong rational term  $\|B\|\approx 1$ compared to the data used to compute the values in Table \ref{Tabel:pol}.  Notice that the bounds in the latter case are much tighter. The shift $\kappa_2=2$ is computed numerically as above.

\begin{table}[h]
\caption{\small Polaritonic material model \eqref{eq:polaritonic}. Bounds on $\wt{\la}_{1,j}$, $\wt{\la}_{2,j}$ for $\Omega_1$ disk of radius $r=0.475$,
$k=(\pi,0)$, $N(\widetilde A_{k},c_{1,1})=4$,  $\aaa_1=10.9$, $\aaa_2=1$, $\cc_{1,1}=12.73668500$, and here $\bb_{1,1}=1$, $\kappa_2=2$.
\vspace{-2mm}}
\begin{center}
\begin{tabular}{c c c c c c c}
$n$ &  $\wt\la_{1,j}^{L}$ & $\wt\la_{1,j}$ & $\wt\la_{1,j}^{U}$ &  $\wt\la_{2,j+2}^{L}$ & $\wt\la_{2,j}$ & $\wt\la_{2,j+2}^{U}$\\
\hline\\[-2.5mm]
1	&	1.60       	&   	1.61		&	1.71 		&	$c_{1,1}$		&		12.96		&	12.97\\
2	&	3.54       	&   	3.57     	&   3.67		&	$c_{1,1}$   		&   	13.05     	&   13.09\\
3	&	7.67       	&   	7.68     	&   7.90		&	15.47    	&   	15.82     	&   15.85\\
4	&	9.47		&		9.51		&	9.83		&	16.42 	&	16.67	&	16.72\\
5	&	12.36		&		12.38		&	$c_{1,1}$		&	20.07 	&	 20.21	&	20.23\\
\end{tabular}
\label{Tabel:Low}
\vspace{-2mm}
\end{center}
\end{table}




\subsection{Multi-pole case.}

For simplicity we consider only the case of two poles, where the constants in the material model \eqref{eq:epsModel} are $M=2$, $\wh M=1$, $L_1=2$, and $\Omega_1$ is a disk of radius $r=0.2$. 
The symmetric interior penalty method is used to discretize the block operator matrix \eqref{eq:A-multi}.  
A few eigenvalues $\wt \la_{1,j}(k)$, $\wt \la_{2,j}(k)$ are numerically calculated using a selection of vectors $k$ along the line segments between the points 
\begin{equation}
	\Gamma=(0,0),\quad X=(\pi,0),\quad M=(\pi,\pi). 
\end{equation}
The triangular path formed by these points is called the boundary of the irreducible Brillouin zone \cite{JJWM2008}. 

\begin{figure}[h]
\begin{center}
\includegraphics[scale=0.76]{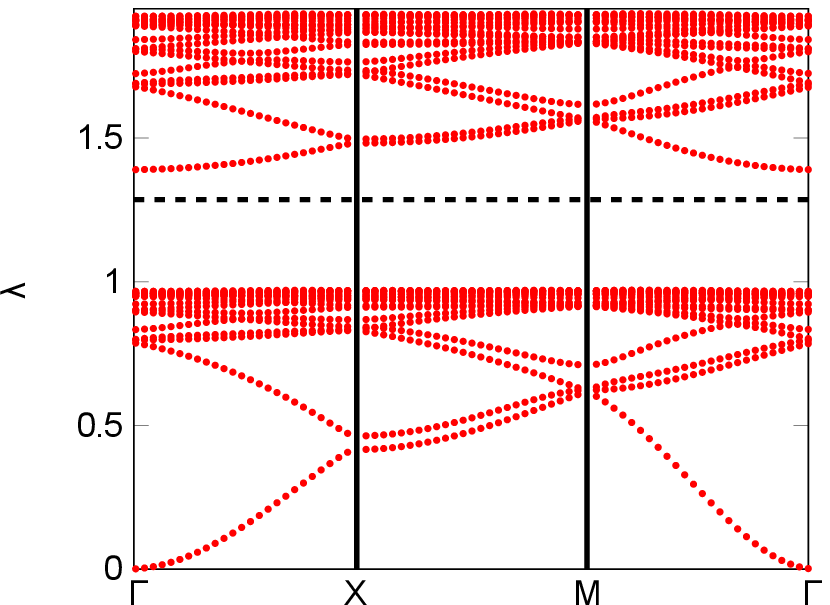}
\includegraphics[scale=0.76]{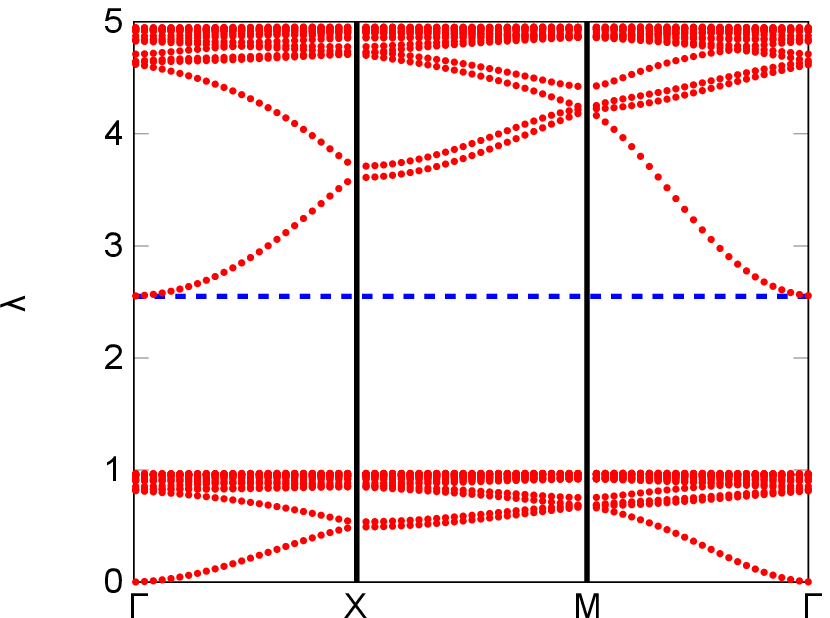}
\caption{\small Lorentz models \eqref{eps-last} with $2$ poles: Band structure with a disk inclusion $\Omega_1$ with $r=0.2$. The dashed line marks the upper limit of (\ref{PC:TwoPoles}). In both cases $\aaa_1=4$, $\aaa_2=1$, $\bb_{1,1}=8$, $\cc_{1,1}=1$; above $\bb_{1,2}=10$, $\cc_{1,2}=2$, below $\bb_{1,2}=2.53$, \vspace{-2mm} $\cc_{1,2}=5$.}
\label{Fig:Band}
\end{center}
\end{figure}

Figure \ref{Fig:Band} shows a 
few eigenvalues for the model $\aaa_1=4$, $\aaa_2=1$, $\bb_{1,1}=10$, $\cc_{1,1}=1$, $\bb_{1,2}=4$, $\cc_{1,2}=2$. The dashed line corresponds to the upper bound in (\ref{PC:TwoPoles}). 
Note that these bounds only require that $A_k>c_{1,2}$ holds. Hence, in this case we can guarantee a band gap using for instance the verified eigenvalue enclosures in \cite{MR2556599} 
to show that $A_k>c_{1,2}$. This problem will, in general, be much less demanding than proving the band gap directly.

Lastly, we consider a case of a permittivity function \eqref{eq:epsModel} where both materials are $\lambda$-dependent
\begin{equation}
\label{eps-last}
	\epsilon (x,\lambda)\!:=\!\left(\aaa_1 \!+\! \frac{\bb_{1,1}}{\cc_{1,1} \!-\! \lambda}\right ) \chi_{\Omega_1}(x)+\left(\aaa_2 \!+\! \frac{\bb_{2,1}}{\cc_{2,1} \!-\! \lambda}\right )\chi_{\Omega_2} (x),
	\quad x\!\in\! \Omega\!=\!\Omega_1\!\dot\cup\Omega_2.
\end{equation}
Here Proposition \ref{prop:4.1} does not apply since $B_{1,1}$ and $B_{2,1}$ do not have the same~range.  

Indeed, the claim of Proposition \ref{prop:4.1} does not seem to hold.
Figure \ref{Fig:BandTwoProjections} shows that numerically we have a $k$-dependent eigenvalue above $c_{1,1}$ that tends to $c_{1,1}$ 
if $k\!\rightarrow\! k_0$ for some values on $k_0\!\in\! (-\pi,\pi)^{2}\!\!$. So, in this case, we cannot expect a band~gap. 

\begin{figure}[hb]
\begin{center}
\includegraphics[scale=0.76]{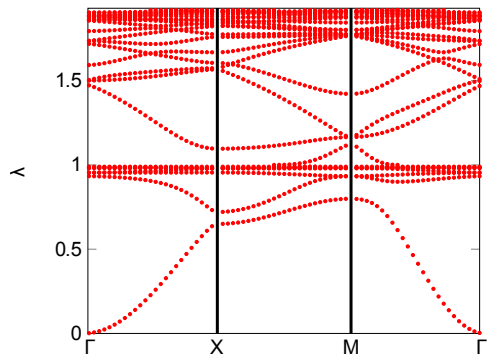}
\caption{\small Lorentz model \eqref{eps-last} with $2$ poles: Band structure with a disk inclusion $\Omega_1$ with $r=0.2$, $\aaa_1=1$, 
$\aaa_2=4$, $\bb_{1,1}=8$, $\cc_{1,1}=1$, $\bb_{2,1}=12$, $\cc_{2,1}=2$. No band gap above \vspace{-2mm} $c_{1,1}=1$.}
\label{Fig:BandTwoProjections}
\end{center}
\end{figure}
